\documentclass[11pt]{article}
\usepackage[margin=1in]{geometry}
\usepackage{amsthm}
\usepackage{amsmath,amssymb}
\usepackage[utf8]{inputenc} 
\usepackage[T1]{fontenc}    
\usepackage{hyperref}       
\usepackage{url}            
\usepackage{booktabs}       
\usepackage{amsfonts}       
\usepackage{nicefrac}       
\usepackage{microtype}      
\usepackage{palatino}
\usepackage{mathpazo}
\usepackage{graphicx}
\usepackage{color}
\usepackage[linesnumbered,ruled,vlined]{algorithm2e}
\SetKwInput{Input}{Input}
\SetKwInput{Output}{Output}
\usepackage{lipsum}

\newtheorem{theorem}{Theorem}
\numberwithin{theorem}{subsection}

\newtheorem{lemma}{Lemma}

\newtheorem{Prop}{Proposition}

\newtheorem{Cor}{Corollary}

\def\<{\langle}
\def\>{\rangle}
\def\be{\begin{equation*}}
\def\ee{\end{equation*}}
\def\bea{\begin{eqnarray*}}
\def\eea{\end{eqnarray*}}

\def\C{\mathbb{C}}

\def\eps{\varepsilon}
\newcommand{\comment}[1]{}

\newcommand{\tr}{\operatorname{Tr}}

\newcommand{\ip}[2]{\langle #1 , #2\rangle}

\newcommand{\wt}[1]{\widetilde{#1}}

\newcommand{\ket}[1]{\ensuremath{\left|#1\right\rangle}}
\newcommand{\op}[2]{|#1\rangle \langle #2|}

\title{Quantum Identity Testing: A Streaming Algorithm and Applications}

\author{
  Nengkun Yu\\
Centre for Quantum Software and Information\\
University of Technology Sydney \\
}

\begin{document}

\maketitle

\begin{abstract}
Do one or more unknown quantum states exhibit a particular property or are they $\epsilon$-far from having that property in $\ell_1$ distance? This is an important question in quantum computing, formulated as a quantum property testing problem. However, unlike classical property testing, where sampling is performed in a fairly standard way, there are several natural choices for the sampling procedure in the quantum setting. The three most pertinent to this paper are the joint measurement, the independent measurement, and the local measurement. The independent measurement approach to quantum property testing is very difficult to implement with current technology but may become easier in the short-to-medium-term future. Joint measurement has a relatively efficient sample complexity, but it is even harder to implement and so unlikely to be realisable in the near future. Local measurement, however, is relatively easy to implement. But perhaps surprisingly, local measurement has not yet been explored as a method of quantum property testing, even for the problem of quantum state tomography.

Hence, one of the main subjects of this paper is to study quantum property testing with local measurement. In particular, we establish a novel $\ell_2$ norm connection between quantum property testing problems and the corresponding distribution testing problems. This connection opens up the potential to derive efficient testing algorithms using techniques developed for classical property testing. As the first demonstration of these possibilities, we designed two streaming algorithms: one for quantum state tomography, and the other for quantum identity testing. Each employs a fixed one-qubit measurement of each qubit regardless of the size of the system. Furthermore, their simplicity means they can be easily implemented with current technology. To the best of our knowledge, no streaming algorithm has yet been used for quantum property testing.

By using the idea of our tomography algorithm, we obtain a streaming algorithm which provide good estimations for each $k$-qubit reduced density matrice of $m$-qubit state using only $\log m$ copies for constant $k$. This is tight and exponential speedup compare with optimal tomography for each $k$-qubit reduced density matrice.

To illustrate the usefulness of our identity testing algorithm, we achieve the following: independence testing for quantum states; identity and independence testing for quantum state collections; and conditional independence testing for classical-quantum-quantum states. Following the widely believed principle ‘Quantum computer, Classical control', these results initialize the property testing problems of classical-quantum states, which could be particularly useful for the analysis of data generated by future quantum computers.
Additionally, using a dimension splitting technique, we derive a matching lower bound up to log factor for independence testing with joint measurement.

\end{abstract}
\section{Introduction}
\subsection{Background and motivation}
The ability to test whether an unknown object satisfies a hypothetical model based on observed data plays a particularly important role in science \cite{LR05}.
Initially proposed by Rubinfeld and Sudan \cite{RS92,RS96}  to test algebraic properties of polynomials, the concept of property testing has been extended to many objects: graphs, Boolean functions, and so on \cite{Goldreich:1998:PTC:285055.285060,GR11}. At the beginning of this century,
Batu $et$ $al.$ introduced the problem of testing properties associated with discrete probability distributions
\cite{BFR+00,Batu:2004:SAT:1007352.1007414}. In other words, how many samples from a collection of probability distributions
are needed to determine whether those distributions satisfy a particular property with high confidence?
Over the past two decades, this area has become an extremely well-studied and successful branch of property testing due in part to the ongoing data science revolution.
Never have computationally-efficient algorithms, a.k.a, testers, that can identify and/or classify properties using as few samples as possible been in higher demand.

A direct approach to distribution property testing is to accurately reconstruct the given distributions from sufficiently many samples.
It is well known
that, after taking $\Theta(d/\eps^2)$ samples from a $d$-dimensional probability distribution $p$,
the empirical distribution is, with high probability, $\eps$-close to $p$ in total variance distance.
Surprisingly, algorithms using less number of samples than $\Theta(d/\eps^2)$, the number of samples to reconstruct the distribution, exist for many important properties. A very incomplete list of works includes
\cite{BFR+00,Batu:2001:TRV:874063.875541,Batu:2002:CAE:509907.510005,Batu:2004:SAT:1007352.1007414,
Pan08,Valiant:2008:TSP:1374376.1374432,pmlr-v19-acharya11a,LRR11,VV11,Valiant:2011:PLE:2082752.2082882,
Indyk:2012:ATK:2213556.2213561,Daskalakis:2013:TKM:2627817.2627948,Valiant:2014:AIP:2706700.2707449,
Chan:2014:OAT:2634074.2634162,NIPS2015_5839,7354450,JVHW15,Valiant:2016:IOL:2897518.2897641,
WY16,7782980,VV17,pmlr-v65-daskalakis17a,
GOL17,diakonikolas_et_al:LIPIcs:2017:7493,Acharya:2017:SDE:3039686.3039769,Canonne2018,Daskalakis:2018:TIM:3174304.3175435,NIPS2018_7858,CDG17},
and two excellent surveys include more
\cite{Rubinfeld:2012:TBP:2331042.2331052,Canonne2015ASO}.
The equality, or identity, of distributions is a central problem in this branch of study, and one that is frequently revisited  with different approaches due to its importance. In 2014, Chan $et$ $al.$, in \cite{Chan:2014:OAT:2634074.2634162}, settled the complexity of the equality of distributions $\Theta(\max(\sqrt{d}/\epsilon^2,d^{2/3}/\epsilon^{4/3}))$.
Further, the ideas and techniques developed in studying the equality of distributions have helped to completely solve the independence testing for discrete distributions by Diakonikolas and Kane in \cite{DK16},
and have yielded an efficient form of conditional independence testing \cite{Canonne:2018:TCI:3188745.3188756}.

The concept of quantum property testing was formally introduced in Montanaro and de Wolf's comprehensive survey \cite{MdW13}.
At this stage of development in the field of quantum computation, testing the properties of new devices as they are built is a basic problem.
A standard quantum device outputs some known $d$-dimensional (mixed) state $\sigma \in \mathcal{D}(\C^{d})$
but inevitably, the results are noisy such that the actual output state $\rho \in \mathcal{D}(\C^{d})$ is not equal $\sigma$, maybe not even close to. Similar to property testing with classical distributions, properties of $\rho$ need to be verified by accessing the device, say, $n$ times, to derive $\rho^{\otimes n}$.

A significant difference between quantum property testing and classical property testing is the way the objects are sampled.
In classical property testing, each sample is output with a classical index according to the probability distribution and
given a fixed number of samples, the output string obeys the product probability distribution. However,
with quantum property testing, the sampling methods have much richer structures.

\begin{table}[h]
 \centering
\begin{tabular}{|c|c|c|c|}
\hline
\textbf{Measurement}&\textbf{Complexity} & \textbf{Dimension}& \textbf{Implementation}\\ \hline
Joint & Low & $d^k$ &Hard even in the future\\ \hline
Independent & Medium & $d$& Hard right now\\ \hline
Local & High & $2$ &Easy right now\\ \hline
\end{tabular}
\caption{Typical Quantum Sampling Methods}
\end{table}
Among the many available sampling methods for quantum property testing (given a fixed number of copies, says $n$, of the states $\rho\in\mathcal{D}(\C^{d})$), the three listed in Table 1 are of particular interests, i.e., joint measurement, independent measurement, and the local measurement. Joint measurement, the most general, allows arbitrary measurements of $\C^{d^n}$. Independent measurement only allows non-adaptive measurements of each copy of $\rho$, which results in, $n$ measurements of $\C^{d}$. And, if $\rho$ is regarded as a $\lceil\log d \rceil$ qubit state, local measurement only allows non-adaptive measurement of each qubit. Unlike independent and local measurement, joint measurement has the potential to provide the optimal number of samples, but there are two caveats. ``Optimal'' joint measurement algorithms usually require an \textit{exponential} number of copies of the quantum state to produce optimal results. They are also based on the assumption of noiseless, universal quantum computation on the \textit{exponential} number of copies of the quantum state. For instance, the optimal tomography algorithms of $k$-qubit quantum state in \cite{HHJ+16,OW16,OW17} require a measurement on $O(\frac{k2^{2k}}{\epsilon^2})$ qubits. Even in the future when quantum computers become a reality, implementing optimal joint measurement would be extremely hard given these conditions. Independent measurements are not feasible with the currently-available technology but should be relatively easy to implement once we have a quantum computer. For instance, performing an independent measurement $[\frac{1}{4}(2I^{\otimes 8}+X^{\otimes 8}+Y^{\otimes 8}),\frac{1}{4}(2I^{\otimes 8}-X^{\otimes 8}-Y^{\otimes 8})]$ on an $8$ qubit quantum system is currently very hard.

Today, most measurements are based on local measurement because, unlike joint and independent measurement, the local measurement method is fashioned after the way streaming algorithms work. To illustrate: for any $n$-qubit state $\rho$ that Alice sends to Bob, Bob only needs to perform a single-qubit measurement of each qubit he receives; the results are then stored as classical outputs. Hence, when $n$ is large, and not all  $n$ qubits can be received simultaneously, local measurement schemes do not require quantum storage whereas a general independent measurement scheme would require quantum storage for $n$ qubit.

A more direct approach to quantum property testing is to estimate $\rho$ by sampling from $\rho^{\otimes n}$, which also means one could check any property of interest. This problem has been solved in three recent papers \cite{HHJ+16,OW16,OW17}. All three show that with joint measurement $\wt{\Theta}(d^2 /\epsilon^2)$ copies of~$\rho$ are sufficient to produce a trace distance error of less than $\eps$ with high probability.
Haah $et$ $al.$ \cite{HHJ+16} proved that, $\wt{\Theta}(d^2 /\epsilon^2)$ copies are necessary. In addition, Haah $et$ $al.$ \cite{HHJ+16} proved that, with independent measurement, $\Omega(d^3/\eps^2)$ copies are needed for a good estimation. Together with the sequence of works \cite{compressed,FlammiaGrossLiuEtAl2012,Vlad13,KRT14}, we now have a complete picture that shows $\Theta(d^3/\eps^2)$ copies are optimal for estimating $\rho$.

However, like classical property testing, this idea is not optimal for a general property. For example, O’Donnell and Wright \cite{OW15} proved that, with joint measurement, $\Theta(d/\epsilon^2)$ copies are necessary and sufficient to test whether a given state is maximally mixed. Yet in the general setting, where one is given access to unknown $d$-dimensional quantum mixed states $\rho$ and $\sigma$, $\Theta(d/\epsilon^2)$ copies are necessary and sufficient to test whether they are identical or $\epsilon$-far apart \cite{BOW17}.

Acharya $et$ $al.$ \cite{JISW17} studied methods for estimating the von Neumann entropy of general quantum states, while Gross $et$ $al.$ \cite{GNW17} showed that ``stabilizerness" can be tested efficiently. In \cite{NIPS2018_8111}, the online learning of quantum states was studied, which allows adaptive measurement of each copy of $\rho$. In \cite{Aaronson:2018:STQ:3188745.3188802,AR19}, the connection between quantum learning and differential privacy was established. Testing properties of distributions using quantum algorithms were studied in \cite{GL19}.

Following the widely believed principle `Quantum computer, Classical control', the fully-fledged quantum computer will be controlled through a classical system. Therefore, the data generated by quantum computers would be modeled by classical-quantum states, e.g., classical collections of quantum states. The importance of classical-quantum states also comes from its central role in studying quantum communication complexity \cite{1238196,Anshu:2017:ESQ:3055399.3055401}. In classical property testing, Levi, Ron, and Rubinfeld initialized the study of property testing of collections of distributions in their pioneering work \cite{LRR11}. Their work would be directly applied to property testing of data generated from several locations in a large sensor-net where the samples of distributions usually come from several locations. This motivates us to study the property testing problems of classical-quantum states.

In this paper, we primarily focus on quantum property testing through local measurement, which is the most friendly approach given the currently-available experimental technology, although we do explore joint and independent measurement by way of comparison. Surprisingly, local measurement has not yet been properly considered as a quantum property testing scheme, not even for the fundamental quantum state tomography problem. We also initialize the study of the property testing problems of classical-quantum states, due to the importance of classical-quantum states.

\subsection{Our contributions}

The local measurement scheme presented in this article maintains an interesting relation between the $\ell_2$ distance
of quantum states and the $\ell_2$ distance of the generated corresponding probability distributions.
Given that $\ell_2$ distance plays a central role in classical property testing, our approach invokes an immediate connection between quantum and classical property testing. This is a conceptual contribution because previous research into quantum property testing has always been in isolation of classical property testing, whereas this scheme opens up the potential to design streaming algorithms for quantum property testing from ingenious ideas and techniques of distribution testing. Further, this is a fixed measurement scheme that does not depend on the property to be tested, which makes our algorithms a perfect fit for implementation with current experiments. Our approach is the first reduction from quantum property testing to distribution testing, this paper includes derivations of efficient quantum streaming algorithms for quantum state tomography, quantum identity testing and independence testing using local measurement, along with classical property testing algorithms. As a comparison, we have also included the results derived from both the independent and joint measurement approaches with our tester and the identity tester in \cite{BOW17}.

The streaming algorithms presented have run-times that are exponential to the
number of qubits, but are entirely \textit{classical} once the sampling complete. An exponential running time is unavoidable because, with quantum state tomography, the length of the output is already exponential in the
number of qubits which entails exponential run time. For most of the problems considered here with a joint measurement approach, a dimension splitting technique demonstrates matching lower bounds that are exponential in the
number of qubits up to a polylog factor in some cases.

In terms of specific formal demonstrations, the first is the following connection between quantum property testing and distribution testing in the
independent measurement setting.
\begin{theorem}\label{indepent-measurement-classical}
There is an independent measurement scheme that maps the quantum states in  $\mathcal{D}(\C^{d})$ into $d(d+1)$-dimensional probability distributions such that, for any pair of quantum states $\rho$ and $\sigma$
\begin{align}
||p-q||_2&=\frac{||\rho-\sigma||_2}{d+1},\\
||p||_2, ||q||_2& \leq\frac{\sqrt{2}}{d+1},
\end{align}
where $p$ and $q$ are the corresponding probability distributions of $\rho$ and $\sigma$, respectively	.
\end{theorem}
By generalizing this result, we demonstrate the following connection between quantum property testing and distribution testing in a
local measurement setting.
\begin{theorem}\label{local-measurement-classical}
For $\rho_{1,2,\dots,m}$ and $\sigma_{1,2,\dots,m}\in \mathcal{D}(\C^{d_1}\otimes \C^{d_2}\otimes\cdots\otimes\C^{d_m})$ of dimension $d_i$ are to the power of $2$. When measuring each local party using measurements that correspond to MUB of $d_i$, the resulting probability distributions
$p_{1,2,\dots,m},q_{1,2,\dots,m}\in \Delta(\times_{i=1}^m[d_i(d_i+1)])$ for any $S\subset[m]$ satisfy that
\begin{align*}
&||p_{1,2,\dots,m}-q_{1,2,\dots,m}||_2=\frac{\sqrt{\sum_{S\subset[m]}||\rho_{S}-\sigma_{S}||_2^2}}{\Pi_{i=1}^m(d_i+1)},\\
&||p_{S}||_2,||q_S||_2\leq\frac{2^{|S|/2}}{\Pi_{i\in S}(d_i+1)}.
\end{align*}	
\end{theorem}
These connections enables us to derive the upper bounds of sample complexity on quantum property testing as follows. For $m$ qubit state, we summarize these bounds and some known bounds together in Table \ref{BTP}.
\begin{table}[h]
 \centering
\begin{tabular}{|c|c|c|c|}
\hline
\textbf{Measurement}&\textbf{Local} & \textbf{Independent}& \textbf{Joint}\\ \hline
Tomography & $O(\frac{18^m}{\epsilon^2})$ & $\Theta(\frac{8^m}{\epsilon^2})$ \cite{HHJ+16} & $\Theta(\frac{4^m}{\epsilon^2})$\cite{HHJ+16,OW16}\\ \hline
Local tomography & $\Theta(\frac{\log m}{\epsilon^2})$ &  & \\ \hline

Identity & $O(\frac{(6\sqrt{2})^m}{\epsilon^2})$ & $O(\frac{4^m}{\epsilon^2})$ & $\Theta(\frac{2^m}{\epsilon^2})$ \cite{BOW17}\\ \hline
Independence & $O(\frac{(6\sqrt{2})^m}{\epsilon^2})$ & $O(\frac{4^m}{\epsilon^2})$ & $\tilde\Theta(\frac{2^m}{\epsilon^2})$\\ \hline
Identity of collection & $O(\frac{(6\sqrt{2})^m}{\epsilon^2})$ & $O(\frac{4^m}{\epsilon^2})$ & $\Theta(\frac{2^m}{\epsilon^2})$\\ \hline
Independence of collection& $O(\frac{(6\sqrt{2})^m}{\epsilon^2})$ & $O(\frac{4^m}{\epsilon^2})$ & $\tilde\Theta(\frac{2^m}{\epsilon^2})$\\ \hline
Conditional independence & & Theorem \ref{cid} & Theorem \ref{cid}  \\ \hline
\end{tabular}
\label{BTP}
\caption{Summary of sample complexity of quanutm property testing.}
\end{table}

For all these properties, optimal joint measurement protocols require to implement noiseless quantum measurements on exponential number of qubits. 

For local measurement, e.g. streaming algorithm, we want to make the following classification because this concept is confusing in the quantum setting.
\begin{itemize}
\item  Pauli measurements can not be performed in the local measurement model. For example, to implement $\{I,X,Y,Z\}^{\otimes 2}$ using local measurements, one can not directly perform
$\{I,X,Y,Z\}$ on each qubit, simply because $X,Y,Z$ are not semi-definite positive. One possible way is to perform $\{I/4,(I+X)/4,(I+Y)/4,(I+Z)/4\}$ on each qubit. Although we can compute the expectation value of the original Pauli measurement through the measurement outcome of the latter one, the statistical fluctuations are now correlated and no explicit method can be used for clarification.

\item The Swap test is not a streaming algorithm, even if we restrict the states $\rho$ and $\sigma$ are both pure. Consider $\rho=\sigma=\op{\psi}{\psi}$ with $\ket{\psi}=\frac{\ket{00}+\ket{11}}{\sqrt{2}}$. The Swap test will output `1' with probability 1, indicating they are equal, using a four-qubit measurement. However, if we are using a ``streaming'' Swap test, in the sense that preforming Swap test on each qubit independently. The Swap test of the first qubit will output `1' with probability $\frac{3}{4}$, so will that of the second qubit. It can not output `1' with probability 1, and therefore, this is not an implementation of the original Swap test.
\end{itemize}
Theorem \ref{local-measurement-classical} lets us derive the first-ever streaming algorithm for quantum state tomography.
\begin{theorem} \label{tomography}
The sample complexity of $m$-qubit quantum state tomography with a streaming algorithm is $O(\frac{18^m}{\epsilon^2})$.
\end{theorem}

The local measurement scheme has a structural advantage: It measures each qubit independently. This could be of great importance by the following observation: Understanding the relation between quantum marginals is considered to be one of the most fundamental
problems in quantum information theory and quantum chemistry because many important physical quantities, such as energy and entropy, only depend on very
small parts of a whole system only–these small parts are the marginal or reduced density matrices.
As a small example of three qubit state $\rho_{1,2,3}$, our local measurement method consists of one-qubit measurements on each qubit.
In this sense, if we want to obtain correlation informaiton about $\rho_{1,2}$, our measurement does not corrupt the correlation information about $\rho_{1,3}$ very much in expectation. Formally, 
the local measurement structure enables us to derive the $k$-local tomography whose goal is to output good estimation of all $k$-qubit reduced density matrices with good precision with high probability.
\begin{theorem} \label{localtomography}
The sample complexity of $k$-local tomography of $m$-qubit quantum state with a streaming algorithm is $$O(\frac{108^k(\log{{{m}\choose{k}}}+k\log 6)}{\epsilon^2})$$. For constant $k$, it is $\Theta(\frac{\log m}{\epsilon^2})$.
\end{theorem}
For constant $k$, this is a exponential speedup comparing with the optimal joint measurement protocol on each $k$-qubit. The reason is that optimal measurement of $k$-qubit would corrupt many other $k$-qubit state. Therefore, the measurement results can not be resued.

As a direct consequence of Theorem \ref{indepent-measurement-classical}, an $\ell_1$-identity tester can be obtained through independent measurement.
\begin{theorem}\label{iit}
For $\rho,\sigma\in\mathcal{D}(\C^d)$, $O(\frac{d^2}{\epsilon^2})$ copies are sufficient to distinguish via independent measurements, with at least a $\frac{2}{3}$ probability of success, the cases where $\rho=\sigma$ from the cases where $||\rho-\sigma||_1> \epsilon$.
\end{theorem}
This is better than directly using the Swap test which uses $O(\frac{d^2}{\epsilon^4})$ copies, although the Swap test is already a joint measurement.

Theorem \ref{local-measurement-classical} yields a streaming algorithm for quantum identity testing.
\begin{theorem}\label{lcoaliit}
For $m$-qubit quantum states $\rho,\sigma$, $O(\frac{(6\sqrt{2})^m}{\epsilon^2})$ copies are sufficient to distinguish, with at least a $\frac{2}{3}$ probability of success,  the cases where $\rho=\sigma$ from the cases where $||\rho-\sigma||_1\geq \epsilon$ using a streaming algorithm.
\end{theorem}

The above $\ell_1$ identity testers for independent and local measurement together with the $\ell_1$ identity tester of \cite{BOW17} for joint measurement let us derive algorithms for multipartite independence testing.
\begin{theorem}\label{thm2}
The sample complexity of independent testing for $\mathcal{D}(\C^{d_1}\otimes \C^{d_2}\otimes\cdots\otimes\C^{d_m})$, i.e., distinguishing, with at least a $\frac{2}{3}$ probability of success, the cases where $\rho$ is in the tensor product form $\rho_1\otimes\rho_2\otimes\cdots\otimes\rho_m$, or $||\rho-\sigma||_1>\epsilon$ for any tensor product $\sigma$ is:
\begin{itemize}
\item $\Tilde{\Theta}(\frac{\Pi_{i=1}^m d_i}{\epsilon^2})$ with joint measurement;
\item $O(\frac{\Pi_{i=1}^m d_i^2}{\epsilon^2})$ with independent measurement; and
\item $O(\frac{\Pi_{i=1}^m d_i^{1.5+\log 3}}{\epsilon^2})$ with a streaming algorithm where all $d_i$ are to the power of $2$.
\end{itemize}
\end{theorem}
With joint measurement, the lower bound of the quantum independence is derived using the dimension splitting technique and a reduction from determining whether a given state is a maximally mixed state.

Now consider a quantum setting in which one receives data that is most naturally thought of as samples of several quantum states--for example, when studying quantum systems in several geographic locations. Such data could also be generated when samples of the distributions come from various quantum sensors that are each part of a large quantum sensor-net.
Another motivation of studying the property testing of collections of quantum states is the quantum state preparation. Suppose there are different ways of generating a quantum state. We want to know whether these methods all work well. This problem can be formulated as property testing of collections of quantum states.

Motivated by \cite{LRR11} on property testing of collections of distributions, we generalize the identity tester and independence tester from quantum states into collections of quantum states in the query model using the framework in \cite{DK16}.
The result of the identity test for collections of quantum states is as follows.
\begin{theorem}\label{cidentical}
Given access to the quantum states $\rho_1$, . . . , $\rho_n$ on $\mathcal{D}(\C^{d})$ and an explicit $c_i>0$ with $C_1\geq\sum_{i} c_i\geq C_0>0$ where $C_0,C_1$ are absolute constants, the sample complexity of distinguishing, with at least a $\frac{2}{3}$ probability of success, the cases where all $\rho_i$ are identical from the cases where $\sum_i c_i ||\rho_i-\sigma||_1> \epsilon$ for any $\sigma$ is
\begin{itemize}
\item  $\Theta(\frac{d}{\epsilon^2})$ with joint measurement;
\item $O(\frac{d^2}{\epsilon^2})$ with independent measurement; and
\item $O(\frac{d^{1.5+\log 3}}{\epsilon^2})$ with a streaming algorithm where all $d_i$ are to the power of $2$.
\end{itemize}
\end{theorem}
Our result of independence tester about collections of quantum states follows.
\begin{theorem}\label{cindependent}
Given sample access to quantum states $\rho_1$, . . . , $\rho_n$ in $\mathcal{D}(\C^{d_1}\otimes \C^{d_2}\otimes\cdots\otimes\C^{d_m})$ with $d_1\geq d_2\geq\cdots\geq d_m$ and explicit $c_i>0$ with  $C_1\geq\sum_{i} c_i\geq C_0>0$ where $C_0,C_1$ are absolute constants, the sample complexity of distinguishing, with at least a $\frac{2}{3}$ probability of success, the cases where all $\rho_i$ are $m$-partite independent from the cases where $\sum_i c_i ||\rho_i-\otimes_{k=1}^m\sigma_{k,i}||_1> \epsilon$ for any $\sigma_{k,i}\in\mathcal{D}(\C^{d_k})$ is
\begin{itemize}
\item $\Tilde{\Theta}(\frac{\Pi_{i=1}^m d_i}{\epsilon^2})$ with joint measurement;
\item $O(\frac{\Pi_{i=1}^md_i^2}{\epsilon^2})$ with independent measurement; and
\item $O(\frac{d^{1.5+\log 3}}{\epsilon^2})$ with a streaming algorithm where all $d_i$ are to the power of $2$.
 \end{itemize}
\end{theorem}
Like their classical counterparts, the complexity does not depend on the number of states.

In further work, we explore the problem of testing conditional independence with classical-quantum-quantum states. This question naturally arises in studying distributed quantum computing. One typical example is environment assisted entanglement distribution. Suppose $\rho_{ABC}$ is a tripartite state. We want to reach the goal of sharing a bipartite state $\sigma_{AB}$. $C$ should perform a measurement on its system, now the state becomes classical-quantum-quantum.

This problem is a generalization of the independent testing of collections of quantum states in the sense that the prior coefficient of the $\ell_1$ distance is not given explicitly but may be approximated through sampling. One motivation for studying this problem is a simplified version of the conditional independence of general tripartite quantum states, which a fundamental concept in theoretical physics and quantum information theory.

More specifically, we modify the $\ell_2$ estimator developed in \cite{BOW17} for joint measurement and develop a finer
$\ell_2$ estimator for independent measurement. Then we plug that estimator into the classical conditional independence testing framework developed
in \cite{Canonne:2018:TCI:3188745.3188756}.
\begin{theorem}\label{cid}
Given classical-quantum-quantum state
$\rho_{ABC}\in\mathcal{D}(\C^{d_1}\otimes\C^{d_2})\otimes\Delta(C)$ with $\Delta(C)$ being the probabilistic simplex of $C$ and $n=|C|$, the sample complexity of testing whether $A$ and $B$ are conditionally independent given $C$ is
\begin{itemize}
\item $O(\max\{\frac{\sqrt{n}d_1d_2}{\epsilon^2},\min\{\frac{d_1^{\frac{4}{7}}d_2^{\frac{4}{7}}
n^{\frac{6}{7}}}{\epsilon^{\frac{8}{7}}},\frac{\sqrt{d_1d_2}n^{\frac{7}{8}}}{\epsilon}\}\})$
with joint measurement; and
\item $O(\max\{\frac{\sqrt{n}d_1^2d_2^2}{\epsilon^2},\min\{\frac{d_1^{\frac{6}{7}}d_2^{\frac{6}{7}}n^{\frac{6}{7}}}{\epsilon^{\frac{8}{7}}},\frac{d_1^{\frac{3}{4}}d_2^{\frac{3}{4}}n^{\frac{7}{8}}}{\epsilon}\}\})$ with independent measurement.
\end{itemize}
\end{theorem}

\subsection{Organization of this paper}
In Section \ref{PT}, we review prior work on quantum tomography and spectrum estimation, as well as some relevant work on testing properties associated with discrete distributions. Section \ref{PR} recalls the basic definitions of distance with discrete distributions and quantum states and presents some formal tools from earlier work that are used here. In Section \ref{LT}, we derive technical lemmata about the independence and conditional independence of quantum states. Section \ref{connection} demonstrates Theorems \ref{indepent-measurement-classical} and \ref{local-measurement-classical}. Section \ref{SAQ} proves Theorem \ref{tomography}. Section \ref{QSC} contains the results of identity testing and, in particular, our proofs of Theorems \ref{iit} and \ref{lcoaliit}. The proof of Theorem \ref{thm2} is given in Section \ref{TI}. The results of the identity and independence testing of collections of states are provided in Section \ref{TPC}, along with the proofs of Theorems \ref{cidentical} and \ref{cindependent}. The paper concludes with the proof of Theorem \ref{cid} in Section \ref{CIT}.

\section{Prior work on property testing}
\label{PT}The purpose of this section is to review some results on learning and testing unknown quantum states, as well as the corresponding classical problem of learning and testing discrete distributions. We are not able to provide a complete description of all literature simply because this is a very active research line. However, what we have done is provide an overview of some of the best known and most recent results.

\subsection{Quantum state tomography}

Starting from the very basic problem of quantum state tomography, a fundamental problem is to decide how many copies of an unknown mixed quantum state $\rho\in\mathcal{D}(\C^d)$ is necessary and sufficient to output a good approximation of ${\rho}$ with high probability. This problem has been studied extensively since the birth of quantum information theory. The main-stream approach is an independent measurement. A sequence of work \cite{compressed,FlammiaGrossLiuEtAl2012,Vlad13,KRT14} is dedicated to showing that $O(d^3/\eps^2)$ copies are sufficient in an $\ell_1$ distance of no more than $\epsilon$. Haah $et$ $al$. \cite{HHJ+16} showed that $\Theta(d^3/\eps^2)$ is the sample complexity for independent measurement. In the same paper, we proved that $O(\frac{d^2}{\delta}\log(\frac{d}{\delta}))$ copies are sufficient to obtain an infidelity of no more than $\delta$ with a probability of $1-\exp(-\Theta(d^2))$ with joint measurement, which can be regarded as a quantum generalization of Sanov's theorem \cite{SANOV57}. Further, by combining the lower bound of \cite{HHJ+16} and upper bound of \cite{OW16}, the sample complexity of state tomography with joint measurement is $\Theta(\frac{d^2}{\epsilon^2})$ in an $\ell_1$ distance of no more than $\epsilon$.
The work of  \cite{FlammiaGrossLiuEtAl2012} includes in footnote 2 an algorithm (which also appears in the section on quantum process tomography from Nielson and Chuang \cite{Nielsen:2011:QCQ:1972505}) which uses $O(\frac{d^4 \log d}{\epsilon^2})$ copies and only uses Pauli measurements.
\subsection{Quantum identity testing}

Naturally, studying the sample complexity of specific problems in different measurement settings is attractive--for instance, quantum state identification, spectrum estimation, entanglement testing. To do this, easily implementable algorithms that use fewer copies of quantum state than tomography is desirable.
One problem that has received much attention is quantum state identification. Suppose we are given query access to two states $\rho,\sigma\in\mathcal{D}(\C^d)$, and we want to test whether they are equal or have a large $\ell_1$ distance. For practical purposes, the results from cases where $\sigma$ is a known pure state have been extensively studied, in independent measurement setting, and completely characterized \cite{PhysRevLett.106.230501,PhysRevLett.107.210404,AGKE15}, show that $O(\frac{d}{\epsilon^2})$ copies are sufficient.
\cite{OW15} solved the problem, in joint measurement setting, where $\sigma$ is a maximally mixed state case by showing that $\Theta(\frac{d}{\epsilon^2})$ copies are necessary and sufficient. Importantly, the sample complexity of the general problem was proven to be $\Theta(\frac{d}{\epsilon^2})$ in \cite{BOW17} by providing an efficient $\ell_2$ distance estimator between two unknown quantum states.
\subsection{Independence testing}

A multipartite pure state $\ket{\psi}\in(\C^{d})^{\otimes m}$ is called a product state if it can be written as $\ket{\psi}=\otimes_{j=1}^m\ket{\psi_j}$ for some $\ket{\psi_j}\in\C^d$. This type of pure state is called entangled if it is not a product. Entanglement is a ubiquitous phenomenon in quantum information theory. Pure entanglement testing was first discussed by Mintert $et.al$ \cite{PhysRevLett.95.260502}. Harrow and Montanaro \cite{Harrow:2013:TPS:2432622.2432625} subsequently proved that $O(\frac{1}{\epsilon^2})$ copies are sufficient.
A mixed state $\rho\in\mathcal{D}((\C^{d})^{\otimes m})$ is called entangled if it can not be written as a convex combination of the density matrices of product states. That said, the problem of entanglement testing for the general mixed state has barely been touched, even in the simplest case of independence testing the tensor product of a general mixed state.
\subsection{Classical identity testing and independence testing}It is a widely known that $\Theta(\frac{d}{\epsilon^2})$ samples are sufficient and necessary to provide a good approximation of $d$-dimensional probability distribution \cite[pages 10 and 31]{DL01}. Therefore, finding algorithms that use $o(d)$ samples for testing problems is highly desirable.
\subsubsection{Identity testing}

\textbf{Uniformity testing:} In an important work \cite{GR11}, Goldreich and Ron found that the $\ell_2$ norm can be estimated from $O(\frac{\sqrt{d}}{\epsilon^2})$ samples. This led to an algorithm for uniformity testing, $i.e.$, to determine whether a probability is a uniform distribution using $O(\frac{\sqrt{d}}{\epsilon^4})$ samples. Here, both Paninski \cite{Pan08} and Valiant and Valiant \cite{VV11} showed that the complexity is $\Theta(\frac{\sqrt{d}}{\epsilon^2})$.\\
\textbf{Identity testing to a known distribution:} In their seminal work, Batu $et al.$ \cite{BFR+00,Batu:2001:TRV:874063.875541} presented an $\ell_2$-identity tester and used it to build an $\ell_1$ estimator using $O(\frac{\sqrt{d}\log d}{\epsilon^2})$ samples in cases where one distribution is known. In \cite{Valiant:2014:AIP:2706700.2707449}, Valiant and Valiant improved this result to show the sample complexity is $\Theta(\frac{\sqrt{d}}{\epsilon^2})$.\\
\textbf{Identity testing between unknown distributions:} Identity testing between unknown distributions was studied by Batu $et.al.$ \cite{BFR+00} in which they provided a tester using $O(\frac{d^{2/3}\log d}{\epsilon^{8/3 }})$ samples. Chan $et$ $al.$ \cite{Chan:2014:OAT:2634074.2634162} showed the complexity of this problem is $\Theta(\max\{\frac{d^{2/3}}{\epsilon^{4/3}},\frac{\sqrt{d}}{\epsilon^2}\})$. Diakonikolas and Kane  demonstrated a unified approach based on a standard $\ell_2$ tester to resolve the sample complexity of a wide variety of testing problems, including an alternative proof for identity testing in their important work \cite{DK16}.\\\subsubsection{Independence testing}In \cite{Batu:2001:TRV:874063.875541}, Batu $et.al$ presented an independence tester for bipartite independence testing over $[d_1]\times[d_2]$ with a sample complexity of $\tilde{O}(d_1^{2/3}d_2^{1/3})\cdot \mathrm{Poly}(\frac{1}{\epsilon})$, for $d_1\geq d_2$. Levi, Ron and Rubinfeld in \cite{LRR11} showed a lower bound for the sample complexity of $\Omega(\sqrt{d_1d_2})$ for all $d_1\geq d_2$ and $\Omega(d_1^{2/3}d_2^{1/3})$ for $d_1=\Omega(d_2\log d_2)$, while Acharya $et$ $al.$ \cite{NIPS2015_5839} introduced a tester for multipartite independence testing over $\times_{j=1}^m[d_j]$ with sample complexity $O(\frac{\sqrt{\Pi_{j=1}^m d_j}+\sum_{j=1}^m d_j}{\epsilon^2})$ . In their seminal work \cite{DK16}, Diakonikolas and Kane resolved this problem by showing that the sample complexity is $\Theta(\max_k\{\frac{\sqrt{\Pi_{j=1}^m d_j}}{\epsilon^2},\frac{d_k^{1/3}\Pi_{j=1}^m d_j^{1/3}}{\epsilon^{4/3}}\})$. These works motivate our Theorem \ref{thm2}.

\subsubsection{Equivalence testing for collections of discrete distributions}
Levi, Ron and Rubinfeld \cite{LRR11} initialized the study of equivalence testing for collections of discrete distributions in query model. The complexity of this problem is proved to be $\Theta(\max\{\frac{\sqrt{n}}{\epsilon^2},\frac{n^{2/3}}{\epsilon^{4/3}}\})$ by Diakonikolas and Kane in \cite{DK16}. We use these techniques to conduct Theorems \ref{cidentical} and \ref{cindependent}.

\subsubsection{Conditional independence testing} Canonne $et$ $al.$ \cite{Canonne:2018:TCI:3188745.3188756} initiated the study of conditional independence within property testing framework. Given a probability distribution over $[d_1]\times[d_2]\times[n]$ with $d_1\geq d_2$, they developed a tester using $O(\max\{\min\{\frac{n^{7/8}d_1^{1/4}d_2^{1/4}}{\epsilon},\frac{n^{6/7}d_1^{2/7}d_2^{2/7}}{\epsilon^{8/7}}\},\frac{n^{3/4}d_1^{1/2}d_2^{1/2}}{\epsilon} , \frac{n^{2/3}d_1^{2/3}d_2^{1/3}}{\epsilon^{4/3}},\frac{n^{1/2}d_1^{1/2}d_2^{1/2}}{\epsilon^2}\})$ samples, which was proven to be optimal for constant $d_1,d_2$. This framework was used to derive our result for Theorem \ref{cindependent}.

\section{Preliminaries}\label{PR}

This section begins with some standard notations and definitions used throughout the paper.

\subsection{Basic facts for probability distributions}
For $m \in\mathbb{N}$, $[m]$ denotes the set $\{1,\cdots,m\}$, and log denotes the binary logarithm.

A probability distribution over discrete domain ${\Omega}$ is a function
$p:\Omega\mapsto  [0, 1]$
such that $$\sum_{\omega\in\Omega}p(\omega)=1.$$
$|\Omega|$ is the cardinality of set $\Omega$.

$\Delta(\Omega)$ denotes the set of probability distributions over $\Omega$, i.e., the probability simplex of $\Omega$.

The marginal distributions $p_1\in\Delta(A)$ and $p_2\in\Delta(B)$ of a bipartite distribution $p_{1,2}\in\Delta(A\times B)$ can be defined as
\begin{align*}
p_1(a)=\sum_{b\in B} p_{1,2}(a,b),\\
p_{1}(b)=\sum_{a\in A}p_{1,2}(a,b).
\end{align*}
The product distribution $q_1\otimes q_2$ of distributions $q_1\in\Delta(A)$ and $q_2\in\Delta(B)$ can be defined as
$$
[q_1\otimes q_2](a,b)=q_1(a)q_2(b),
$$
for every $(a,b)\in A\times B$.

The $\ell_1$ distance between two distributions $p,q\in\Delta(\Omega)$ is
$$
||p-q||_1=\sum_{\omega\in\Omega}|p(\omega)-q(\omega)|,
$$
and their $\ell_2$ distance is
$$
||p-q||_2=\sqrt{\sum_{\omega\in\Omega}(p(\omega)-q(\omega))^2}.
$$
The Possion distribution $\mathrm{Poi}(\lambda)$ for $\lambda>0$ is
$$
p(k)=e^{-\lambda}\frac{\lambda^k}{k!}
$$
where $k\in\mathbb{N}\cup\{0\}$.

We have employed the standard ``Poissonization" approach, which is widely used in property testing for classical distributions. Namely, we have assumed that, rather than drawing $k$ independent samples for a fixed $k$,
the first step is to select $k'$ according to $\mathrm{Poi}(k)$, and then draw $k'$ samples.
This technique means the number of times the different elements occur becomes independent, which significantly
simplifies the analysis.
Note that if $\mathrm{Poi}(k)$ were tightly concentrated about $k$, this trick would
only lose sub constant factors in terms of sample complexity.

\subsection{Basic quantum mechanics}
An isolated physical system is associated with a
Hilbert space, which is called the {\it state space}. A {\it pure state} of a
quantum system is a normalized vector in its state space, and a
{\it mixed state} is represented by a density operator on the state
space. Here, a density operator $\rho$ on $d$-dimensional Hilbert space $\C^d$ is a
semi-definite positive linear operator such that $\tr(\rho)=1$.
We let
$$
\mathcal{D}({\C^{d}})=\{\rho: \rho~\mathrm{is}~d\mathrm{-dimensional~density~operator~of~}\C^d\}
$$
denote the set of quantum states.

The Pauli matrices in one-qubit are given as
$$I_2=\begin{bmatrix}1 &0\\0&1\end{bmatrix}, X=\begin{bmatrix}0 &1\\1&0\end{bmatrix},  Z=\begin{bmatrix}1 &0\\0&-1\end{bmatrix}, Y=\begin{bmatrix}0 &i\\-i&0\end{bmatrix}.$$
$I_{d}$ denotes the identity operator of $\mathcal{D}(\C^{d})$, and $\frac{I_{d}}{d}$ denotes the maximally mixed states of $\mathcal{D}(\C^{d})$.

\subsection{The tensor product of Hilbert space}

The state space of a composed quantum system is the tensor product of the state spaces of its component systems.
Let $\C^{d_k}$ be a Hilbert space, and $\{\ket{\varphi_{i_k}}\}$  be the corresponding orthonormal basis for $1\leq k\leq m$.
One can define a Hilbert space $\bigotimes_{k=1}^{m}\C^{d_k}$ as the tensor product of Hilbert spaces $\C^{d_k}$ by defining its orthonormal basis as
 $\{|\varphi_{i_1}\rangle\cdots|\varphi_{i_m}\rangle\}$ where the tensor product of vectors is defined by the following new
vector
$$\bigotimes_{k=1}^m\left(\sum_{i_k} \lambda_{i_k} |\psi_{i_k}\>\right)=\sum_{i_1,\cdots,i_m} \lambda_{i_1}\cdots\lambda_{i_m}
|\psi_{i_1}\>\otimes\cdots\otimes |\psi_{i_m}\>.$$
Thus, $\bigotimes_{k=1}^{m}\C^{d_k}$ is also a
Hilbert space where the inner product is defined as follows:
For any $|\psi_{k}\>,|\phi_{k}\>\in\C^{d_k}$
$$\<\psi_1\otimes\cdots\otimes \psi_m|\phi_1\otimes\cdots\otimes\phi_m\>=\<\psi_1|\phi_1\>\cdots\<
\psi_m|\phi_n\>.$$
The following notation denotes the quantum state on multipartite system $\bigotimes_{k=1}^{m}\C^{d_k}$,
\begin{align*}
\mathcal{D}(\otimes_{i=1}^m\C^{d_i})=\mathcal{D}(\C^{\Pi_{i=1}^m d_i}).
\end{align*}

The reduced quantum state of a bipartite mixed state $\rho_{1,2}\in\mathcal{D}(\C^{d_1}\otimes\C^{d_2})$ on the second system
is the
density operators $\rho_2:=\tr_{1}\rho_{1,2}=\sum_i \<i|A|i\>$, where $\{|i\>\}$ is the orthonormal
basis of $\C^{d_1}$. The partial trace of $\rho_1:=\tr_{2}\rho_{1,2}$ can be similarly defined, note that
partial trace functions are also
independent of the selected orthonormal basis.
This definition can be directly generalized into multipartite quantum states.

\subsection{Quantum measurement}

A positive-operator valued measure (POVM) is a measure whose values are non-negative self-adjoint operators in a Hilbert space $\C^{d}$, which is
described by a collection of matrices $\{M_i\}$ with $M_i\geq 0$ and
$$\sum_{i}M_i=I_{d}.$$
If the state of a quantum system was $\rho$
immediately before measurement $\{M_i\}$ was performed on it, then
the probability of that result $i$ recurring is
$$p(i)=\tr(M_i\rho).$$
A Hermitian $O=\sum_j \lambda_j \op{\psi_j}{\psi_j}$ (that is $O=O^{\dag}$) on $\mathcal{H}$ with orthonormal basis $\ket{\psi_j}$ and $\lambda_j\in\mathbb{R}$ always corresponds to the following measurement protocol.
Suppose the state of a quantum system is $\rho$, when measured in basis $\op{\psi_j}{\psi_j}$. If the outcome is $j$, the observed result is $\lambda_j$.
Directly, the output corresponds to a random variable $X$ such that
$$
p(X=\lambda_j)=\tr(\rho\op{\psi_j}{\psi_j}).
$$
Thus,
\begin{align*}
&\mathbb{E}~X=\tr(\rho O)\\
&\mathrm{Var} ~X=\tr(\rho O^2)-\tr^2(\rho O).
\end{align*}
Now suppose we are given $n$ copies of a quantum state $\rho\in\mathcal{D}(\C^{d_1}\otimes\C^{d_2}\otimes\cdots\otimes\C^{d_m})$. The most general way of obtaining information is to perform a measurement $\mathcal{M}=\{M_j\}$ of $\rho^{\otimes n}$,
\begin{align*}
\sum_j M_j=I_{\Pi_{i=1}^m d_i^n}.
\end{align*}
An $M_i$ with the dimension $\Pi_{j=1}^m d_j^n$ is called a joint measurement, which usually does not in tensor product form, not even within the convex cone of the tensor product of $n$ semi-definite positive matrices. Although joint measurement can usually provide an optimal learning, implementing this scheme is usually hard.
For example, the optimal tomography protocols and state certification protocols
given in \cite{OW15,HHJ+16,OW16,OW17,BOW17} for $n$-qubit system require noiseless measurement on an \textit{exponential} number of copies of the quantum state.

A more restricted set of measurements is the set of independent measurement for elements of the form $\mathcal{M}=\otimes_{k=1}^n\mathcal{M}^{(k)}$, where $\mathcal{M}^{(k)}=\{M_{j}^{(k)}\}$ is the measurement of $\rho$,
\begin{align*}
\sum_j M_j^{(k)}=I_{\Pi_{i=1}^m d_i}.
\end{align*}
Compared to joint measurement, independent measurement is much easier to implement, although the information gain is usually not as efficient as an joint measurement for a fixed $n$.

A proper subset of independent measurement is the set of local measurement, whose elements are in the form $\mathcal{M}=\otimes_{k=1}^n\otimes_{t=1}^m\mathcal{M}^{(k,t)}$, where $\mathcal{M}^{(k,t)}=\{M_{j}^{(k,t)}\}$ is the measurement of $\C^{d_t}$,
\begin{align*}
\sum_j M_j^{(k,t)}=I_{d_t}.
\end{align*}
In other words, local measurement $\mathcal{M}^{(k,t)}=\{M_{j}^{(k,t)}\}$ is performed on a local system $\C^{d_t}$. The main advantage of this measurement scheme is that it is easy to implement, even easier than independent measurement. Usually, we choose $d_t=2$.

\subsection{$\ell_1$ distance}
$\ell_1$ distance is used to characterize the difference between quantum states. The $\ell_1$ distance between $\rho$ and $\sigma$ is defined as
\begin{equation*}
||\rho-\sigma||_1\equiv\mathrm{Tr}|\rho-\sigma|
\end{equation*}
where $|A|\equiv\sqrt{A^\dag A}$ is the positive square root of $A^\dag A$.

Given a general operator $A$, the $\ell_1$ norm is defined as
$$
||A||_1=\mathrm{Tr}|A|.
$$
And Lemma \ref{pt1norm} always applies:
\begin{lemma}\cite{Nielsen:2011:QCQ:1972505}\label{pt1norm}
The $\ell_1$ distance is decreasing under partial trace. That is
$$||\rho_1-\sigma_1||_1,||\rho_2-\sigma_2||_1\leq ||\rho_{1,2}-\sigma_{1,2}||_1.$$
\end{lemma}
Their $\ell_2$ distance is defined as
$$
||\rho-\sigma||_2=\sqrt{\mathrm{Tr}(\rho-\sigma)^2}.
$$
For $\rho,\sigma\in\mathcal{D}(\C^{d})$, we have the following relation between $\ell_1$ and $\ell_2$ distances,
$$
||\rho-\sigma||_2\leq ||\rho-\sigma||_1\leq \sqrt{d}||\rho-\sigma||_2.
$$
Given a subset $\mathcal{P}\subsetneq \mathcal{D}(\C^{d})$, the $\ell_1$ distance between $\rho$ and $\mathcal{P}$ is defined as
$$
||\rho-\mathcal{P}||_1=\inf_{\sigma\in\mathcal{P}}||\rho-\sigma||_1.
$$
If $||\rho-\mathcal{P}||_1>\epsilon$, we say that $\rho$ is $\epsilon$-far from $\mathcal{P}$; otherwise, it is $\epsilon$-close.

\subsection{Classical quantum state}
For classical system $C$ and quantum system $B$, the following set of state is usually called classical-quantum state.
$$
\mathcal{T}_{BC}=\{\sum_{i\in C}p_c \rho^{c}_{B}\otimes\op{c}{c}: p_c\geq 0,~\sum_{c\in C}p_c=1,~\rho^{c}_{B}\in\mathcal{D}(\C^{d_1})\},
$$
where $\ket{c}$ are fixed orthonormal bases of system $C$.
\subsection{Mutually unbiased bases}

In quantum information theory, mutually unbiased bases (MUB) in $d$-dimensional Hilbert space are two orthonormal bases $\{|e_1\rangle, \dots, |e_d\rangle\}$ and $\{|f_1\rangle, \dots, |f_d\rangle\}$ such that the square of the magnitude of the inner product between any basis states $|e_j\rangle$ and $|f_k\rangle$ equals the inverse of the dimension $d$:
$$
    |\langle e_j|f_k \rangle|^2 = \frac{1}{d}, \quad \forall j,k \in \{1, \dots, d\}.
$$
These bases are unbiased in the following sense: if a system is prepared in a state belonging to one of the bases, then all outcomes of the measurement with respect to the other basis will occur with equal probability. It is known that, for $d=p^n$ with prime $p$, there exists $d+1$ MUBs \cite{DEBZ10}.

Note that there are numerous papers that analyze the performance of MUB POVMs, for instance \cite{Roy_2007}.

\subsection{Quantum property testing}
Let $\mathcal{D}(\C^d)$ denote the set of mixed states in Hilbert space $\C^d$, and let a known $\mathcal{T}\subset\mathcal{D}(\C^d)$ be the working domain of the quantum states.
In a standard of property testing scenario, a testing algorithm for a property
$\mathcal{P}\subset \mathcal{T}$ would be an algorithm that, when granted access to independent samples from an unknown quantum state
$\rho\in \mathcal{T}$ as well as an $\ell_1$ distance parameter of $0<\epsilon\leq 1$, outputs either "Yes" or "No", with the following
guarantees:
\begin{itemize}
\item If $\rho\in\mathcal{P}$, then it outputs "Yes" with a probability of at least $\frac{2}{3}$.
\item If $\rho$ is $\epsilon$-far from $\mathcal{P}$, then it outputs "No" with a probability of at least $\frac{2}{3}$.
\end{itemize}
Our interest is in designing computational efficient algorithms with the smallest sample complexity (i.e., the smallest number of samples drawn of $\rho$.).

Confidence of $\frac{2}{3}$ is not essential here, it could be replaced by any constant greater than $\frac{1}{2}$. This would only change the sample complexity by a multiplicative constant. According to the Chernoff bound, the probability of success becomes $1-2^{-\Omega(k)}$, after repeating the algorithm $k$ times.

\subsection{Tools from earlier work}
The following results were established in earlier work, and are used within this paper.

\begin{theorem}
\cite{BBRV02}\label{MUBs}
The Pauli group $\mathcal{P}_k=\{I,X,Y,Z\}^{\otimes k}$ of order $4^k$ can be divided into $2^k+1$ Abelian subgroups with an order of $2^k$, say, $G_0,\dots,G_{2^k}$ such that $G_i\bigcap G_j=\{I_2^{\otimes k}\}$ for $i\neq j$. Each subgroup can be simultaneously diagonalizable by a corresponding basis. All these $2^k+1$ bases form $2^k+1$ MUBs.
\end{theorem}

\begin{theorem}\cite{OW15} \label{mixness}
$100\frac{d}{\epsilon^2}$ copies are sufficient and $0.15\frac{d}{\epsilon^2}$ copies are necessary to test whether $\rho\in\mathcal{D}(\C^d)$ is the
maximally mixed state $\frac{I_d}{d}$ or $||\rho-\frac{I_d}{d}||_1>\epsilon$ with at least a $2/3$ probability of success.
\end{theorem}

\begin{algorithm}[H]\label{mix-OW}
	\Input{$100\frac{d}{\epsilon^2}$ copies of $\rho\in \mathcal{D}(\C^d)$}	.
    \Output{"Yes" with a probability of at least $\frac{2}{3}$ if $\rho=\frac{I_d}{d}$; and "No" with a probability of at least $\frac{2}{3}$ if $||\rho-\frac{I_d}{d}||_1>\epsilon$.}
	\caption{\textsf{A Mixness Test}}
	\label{algo:Mixness-Test}
\end{algorithm}

\begin{theorem}\cite{BOW17}\label{BOW}
For the mixed states $\rho,\sigma\in \mathcal{D}(\C^d)$ and $\epsilon> 0$, there is an algorithm that, given
$\Theta(\frac{d}{\epsilon^2})$ copies of $\rho$ and $\sigma$, distinguishes between cases where $\rho=\sigma$ and cases where $||\rho-\sigma||_1>\epsilon$ with high probability.
And, in particular, for $n\geq 4$, there is an observable $M$ such that
\begin{align*}
\mathbb{E}(M)&=\mathbb{E}[M(\rho^{\otimes n}\otimes\sigma^{\otimes n})]=||\rho-\sigma||_2^2,\\
\mathrm{Var}(M)&=\tr[M^2(\rho^{\otimes n}\otimes\sigma^{\otimes n})]-||\rho-\sigma||_2^4=O(\frac{1}{n^2}+\frac{||\rho-\sigma||_2^2}{n}).
\end{align*}
\end{theorem}
The requirement of $n\geq 4$ is not explicitly implied in the proof of Proposition 5.6 in \cite{BOW17}.

\begin{algorithm}[H]
	\Input{$O(\frac{d}{\epsilon^2})$ copies of $\rho\in \mathcal{D}(\C^d)$ and $O(\frac{d}{\epsilon^2})$ copies of $\sigma\in \mathcal{D}(\C^d)$}	
    \Output{"Yes" with a probability of at least $\frac{2}{3}$ if $\rho=\sigma$; and "No" with a probability of at least $\frac{2}{3}$ if $||\rho-\sigma||_1>\epsilon$.}
	\caption{\textsf{A Identity Test with Joint Measurement}}
	\label{algo:Identity-Test-Joint-Measurements}
\end{algorithm}

\begin{theorem}\cite{Chan:2014:OAT:2634074.2634162}\label{classical}
For $n$-dimensional probability distributions of $p$ and $q$, $O(\frac{b}{\epsilon^2})$ samples are sufficient to distinguish, with at least a $\frac{2}{3}$ probability, the cases where $p=q$ from the cases where $||p-q||_2> \epsilon$, where $b\geq ||p||_2,||q||_2$.
\end{theorem}

\begin{algorithm}[H]
	\Input{$O(\frac{b}{\epsilon^2})$ copies of $p$ and $O(\frac{b}{\epsilon^2})$ copies of $q$}	
    \Output{"Yes" with probability at least $\frac{2}{3}$ if $p=q$, "No" with probability at least $\frac{2}{3}$ if $||p-q||_2>\epsilon$.}
	\caption{\textsf{An $\ell_2$ norm Identity Test}}
	\label{algo:2-norm-Identity-Test}
\end{algorithm}

\begin{theorem}\cite{Canonne:2018:TCI:3188745.3188756}\label{classical-estimate}
For probability distributions $p_{AB}\in\Delta(A\times B)$ with $b\geq ||p_{AB}||_2,||p_A\otimes p_B||_2$ there is an estimator $Q:(A\times B)^n\mapsto \mathbb{R}$ such that $n\geq 4$ samples of $p_{AB}\in \Delta(A\times B)$, say $X$,
\begin{align*}
&\mathbb{E}Q(X)=||p_{AB}-p_A\otimes p_B||_2^2,\\
&\mathrm{Var}[Q(X)]=O(\frac{b||p_{AB}-p_A\otimes p_B||^2_2}{n}+\frac{b^2}{n^2}).
\end{align*}
\end{theorem}

\section{Quantum Independence and Technical Lemmata}\label{LT}
\subsection{Bipartite independence and approximate independence}
We say that $\rho_{1,2}\in \mathcal{D}(\C^{d}\otimes\C^{d_2})$ is independent if $\rho_{1,2}=\sigma_1\otimes\sigma_2$ for some $\sigma_i\in \mathcal{D}(\C^{d_i})$. One can directly verify that, if $\rho_{1,2}$ is independent, then $\rho=\rho_1\otimes\rho_2$ with $\rho_1$ and $\rho_2$ being the reduced density matrices of $\rho_{1,2}$.

We say that $\rho$ is $\epsilon$-independent with respect to the $\ell_1$ distance if there is an  independent state $\sigma$ such that $||\rho-\sigma||_1\leq\epsilon$.
We say that $\rho$ is $\epsilon$-far from being independent with respect to the $\ell_1$ distance if $||\rho-\sigma||_1>\epsilon$ for any independent state $\sigma$.
\begin{Prop}\label{dis1}
Let $\rho$ and $\sigma$ be bipartite states of $\mathcal{D}(\C^{d}\otimes\C^{d_2})$. If $||\rho-\sigma||_1\leq\epsilon/3$ and $\sigma$ is independent, then $||\rho-\rho_1\otimes\rho_2||_1\leq \epsilon$.
\end{Prop}
Proposition \ref{dis1} follows from the Lemma \ref{dis2}.

According to triangle inequality, observe that
$$
||\rho_1\otimes\rho_2-\sigma_1\otimes\sigma_2||_1\leq ||\rho_1\otimes\rho_2-\rho_1\otimes\sigma_2||_1+||\rho_1\otimes\sigma_2-\sigma_1\otimes\sigma_2||_1=||\rho_1-\sigma_1||_1+||\rho_2-\sigma_2||_1.
$$
Therefore, we have
\begin{lemma} \label{dis2}
$$||\rho_1\otimes\rho_2-\sigma_1\otimes\sigma_2||_1\leq ||\rho_1-\sigma_1||_1+||\rho_2-\sigma_2||_1.$$
\end{lemma}
\textit{Proof of Proposition \ref{dis1}:} Clearly, $\sigma=\sigma_1\otimes\sigma_2$. Thus, we have
$$
||\rho-\rho_1\otimes\rho_2||_1\leq ||\rho-\sigma||_1+||\sigma-\rho_1\otimes\rho_2||_1=||\rho-\sigma||_1+||\sigma_1\otimes\sigma_2-\rho_1\otimes\rho_2||_1\leq \epsilon/3+2\epsilon/3=\epsilon,
$$
where the last inequality accords to Lemmas \ref{dis2} and \ref{pt1norm},
$$||\rho_1-\sigma_1||_1,||\rho_2-\sigma_2||_1\leq ||\rho-\sigma||_1\leq \epsilon/3.$$

\subsection{Multipartite independence and approximate independence}

We say that $\rho\in\mathcal{D}(\C^{d}\otimes\C^{d_2}\otimes\cdots\otimes\C^{d_m})$ is $m$-partite independent if $\rho=\rho_1\otimes\rho_2\otimes\cdots\otimes\rho_m$,
and that $\rho$ is $\epsilon$-independent with respect to the $\ell_1$ distance if there is a state $\sigma$ that is $m$-partite independent and $||\rho-\sigma||_1\leq\epsilon$.
We say that $\rho$ is $\epsilon$-far from being independent with respect to the $\ell_1$ distance if $||\rho-\sigma||_1>\epsilon$ for any $m$-partite independent state $\sigma$ .
\begin{Prop}\label{dism}
Let $\rho$ and $\sigma$ be $m$-partite states, if $||\rho-\sigma||_1\leq{\epsilon}$, and $\sigma$ is $m$-partite independent, then $||\rho-\rho_1\otimes\rho_2\otimes\cdots\otimes\rho_m||_1\leq (m+1)\epsilon$.
\end{Prop}
Applying the triangle inequality, we observe that
\begin{eqnarray*}
&&||\rho_1\otimes\rho_2\otimes\cdots\otimes\rho_m-\sigma_1\otimes\sigma_2\otimes\cdots\otimes\sigma_m||_1\\
&\leq& ||\rho_1\otimes\rho_2\otimes\cdots\otimes\rho_m-\rho_1\otimes\sigma_2\otimes\cdots\otimes\sigma_m||_1
+||\rho_1\otimes\sigma_2\otimes\cdots\otimes\sigma_m-\sigma_1\otimes\sigma_2\otimes\cdots\otimes\sigma_m||_1\\
&=& ||\rho_2\otimes\cdots\otimes\rho_m-\otimes\sigma_2\otimes\cdots\otimes\sigma_m||_1+||\rho_1-\sigma_1||_1\\
&&\cdots\\
&\leq& \sum_{i=1}^m ||\rho_i-\sigma_i||_1.
\end{eqnarray*}
Then we have
\begin{lemma} \label{dism1}
$$||\rho_1\otimes\rho_2\otimes\cdots\otimes\rho_m-\sigma_1\otimes\sigma_2\otimes\cdots\otimes\sigma_m||_1\leq \sum_{i=1}^m ||\rho_i-\sigma_i||_1.$$
\end{lemma}
\textit{Proof of Proposition \ref{dism}:} Clearly, $\sigma=\sigma_1\otimes\sigma_2\otimes\cdots\otimes\sigma_m$. From Lemma \ref{dism1}, we have
\begin{align*}
&||\rho-\rho_1\otimes\rho_2\otimes\cdots\otimes\rho_m||_1\\
\leq &||\rho-\sigma||_1+||\sigma-\rho_1\otimes\rho_2\otimes\cdots\otimes\rho_m||_1\\
=&||\rho-\sigma||_1+||\sigma_1\otimes\sigma_2\otimes\cdots\otimes\sigma_m-\rho_1\otimes\rho_2\otimes\cdots\otimes\rho_m||_1\\
\leq &\epsilon+\sum_{i=1}^m ||\sigma_i-\rho_i||_1\\
\leq &\epsilon+\sum_{i=1}^m ||\rho-\sigma||_1\\
\leq&(m+1)\epsilon.
\end{align*}
Proposition \ref{dis2m} establishes a connection between bipartite independence and multipartite independence. Specifically, it shows that if an $m$-partite state is close to bipartite independence in any $1$ versus $m-1$ cut, it is close to being $m$ partite independent.
\begin{Prop}\label{dis2m}
Let $\rho$ be an $m$-partite states. If for any $1\leq i\leq m$, there exists a state ${\sigma}^{(i)}_{i}$ of party $i$, and a state ${\psi}_{S\setminus\{i\}}$ of parties $S\setminus\{i\}$ such that $||\rho-{\sigma}^{(i)}_{i}\otimes{\psi}_{S\setminus\{i\}}||_1\leq\epsilon$, then $||\rho-\rho_1\otimes\rho_2\otimes\cdots\otimes\rho_m||_1\leq 5m\epsilon$.
\end{Prop}
\textit{Proof of Proposition \ref{dis2m}:}
First, we can prove by induction that
$$
||\rho-{\sigma}^{(1)}_{1}\otimes{\sigma}^{(2)}_{2}\otimes\cdots\otimes{\sigma}^{(m)}_{m}||_1\leq 4m\epsilon.
$$
If $m=2$, we know that
\begin{align*}
&||\rho-{\sigma}^{(1)}_{1}\otimes{\sigma}^{(2)}_{2}||_1\\
\leq &||\rho-{\sigma}^{(1)}_{1}\otimes{\psi}_{2}||_1+||{\sigma}^{(1)}_{1}\otimes{\sigma}^{(2)}_{2}-{\sigma}^{(1)}_{1}\otimes{\psi}_{2}||_1\\
\leq &\epsilon+||{\sigma}^{(2)}_{2}-\psi_2||_1\\
\leq &\epsilon+||\rho_2-\psi_2||_1+||{\sigma}^{(2)}_{2}-\rho_2||_1\\
\leq &\epsilon+||\rho-{\sigma}^{(1)}_{1}\otimes\psi_2||_1+||\psi_1\otimes{\sigma}^{(2)}_{2}-\rho||_1\\
\leq &3\epsilon\\
\leq& 8\epsilon.
\end{align*}
By Lemma \ref{pt1norm}, we know that, for any $1<i\leq m$, $\beta_{S\setminus\{1,i\}}:=\tr_{1}{\psi}_{S\setminus\{i\}}$ satisfies
$$||\rho_{S\setminus\{1\}}-{\sigma}^{(i)}_{i}\otimes\beta_{S\setminus\{1,i\}}||_1\leq \epsilon.$$
According to Proposition \ref{dis1}, we have
$$
||\rho-\rho_1\otimes\rho_{S\setminus\{1\}}||_1\leq 3\epsilon.
$$
Therefore, by induction, we have
\begin{eqnarray*}
&&||\rho-{\sigma}^{(1)}_{1}\otimes{\sigma}^{(2)}_{2}\otimes\cdots\otimes{\sigma}^{(m)}_{m}||_1\\
\leq& & ||\rho-\rho_1\otimes\rho_{S\setminus\{1\}}||_1+||\rho_1\otimes\rho_{S\setminus\{1\}}-{\sigma}^{(1)}_{1}\otimes{\sigma}^{(2)}_{2}\otimes\cdots\otimes{\sigma}^{(m)}_{m}||_1\\
\leq &&3\epsilon+ ||\rho_{S\setminus\{1\}}-{\sigma}^{(2)}_{2}\otimes\cdots\otimes{\sigma}^{(m)}_{m}||_1+||\rho_1-{\sigma}^{(1)}_{1}||_1\\
=&& ||\rho_{S\setminus\{1\}}-{\sigma}^{(2)}_{2}\otimes\cdots\otimes{\sigma}^{(m)}_{m}||_1+4\epsilon\\
\leq&& 4(m-1)\epsilon+4\epsilon\\
=&&4m\epsilon,
\end{eqnarray*}
The third inequality is also derived by induction, and we have
\begin{eqnarray*}
&&||\rho-\rho_1\otimes\rho_2\otimes\cdots\otimes\rho_m||_1\\
\leq&&||\rho-{\sigma}^{(1)}_{1}\otimes{\sigma}^{(2)}_{2}\otimes\cdots\otimes{\sigma}^{(m)}_{m}||_1+||\rho_1\otimes\rho_2\otimes\cdots\otimes\rho_m-{\sigma}^{(1)}_{1}\otimes{\sigma}^{(2)}_{2}\otimes\cdots\otimes{\sigma}^{(m)}_{m}||_1\\
\leq&& 4m\epsilon
+\sum_{i=1}^m ||\rho_i-{\sigma}^{(i)}_{i}||_1\\
\leq&& 4m\epsilon+\sum_{i=1}^m||\rho-{\sigma}^{(i)}_{i}\otimes{\psi}_{S\setminus\{i\}}||_1\\
\leq &&5m\epsilon.
\end{eqnarray*}
\subsection{Conditional independence}
Consider the following set of classical-quantum-quantum states:
$$
\mathcal{T}_{ABC}=\{\sum_{i\in C}p_c \rho^{c}_{AB}\otimes\op{c}{c}: p_c\geq 0,~\sum_{c\in C}p_c=1,~\rho^{c}_{AB}\in\mathcal{D}(\C^{d_1}\otimes\C^{d_2})\},
$$
where $\ket{c}$ are fixed orthonormal bases of system $C$.

$\mathcal{P}_{A,B|C}$ denotes the property of conditional independence for the classical-quantum-quantum states:
$$
\mathcal{P}_{A,B|C}=\{\sum_{c\in C}p_c\sigma^{c}_{A}\otimes\sigma^{c}_{B}\otimes \op{c}{c}: p_c\geq 0,~\sum_{c\in C}p_c=1,~\rho^{c}_{A}\in\mathcal{D}(\C^{d_1}),~\rho^{c}_{B}\in\mathcal{D}(\C^{d_2})\}\subset\mathcal{T}_{ABC}.
$$
We can directly generalize the definition of the $\ell_1$ distance for the states in $\mathcal{T}_{ABC}$ as follows. Given $\rho_{ABC}=\sum_{i\in C}p_c\rho^{i}_{AB}\otimes \op{c}{c}$ and $\sigma_{ABC}=\sum_{i\in C}q_c \sigma^{i}_{AB}\otimes\op{c}{c}$, we define
$$
||\rho_{ABC}-\sigma_{ABC}||_1=||\sum_{i\in C}p_c\rho^{c}_{AB}\otimes \op{c}{c}-\sum_{c\in C}q_c \sigma^{c}_{AB}\otimes\op{c}{c}||_1=\sum_{c\in C}||p_c\rho^{c}_{AB}-q_c\sigma^{c}_{AB}||_1.
$$
We say $\rho_{ABC}\in \mathcal{T}_{ABC}$ is $\epsilon$-far from $\mathcal{P}_{AB|C}$ if $||\rho_{ABC}-\sigma_{ABC}||_1>\epsilon$ for any $\sigma\in\mathcal{P}_{AB|C}$.

Given access to $\rho_{ABC}\in \mathcal{T}_{ABC}$, the goal is to distinguish between $\rho_{ABC}\in\mathcal{P}_{AB|C}$ and $\rho_{ABC}$ being $\epsilon$-far from $\mathcal{P}_{AB|C}$, that is, to determine whether $A$ and $B$ are conditionally independent given $C$, versus $\epsilon$-far in $\ell_1$ distance.

The following lemmata are useful here.
\begin{lemma}
Given $\rho_{ABC}=\sum_{c\in C}p_c \rho^{c}_{AB}\otimes\op{c}{c}$ and $\sigma_{ABC}=\sum_{c\in C}q_c \sigma^{c}_{AB}\otimes\op{c}{c}$, we have that
$$
||\rho_{ABC}-\sigma_{ABC}||_1\leq \sum_{c\in C} p_c ||\rho^c_{AB}-\sigma^c_{AB}||_1+\sum_{c\in C} |p_c-q_c|=\sum_{c\in C} p_c ||\rho^c_{AB}-\sigma^c_{AB}||_1+||p-q||_1.
$$
\end{lemma}
\begin{proof}
Let
$$
\tau_{ABC}=\sum_{c\in C}p_c \sigma^{c}_{AB}\otimes\op{c}{c}
$$
By the triangle inequality, we have
\begin{align*}
||\rho_{ABC}-\sigma_{ABC}||_1\leq ||\rho_{ABC}-\tau_{ABC}||_1+||\tau_{ABC}-\sigma_{ABC}||_1=\sum_{c\in C} p_c ||\rho^c_{AB}-\sigma^c_{AB}||_1+\sum_{c\in C} |p_c-q_c|.
\end{align*}
\end{proof}

From this, we directly observe the following.
\begin{Cor}
If $\rho_{ABC}=\sum_{c\in C}p_c \rho^{c}_{AB}\otimes\op{c}{c}$ is $\epsilon$-far from $\mathcal{P}_{A,B|C}$, then for every $\sigma_{ABC}=\sum_{c\in C}q_c \sigma^{c}_{A}\otimes\sigma^c_{B}\otimes\op{c}{c}$, either $||p-q||_1>\frac{\epsilon}{2}$, or $\sum_{c\in C} p_c ||\rho^c_{AB}-\sigma^{c}_{A}\otimes\sigma^c_{B}||_1>\frac{\epsilon}{2}$.
\end{Cor}

By generalizing Proposition \ref{dis1}, Lemma \ref{555} shows a useful structural property of conditional independence that is crucial to
our algorithms. It shows that, if $\rho_{ABC}\in \mathcal{T}_{ABC}$ is close to being conditionally independent,
then it is close to an appropriate mixture of its tensor products of conditional reduced density matrices.
\begin{lemma} \label{555}
Suppose $\rho_{ABC}=\sum_{c\in C}p_c\rho^{c}_{AB}\otimes\op{c}{c}$ is $\epsilon$-close to $\mathcal{P}_{A,B|C}$. Then
\begin{align*}
||\rho_{ABC}-\tilde{\rho}_{ABC}||_1\leq 4\epsilon,
\end{align*}
where $\tilde{\rho}_{ABC}=\sum_{i\in C}p_{c}\rho^{c}_{A}\otimes\rho^{c}_{B}\otimes\op{c}{c}$ with $\rho^{c}_A$ and $\rho^c_{B}$ being reduced density matrices of $\rho^{c}_{AB}$.
\end{lemma}
\begin{proof}
Let $\sigma_{ABC}=\sum_{i\in C}q_c \sigma^{c}_{A}\otimes\sigma^c_{B}\otimes\op{c}{c}$ such that
$$
||\rho_{ABC}-\sigma_{ABC}||_1=\sum_{c\in C}||p_c\rho^{c}_{AB}-q_c\sigma^{c}_{A}\otimes\sigma^c_{B}||_1\leq \epsilon.
$$
According to Lemma \ref{pt1norm}, we have
\begin{align*}
||p-q||_1\leq \epsilon,\\
\sum_{c\in C}||p_c\rho^c_A-q_c\sigma^c_A||_1\leq\epsilon,\\
\sum_{c\in C}||p_c\rho^c_B-q_c\sigma^c_B||_1\leq\epsilon
\end{align*}
Therefore, by the triangle inequality, we have
\begin{align*}
&||\rho_{ABC}-\tilde{\rho}_{ABC}||_1\\
\leq& ||\rho_{ABC}-\sigma_{ABC}||_1+||\sigma_{ABC}-\tilde{\rho}_{ABC}||_1\\
\leq& \epsilon+\sum_{c\in C} ||p_c\rho^c_A\otimes\rho^c_B-q_c\sigma^{c}_{A}\otimes\sigma^c_{B}||_1\\
\leq& \epsilon+\sum_{c\in C} ||p_c\rho^c_A\otimes\rho^c_B-q_c\sigma^{c}_{A}\otimes\rho^c_{B}||_1+\sum_{c\in C} ||q_i\sigma^{c}_{A}\otimes\rho^c_{B}-q_c\sigma^{c}_{A}\otimes\sigma^c_{B}||_1\\
=&\epsilon+\sum_{c\in C}  ||p_c\rho^c_A-q_c\sigma^{c}_{A}||_1+\sum_{c\in C} ||q_c\rho^c_{B}-q_c\sigma^c_{B}||_1\\
\leq & 2\epsilon+\sum_{c\in C} ||q_c\rho^c_{B}-q_c\sigma^c_{B}||_1\\
\leq &2\epsilon+\sum_{c\in C} ||p_c\rho^c_{B}-q_c\rho^c_{B}||_1+\sum_{c\in C} ||p_c\rho^c_{B}-q_c\sigma^c_{B}||_1\\
=& 2\epsilon+||p-q||_1+\epsilon\\
\leq &4\epsilon.
\end{align*}
\end{proof}
In other words, $\rho_{ABC}$ is far from $\mathcal{P}_{A,B|C}$ if and only if $||\rho_{ABC}-\tilde{\rho}_{ABC}||_1$ is large.

\section{Connections between Quantum Property Testing and Distribution Testing}\label{connection}
In this section, we provide measurement schemes that map a quantum state to a probability distribution while maintaining $\ell_2$ norm relations for independent measurement followed by local measurement.

\subsection{Independent measurement}\label{sec:independence}
This subsection presents an independent measurement scheme, which generates a connection between quantum property testing and distribution testing.

Mutually unbiased
bases (MUBs) are used to map the quantum states of $\mathcal{D}(\C^{d})$ into $d(d+1)$ dimensional probability distributions. Without loss of generality, assume $d=2^k$, and we the Pauli group $\mathcal{P}_k=\{I,X,Y,Z\}^{\otimes k}$ be to the order of $4^k$. According to Theorem \ref{MUBs}, any state $\rho\in \mathcal{D}(\C^{d})$ can be written as
$$\rho=\sum_{P\in\mathcal{P}_k}\eta_p P=\frac{I_d}{d}+\sum_{a=0}^d\sum_{\begin{subarray}{c} P\in G_a,\\ P\neq I_d\end{subarray}}\eta_p P=\frac{I_d}{d}+\sum_{i,j}\mu_{i,j}\op{\beta_{i,j}}{\beta_{i,j}},$$
where $G_a$ are the Abelian subgroups with an order of $2^k=d$ such that $\cup G_a=\mathcal{P}_k$ and $G_a\bigcap G_b=\{I_2^{\otimes k}\}$ for $a\neq b$.
The equation is due to the simultaneous spectrum decomposition of $G_a$ through the MUBs bases. That is, for $0\leq i\neq s\leq d,1\leq j,t\leq d$, $$|\ip{\beta_{i,j}}{\beta_{s,t}}|=\frac{1}{\sqrt{d}}.$$
In additional, it is verifiable that $\sum_{j=1}^d \mu_{i,j}=0$ for all $i$ by the traceless property of $P\neq I_d$.
Therefore, we can obtain the following constraint on $\mu_{i,j}$
\begin{align} \label{mubb}
\frac{1}{d}+\sum_{i,j}\mu_{i,j}^2\leq 1.
\end{align}
To observe this, by the property of MUBs, we have
\begin{align*}
\tr\rho^2=&\tr\frac{I_d}{d^2}+\sum_{i,j,s,t}\mu_{i,j}\mu_{s,t}|\ip{\beta_{i,j}}{\beta_{s,t}}|^2\\
=&\frac{1}{d}+\sum_{i,j}\mu_{i,j}^2+\sum_{i\neq s}\sum_{j,t}\frac{\mu_{i,j}\mu_{s,t}}{d}\\
=&\frac{1}{d}+\sum_{i,j}\mu_{i,j}^2+\sum_{i\neq s}\sum_{j}\mu_{i,j}\frac{\sum_{t}\mu_{s,t}}{d}\\
=&\frac{1}{d}+\sum_{i,j}\mu_{i,j}^2\\
\leq& 1,
\end{align*}
where $\sum_{j=1}^d \mu_{i,j}=0$ for all $i$.

Now, POVM
\begin{align}\label{mubm}
\mathcal{M}=\{M_{ij}=\frac{\op{\beta_{i,j}}{\beta_{i,j}}}{d+1}: 0\leq i\leq d,~1\leq j\leq d\}
\end{align}
can be used to map the $d$-dimensional quantum state $\rho$ into a $d(d+1)$ dimensional probabilistic distribution.
The corresponding probability distribution $p=(p(0,1),\dots,p(d,d))$ satisfies
$$
p(i,j)=\frac{\tr(\rho\op{\beta_{i,j}}{\beta_{i,j}})}{d+1}=\frac{\mu_{i,j}+\frac{1}{d}}{d+1},
$$
note that other terms are orthogonal or cancel out due to the property of MUBs and the equations $\sum_{j=1}^d \mu_{i,j}=0$ for all $i$.

Then the $\ell_2$ norm of $p$ can be bounded with
$$
||p||_2=\sqrt{\sum_{i,j} p^2(i,j)}=\frac{\sqrt{\sum_{i,j}(\mu_{i,j}+\frac{1}{d})^2 }}{d+1}=\frac{\sqrt{\sum_{i,j}\mu_{i,j}^2+\frac{d(d+1)}{d^2}+\frac{2\sum_{i,j}\mu_{i,j}}{d}}}{d+1}=\frac{\sqrt{\sum_{i,j}\mu_{i,j}^2+\frac{d+1}{d}}}{d+1}\leq\frac{\sqrt{2}}{d+1}.
$$
More importantly, this map preserves the $\ell_2$ distance, in the sense that the $\ell_2$ distance between
the image probability distributions is exactly the same as the $\ell_2$ distance between the pre-image quantum states with a scaling of $\frac{1}{d+1}$.

For any two states $\rho=\frac{I_d}{d}+\sum_{i,j}\mu_{i,j}\op{\beta_{i,j}}{\beta_{i,j}}$ and $\sigma=\frac{I_d}{d}+\sum_{i,j}\nu_{i,j}\op{\beta_{i,j}}{\beta_{i,j}}$, we have that
$$
||\rho-\sigma||_2=||\sum_{i,j}(\mu_{i,j}-\nu_{i,j})\op{\beta_{i,j}}{\beta_{i,j}}||_2=\sqrt{\sum_{i,j}(\mu_{i,j}-\nu_{i,j})^2},
$$
where similar to Eq. (\ref{mubb}), the other terms are orthogonal or cancel out due to the property of MUBs and the equation $\sum_{j=1}^d \mu_{i,j}=0$ for all $i$.

Using the measurement $\mathcal{M}$ given in Eq. (\ref{mubm}),
the corresponding probability distributions can be obtained:
$p=(p(0,1),\dots,p(d,d))$ and $q=(q(0,1),\dots,q(d,d))$ with
\begin{eqnarray*}
p(i,j)=\frac{\tr(\rho\op{\beta_{i,j}}{\beta_{i,j}})}{d+1}=\frac{\mu_{i,j}+\frac{1}{d}}{d+1},\\
q(i,j)=\frac{\tr(\sigma\op{\beta_{i,j}}{\beta_{i,j}})}{d+1}=\frac{\nu_{i,j}+\frac{1}{d}}{d+1}.
\end{eqnarray*}
Therefore,
$$
||p-q||_2=\frac{\sqrt{\sum_{i,j}(\mu_{i,j}-\nu_{i,j})^2}}{d+1}=\frac{||\rho-\sigma||_2}{d+1}.
$$
\paragraph{Restatement of Theorem \ref{indepent-measurement-classical} }
There is an independent measurement scheme that maps the quantum states in $\mathcal{D}(\C^{d})$ into $d(d+1)$ dimensional probability distributions such that for all quantum states $\rho$ and $\sigma$
\begin{align}
||p-q||_2&=\frac{||\rho-\sigma||_2}{d+1},\\
||p||_2, ||q||_2& \leq\frac{\sqrt{2}}{d+1},
\end{align}
where $p$ and $q$ are the corresponding probability distributions of $\rho$ and $\sigma$, respectively.

\subsection{Local measurement}\label{sec:local}
In this subsection, we generalize the result of the previous subsection into the local measurement setting by proving a restatement of Theorem \ref{local-measurement-classical}.
\paragraph{Restatement of Theorem \ref{local-measurement-classical}}
For $\rho_{1,2,\dots,m},\sigma_{1,2,\dots,m}\in \mathcal{D}(\C^{d_1}\otimes \C^{d_2}\otimes\cdots\otimes\C^{d_m})$ where all $d_i$ are to the power of $2$, each local party is measured using measurement corresponding to MUBs. The resulting probability distribution
$p_{1,2,\dots,m},q_{1,2,\dots,m}\in \Delta(\times_{i=1}^m[d_i(d_i+1)])$ for any $S\subset[m]$ satisfies
\begin{align*}
&||p_{1,2,\dots,m}-q_{1,2,\dots,m}||_2=\frac{\sqrt{\sum_{S\subset[m]}||\rho_{S}-\sigma_{S}||_2^2}}{\Pi_{i=1}^m(d_i+1)},\\
&||p_{S}||_2,||q_S||_2\leq\frac{2^{|S|/2}}{\Pi_{i\in S}(d_i+1)}.
\end{align*}	
Moreover, the first equality is valid for any Hermitian matrix $\sigma$.
\begin{proof}
This general multipartite version of the proof follows naturally from the detailed proof for the bipartite version using the same framework.

According to Theorem \ref{MUBs}, the Pauli group $\mathcal{P}_k=\{I,X,Y,Z\}^{\otimes k}$ can be divided into $2^k+1$ Abelian subgroups of the order $2^k$, say, $G_0,\dots,G_{2^k}$ such that $G_i\bigcap G_j=\{I_2^{\otimes k}\}$ for $i\neq j$. Each subgroup is simultaneously diagonalizable by a corresponding basis. All these $2^k+1$ bases form $2^k+1$ MUBs of a $k$-qubit system.

Assume $d_1=2^{k_1}$ and $d_2=2^{k_2}$. Let $G_{k_1,i}$ be the Abelian subgroups of $\mathcal{P}_{k_1}$ with $\{\ket{\beta_{i,j}}:0\leq i\leq d_1,1\leq j\leq d_1\}$ being the corresponding MUBs. Let $G_{k_2,s}$ be the Abelian subgroups of $\mathcal{P}_{k_2}$ with $\{\ket{\alpha_{s,t}}:0\leq s\leq d_2,1\leq t\leq d_2\}$ being the corresponding MUBs.
That is, for $0\leq i_1\neq i_2\leq d,1\leq j_1,j_2\leq d_1$, $0\leq s_1\neq s_2\leq d,1\leq t_1,t_2\leq d_1$,
\begin{align*}
|\ip{\beta_{i_1,j_1}}{\beta_{i_2,j_2}}|=\frac{1}{\sqrt{d_1}},\\
|\ip{\alpha_{s_1,t_1}}{\alpha_{s_2,t_2}}|=\frac{1}{\sqrt{d_2}}.
\end{align*}

Any state $\rho\in  \mathcal{D}(\C^{d_1}\otimes\C^{d_2})$ can be written as
\begin{align*}
\rho_{1,2}&=\sum_{P_1\in\mathcal{P}_{k_1},P_2\in\mathcal{P}_{k_2}}\eta_{P_1,P_2} P_1\otimes P_2\\
&=\frac{I_{d_1}\otimes I_{d_2}}{d_1 d_2}+\sum_{s=0}^{d_2}\sum_{\begin{subarray}{c}P_2\in G_{k_2,s},\\P_{k_2}\neq I_{k_2}\end{subarray}}\eta_{I,P_2}\frac{I_{d_1}}{d_1}\otimes P_2+\sum_{i=0}^{d_1}\sum_{\begin{subarray}{c}P_1\in G_{k_1,i},\\ P_1\neq I_{k_1}\end{subarray}} \eta_{P_1,I}P_1\otimes \frac{I_{d_2}}{d_2}\\
&+\sum_{s=0}^{d_2}\sum_{\begin{subarray}{c}P_2\in G_{k_2,s},\\ P_{k_2}\neq I_{k_2}\end{subarray}}\sum_{i=0}^{d_1}\sum_{\begin{subarray}{c}P_1\in G_{k_1,i},\\ P_1\neq I_{k_1}\end{subarray}}\eta_{P_1,P_2} P_1\otimes P_2\\
&=\frac{I_{d_1}\otimes I_{d_2}}{d_1 d_2}+\sum_{s,t}\mu_{s,t}\frac{I_{d_1}}{d_1}\otimes \op{\alpha_{s,t}}{\alpha_{s,t}}+\sum_{i,j} \nu_{i,j}\op{\beta_{i,j}}{\beta_{i,j}} \otimes\frac{I_{d_2}}{d_2} +\sum_{i,j,s,t}\chi_{i,j,s,t}\op{\beta_{i,j}}{\beta_{i,j}}\otimes \op{\alpha_{s,t}}{\alpha_{s,t}},
\end{align*}
Following the trace of Pauli matrices, we always have
\begin{align*}
\sum_{t}\mu_{s,t}=0,~~~~~~~\sum_{j}\nu_{i,j}=0,\\
\sum_{t}\chi_{i,j,s,t}=0,~~~~~\sum_{j}\chi_{i,j,s,t}=0.
\end{align*}
Therefore,
\begin{align*}
\rho_1=\frac{I_{d_1}}{d_1}+\sum_{i,j} \nu_{i,j}\op{\beta_{i,j}}{\beta_{i,j}}\\
\rho_2=\frac{I_{d_2}}{d_2}+\sum_{s,t}\mu_{s,t}\op{\alpha_{s,t}}{\alpha_{s,t}}.
\end{align*}
The following inequalities are demonstrated in Section \ref{sec:independence}:
\begin{align*}
\tr\rho_1^2=&\tr\frac{I_{d_1}}{d_1^2}+\sum_{i,j}\nu_{i,j}^2=\frac{1}{d_1}+\nu_{i,j}^2\leq 1,\\
\tr\rho_2^2=&\tr\frac{I_{d_2}}{d_2^2}+\sum_{s,t}\mu_{s,t}^2=\frac{1}{d_2}+\mu_{s,t}^2\leq 1,
\end{align*}
as can the following bound for $\chi_{i,j,s,t}$.
Note that the following terms are mutually orthogonal
\begin{align*}
\frac{I_{d_1}\otimes I_{d_2}}{d_1 d_2},~~~~~~~~~~~~\sum_{s,t}\mu_{s,t}\frac{I_{d_1}}{d_1}\otimes \op{\alpha_{s,t}}{\alpha_{s,t}},\\
\sum_{i,j} \nu_{i,j}\op{\beta_{i,j}}{\beta_{i,j}} \otimes\frac{I_{d_2}}{d_2},~~\sum_{i,j,s,t}\chi_{i,j,s,t}\op{\beta_{i,j}}{\beta_{i,j}}\otimes \op{\alpha_{s,t}}{\alpha_{s,t}}.
\end{align*}
The $\ell_2$ norms of these components can be computed with
\begin{align*}
&||\sum_{s,t}\mu_{s,t}\op{\alpha_{s,t}}{\alpha_{s,t}}||_2^2=\sum_{s_1,t_1,s_2,t_2}\mu_{s_1,t_1}\mu_{s_2,t_2}|\ip{\alpha_{s_1,t_1}}{\alpha_{s_2,t_2}}|^2
=\sum_{s,t}\mu_{s,t}^2+\sum_{s_1\neq s_2}\sum_{t_1}\nu_{s_1,t_1}\frac{\sum_{t_2}\nu_{s_2,t_2}}{d_2}
=\sum_{s,t}\mu_{s,t}^2,\\
&||\sum_{i,j} \nu_{i,j}\op{\beta_{i,j}}{\beta_{i,j}}||_2^2=\sum_{i_1,j_1,i_2,j_2}\nu_{i_1,j_1}\nu_{i_2,s_2}|\ip{\beta_{i_1,j_1}}{\beta_{i_2,j_2}}|^2
=\sum_{i,j}\nu_{i,j}^2+\sum_{i_1\neq i_2}\sum_{j_1}\nu_{i_1,j_1}\frac{\sum_{j_2}\nu_{i_2,j_2}}{d_1}
=\sum_{i,j}\nu_{i,j}^2,
\end{align*}
and
\begin{align*}
&||\sum_{i,j,s,t}\chi_{i,j,s,t}\op{\beta_{i,j}}{\beta_{i,j}}\otimes \op{\alpha_{s,t}}{\alpha_{s,t}}||_2^2\\
=&\sum_{i_1,j_1,s_1,t_1}\sum_{i_2,j_2,s_2,t_2}\chi_{i_1,j_1,s_1,t_1}\chi_{i_2,j_2,s_2,t_2}|\ip{\beta_{i_1,j_1}}{\beta_{i_2,j_2}}|^2|\ip{\alpha_{s_1,t_1}}{\alpha_{s_2,t_2}}|^2\\
=&\sum_{i,j,s,t}\chi^2_{i,j,s,t}+\sum_{i_1\neq i_2,s_1\neq s_2}\frac{\sum_{j_1,j_2}\sum_{t_1,t_2}\chi_{i_1,j_1,s_1,t_1}\chi_{i_2,j_2,s_2,t_2}}{d_1d_2}\\
=&\sum_{i,j,s,t}\chi^2_{i,j,s,t}+\sum_{i_1\neq i_2,s_1\neq s_2}\frac{\sum_{j_1,j_2}\sum_{t_1}\chi_{i_1,j_1,s_1,t_1}\sum_{t_2}\chi_{i_2,j_2,s_2,t_2}}{d_1d_2}\\
=&\sum_{i,j,s,t}\chi^2_{i,j,s,t}.
\end{align*}
Therefore, we have
\begin{align*}
\tr\rho_{1,2}^2=&\tr \frac{I_{d_1}\otimes I_{d_2}}{d^2_1 d^2_2}+\frac{||\sum_{s,t}\mu_{s,t}\op{\alpha_{s,t}}{\alpha_{s,t}}||_2^2}{d_1}+\frac{||\sum_{i,j} \nu_{i,j}\op{\beta_{i,j}}{\beta_{i,j}}||_2^2}{d_2}+||\sum_{i,j,s,t}\chi_{i,j,s,t}\op{\beta_{i,j}}{\beta_{i,j}}\otimes \op{\alpha_{s,t}}{\alpha_{s,t}}||_2^2\\
&=\frac{1}{d_1d_2}+\frac{\sum_{s,t}\mu_{s,t}^2}{d_1}+\frac{\sum_{i,j}\nu_{i,j}^2}{d_2}+\sum_{i,j,s,t}\chi_{i,j,s,t}^2\leq 1.
\end{align*}
Now, POVM
$$
\mathcal{M}=\{M_{i,j,s,t}=\frac{\op{\beta_{i,j}}{\beta_{i,j}}\otimes\op{\alpha_{s,t}}{\alpha_{s,t}}}{(d_1+1)(d_2+1)}: 0\leq i\leq d_1,~1\leq j\leq d_1,~0\leq s\leq d_2,~1\leq t\leq d_2\}
$$
can be used to map the $d_1d_2$-dimensional quantum state into the distribution on $[d_1(d_1+1)]\times[d_2(d_2+1)]$.

Note that this is not merely an independent measurement, it is also a local measurement:
\begin{align*}
\mathcal{M}=&\mathcal{M}_1\otimes\mathcal{M}_2,\\
\mathcal{M}_1=&\{M_{i,j}=\frac{\op{\beta_{i,j}}{\beta_{i,j}}}{d_1+1}: 0\leq i\leq d_1,~1\leq j\leq d_1\},\\
\mathcal{M}_2=& \{N_{s,t}=\frac{\op{\alpha_{s,t}}{\alpha_{s,t}}}{d_2+1}: ~0\leq s\leq d_2,~1\leq t\leq d_2\}.
\end{align*}
The corresponding probability distribution $p_{1,2}$ satisfies
\begin{align*}
p_{1,2}(i,j,s,t)=&\frac{\tr(\rho\op{\beta_{i,j}}{\beta_{i,j}}\otimes\op{\alpha_{s,t}}{\alpha_{s,t}})}{(d_1+1)(d_2+1)}\\
=&\frac{1}{d_1d_2(d_1+1)(d_2+1)}+\frac{\mu_{s,t}}{d_1(d_1+1)(d_2+1)}
+\frac{\nu_{i,j}}{d_2(d_1+1)(d_2+1)}+\frac{\chi_{i,j,s,t}}{(d_1+1)(d_2+1)}.
\end{align*}
$p_{1,2}$ can be regarded as a bipartite distribution on $[d_1(d_1+1)]\times[d_2(d_2+1)]$. Then,
\begin{align*}
p_1(i,j)=\sum_{s,t}p_{1,2}(i,j,s,t)=\sum_{s,t}\tr[\rho_{1,2}(M_{i,j}\otimes N_{s,t})]=\tr[\rho_{1,2}(M_{i,j}\otimes \sum_{s,t}N_{s,t})]=\tr[\rho_{1,2}(M_{i,j}\otimes I_{d_2})]=\tr(\rho_{1}M_{i,j})\\
p_2(s,t)=\sum_{i,j}p_{1,2}(i,j,s,t)=\sum_{i,j}\tr[\rho_{1,2}(M_{i,j}\otimes N_{s,t})]=\tr[\rho_{1,2}(\sum_{i,j}M_{i,j}\otimes N_{s,t})]=\tr[\rho_{1,2}(I_{d_1}\otimes N_{s,t})]=\tr(\rho_{2}N_{s,t}).
\end{align*}
In other words, $p_1$ is the distribution from applying POVM $\mathcal{M}_1$ of $\rho_1$, and $p_2$ is the distribution from applying POVM $\mathcal{M}_2$ of $\rho_2$.

$||p_{1,2}||_2^2$ can be bounded with
\begin{align*}
&(d_1+1)^2(d_2+1)^2\sum_{i,j,s,t} p^2_{1,2}(i,j,s,t)\\
=&\sum_{i,j,s,t}(\frac{1}{d_1d_2}+\frac{\mu_{s,t}}{d_1}
+\frac{\nu_{i,j}}{d_2}+\chi_{i,j,s,t})^2 \\
=&\sum_{i,j,s,t}\frac{1}{d_1^2d_2^2}+\sum_{i,j,s,t}\frac{\mu_{s,t}^2}{d_1^2}+\sum_{i,j,s,t}\frac{\nu_{i,j}^2}{d_2^2}+\sum_{i,j,s,t}\chi_{i,j,s,t}^2+2\sum_{i,j,s,t}\frac{\mu_{s,t}}{d_1^2d_2}+2\sum_{i,j,s,t}\frac{\nu_{i,j}}{d_1d_2^2}+\sum_{i,j,s,t}\frac{2\chi_{i,j,s,t}}{d_1d_2}\\
&+2\sum_{i,j,s,t}\frac{\mu_{s,t}\nu_{i,j}}{d_1d_2}+2\sum_{i,j,s,t}\frac{\mu_{s,t}\chi_{i,j,s,t}}{d_1}+2\sum_{i,j,s,t}\frac{\nu_{i,j}\chi_{i,j,s,t}}{d_2}\\
=&\frac{(d_1+1)(d_2+1)}{d_1d_2}+\frac{d_1+1}{d_1}\sum_{s,t}\mu_{s,t}^2+\frac{d_2+1}{d_2}\sum_{i,j}\nu_{i,j}^2+\sum_{i,j,s,t}\chi_{i,j,s,t}^2\\
&+2\frac{(\sum_{i,j}\nu_{i,j})(\sum_{s,t}\mu_{s,t})}{d_1d_2}+2\sum_{s,t}\frac{\mu_{s,t}(\sum_{i,j}\chi_{i,j,s,t})}{d_1}+2\sum_{i,j}\frac{\nu_{i,j}(\sum_{s,t}\chi_{i,j,s,t})}{d_2}\\
=&\frac{(d_1+1)(d_2+1)}{d_1d_2}+\frac{d_1+1}{d_1}\sum_{s,t}\mu_{s,t}^2+\frac{d_2+1}{d_2}\sum_{i,j}\nu_{i,j}^2+\sum_{i,j,s,t}\chi_{i,j,s,t}^2,
\end{align*}
where we use these facts:
\begin{align*}
\sum_{t}\mu_{s,t}=0,~~~~~~~\sum_{j}\nu_{i,j}=0,\\
\sum_{t}\chi_{i,j,s,t}=0,~~~~~\sum_{j}\chi_{i,j,s,t}=0.
\end{align*}
According to the inequalities obtained, we have
\begin{align*}
&\frac{(d_1+1)(d_2+1)}{d_1d_2}+\frac{d_1+1}{d_1}\sum_{s,t}\mu_{s,t}^2+\frac{d_2+1}{d_2}\sum_{i,j}\nu_{i,j}^2+\sum_{i,j,s,t}\chi_{i,j,s,t}^2\\
=&1+\frac{1}{d_1}+\frac{1}{d_2}+\frac{1}{d_1d_2}+\frac{\sum_{s,t}\mu_{s,t}^2}{d_1}+\sum_{s,t}\mu_{s,t}^2+\frac{\sum_{i,j}\nu_{i,j}^2}{d_2}+\sum_{i,j}\nu_{i,j}^2+\sum_{i,j,s,t}\chi_{i,j,s,t}^2\\
\leq& 4
\end{align*}
Therefore,
\begin{align*}
||p_{1,2}||_2\leq \frac{2}{(d_1+1)(d_2+1)}.
\end{align*}
For any two states
\begin{align*}
\rho_{1,2}=\frac{I_{d_1}\otimes I_{d_2}}{d_1 d_2}+\sum_{s,t}\mu_{s,t}\frac{I_{d_1}}{d_1}\otimes \op{\alpha_{s,t}}{\alpha_{s,t}}+\sum_{i,j} \nu_{i,j}\op{\beta_{i,j}}{\beta_{i,j}} \otimes\frac{I_{d_2}}{d_2} +\sum_{i,j,s,t}\chi_{i,j,s,t}\op{\beta_{i,j}}{\beta_{i,j}}\otimes \op{\alpha_{s,t}}{\alpha_{s,t}}\\
\sigma_{1,2}=\frac{I_{d_1}\otimes I_{d_2}}{d_1 d_2}+\sum_{s,t}\mu_{s,t}'\frac{I_{d_1}}{d_1}\otimes \op{\alpha_{s,t}}{\alpha_{s,t}}+\sum_{i,j} \nu_{i,j}'\op{\beta_{i,j}}{\beta_{i,j}} \otimes\frac{I_{d_2}}{d_2} +\sum_{i,j,s,t}\chi_{i,j,s,t}'\op{\beta_{i,j}}{\beta_{i,j}}\otimes \op{\alpha_{s,t}}{\alpha_{s,t}},
\end{align*}
we have that
\begin{align*}
&||\rho_{1,2}-\sigma_{1,2}||_2\\
=&||\sum_{s,t}(\mu_{s,t}-\mu_{s,t}')\frac{I_{d_1}}{d_1}\otimes \op{\alpha_{s,t}}{\alpha_{s,t}}+\sum_{i,j} (\nu_{i,j}-\nu_{i,j}')\op{\beta_{i,j}}{\beta_{i,j}} \otimes\frac{I_{d_2}}{d_2} +\sum_{i,j,s,t}(\chi_{i,j,s,t}-\chi_{i,j,s,t}')\op{\beta_{i,j}}{\beta_{i,j}}\otimes \op{\alpha_{s,t}}{\alpha_{s,t}}||_2\\
=&\sqrt{\frac{\sum_{s,t}(\mu_{s,t}-\mu_{s,t}')^2}{d_1}+\frac{\sum_{i,j}(\nu_{i,j}-\nu_{i,j}')^2}{d_2}+\sum_{i,j,s,t}(\chi_{i,j,s,t}-\chi_{i,j,s,t}')^2}.
\end{align*}
In performing the POVM
$$
\mathcal{M}=\{M_{i,j,s,t}=\frac{\op{\beta_{i,j}}{\beta_{i,j}}\otimes\op{\alpha_{s,t}}{\alpha_{s,t}}}{(d_1+1)(d_2+1)}: 0\leq i\leq d_1,~1\leq j\leq d_1,~0\leq s\leq d_2,~1\leq t\leq d_2\},
$$
the corresponding probability distributions
$p_{1,2}$ and $q_{1,2}$ satisfy
\begin{eqnarray*}
p_{1,2}(i,j,s,t)=&\frac{1}{d_1d_2(d_1+1)(d_2+1)}+\frac{\mu_{s,t}}{d_1(d_1+1)(d_2+1)}
+\frac{\nu_{i,j}}{d_2(d_1+1)(d_2+1)}+\frac{\chi_{i,j,s,t}}{(d_1+1)(d_2+1)},\\
q_{1,2}(i,j,s,t)=&\frac{1}{d_1d_2(d_1+1)(d_2+1)}+\frac{\mu_{s,t}'}{d_1(d_1+1)(d_2+1)}
+\frac{\nu_{i,j}'}{d_2(d_1+1)(d_2+1)}+\frac{\chi_{i,j,s,t}'}{(d_1+1)(d_2+1)}.
\end{eqnarray*}
Therefore,
\begin{align*}
&||p_{1,2}-q_{1,2}||_2\\
=&\frac{\sqrt{\sum_{i,j,s,t}(\frac{\mu_{s,t}-\mu_{s,t}'}{d_1}
+\frac{\nu_{i,j}-\nu_{i,j}'}{d_2}+\chi_{i,j,s,t}-\chi_{i,j,s,t}')^2}}{(d_1+1)(d_2+1)}\\
=&\frac{\sqrt{(d_1+1)\frac{\sum_{s,t}(\mu_{s,t}-\mu_{s,t}')^2}{d_1}+(d_2+1)\frac{\sum_{i,j}(\nu_{i,j}-\nu_{i,j}')^2}{d_2}+\sum_{i,j,s,t}(\chi_{i,j,s,t}-\chi_{i,j,s,t}')^2}}{(d_1+1)(d_2+1)}\\
=&\frac{\sqrt{||\rho_2-\sigma_2||_2^2+||\rho_1-\sigma_1||_2^2+||\rho_{1,2}-\sigma_{1,2}||_2^2}}{(d_1+1)(d_2+1)},
\end{align*}
wherein the second equality, the cross-terms cancel out using the same argument for computing $||p_{1,2}||_2$,
and the equalities derived in this subsection:
$$
||\rho_{1,2}-\sigma_{1,2}||_2=\sqrt{\frac{\sum_{s,t}(\mu_{s,t}-\mu_{s,t}')^2}{d_1}+\frac{\sum_{i,j}(\nu_{i,j}-\nu_{i,j}')^2}{d_2}+\sum_{i,j,s,t}(\chi_{i,j,s,t}-\chi_{i,j,s,t}')^2},
$$
and the equality obtained in Section \ref{sec:independence}
\begin{align*}
||\rho_2-\sigma_2||_2^2=\sum_{s,t}(\mu_{s,t}-\mu_{s,t}')^2,
||\rho_1-\sigma_1||_2^2=\sum_{i,j}(\nu_{i,j}-\nu_{i,j}')^2.
\end{align*}
\end{proof}

In the simplest version of an $m$-qubit system, we have Proposition \ref{localqubit}:
\begin{Prop}\label{localqubit}
$m$-qubit states can be mapped into $6^m$ probabilistic distributions using $\mathcal{M}=\frac{1}{6^m}\{I+X,I-X,Y+Y,I-Y,I+Z,I-Z\}^{\otimes m}$, such that
\begin{align*}
&||p_{1,2,\dots,m}-q_{1,2,\dots,m}||_2=\frac{\sqrt{\sum_{S\subset[m]}||\rho_{S}-\sigma_{S}||_2^2}}{3^m}\geq \frac{||\rho-\sigma||_2}{3^m}\\
&||p_{S}||_2,~||q_S||_2\leq\frac{2^{|S|/2}}{3^{|S|}}
\end{align*}	
are valid for any $\rho,\sigma$ and the corresponding distributions $p,q$.

Moreover, the first equality is valid for any Hermitian matrix $\sigma$.
\end{Prop}

\section{A Streaming Algorithm for Quantum State Tomography}\label{SAQ}
Having established a connection between quantum property testing and distribution testing, we can begin to explore the potential of classical techniques with various tasks. The first of these, for us, is a streaming algorithm for quantum state tomography.
\paragraph{Restatement of Theorem \ref{tomography}}
The sample complexity of $m$-qubit quantum state tomography with a streaming algorithm is $O(\frac{18^m}{\epsilon^2})$.
\begin{proof}
Use $\mathcal{M}=\frac{1}{6^m}\{I+X,I-X,Y+Y,I-Y,I+Z,I-Z\}^{\otimes m}$ to measure the $m$-qubit quantum state $\rho$, and assume the corresponding distribution is $p$. $p$ is a $6^m$ dimensional distribution.
Moreover, $\mathcal{M}$ maps any $2^m\times 2^m$ Hermitian matrix into a subspace $P\subseteq \mathbb{R}^{6^m}$. For any $v\in P$, there is a Hermitian matrix $A$ such that $\mathcal{M}$ maps $A$ to $v$.

According to Proposition \ref{localqubit}, a good approximation of $p$ leads to a good approximation of $\rho$ as follows.

Using $n$ samples of $p$, we construct $\hat{p}$ as
\begin{align*}
\hat{p}(i)=\frac{\#(i)}{n},
\end{align*}	
where $\#(i)$ denotes the frequency of that $i$ will occur in $n$ samples.
$\#(i)=\sum_{j=1}^n B(i,j)$ is i.i.d., and $B(i,j)$ obeys the Bernoulli distribution with the parameter $p(i)$.
\begin{align*}
\mathbb{E}(||\hat{p}-p||_2^2)=\sum_{i=1}^{6^m}\frac{p(i)-p^2(i)}{n}<\frac{1}{n}.
\end{align*}
$\hat{p}$ may do not lie in $P$.
We can map $\hat{p}$ into $P\hat{p}$. Directly, we have
\begin{align*}
||P\hat{p}-p||_2^2\leq ||\hat{p}-p||_2^2.
\end{align*}
Then
\begin{align*}
\mathbb{E}(||P\hat{p}-p||_2^2)\leq \mathbb{E}(||\hat{p}-p||_2^2)<\frac{1}{n}.
\end{align*}
Let $A$ be the $2^m\times 2^m$ Hermitian matrix corresponds to $P\hat{p}$.
Now
\begin{align*}
\mathbb{E}(||A-\rho||_1)\leq \sqrt{2^m}\mathbb{E}(||A-\rho||_2)\leq  3^m\sqrt{2^m}
\sqrt{\mathbb{E}(||\hat{p}-p||_2^2)}= 3^m\sqrt{2^m}\sqrt{\sum_{i=1}^{6^m}\frac{p(i)-p^2(i)}{n}}<\sqrt{\frac{18^m}{n}}.
\end{align*}
To reach that $||A-\rho||_1\leq \epsilon$ with high probability, only $n=O(\frac{18^m}{\epsilon^2})$ samples are needed, so
we choose $\hat{\rho}$, which minimizes $||X-A||_1$ among all quantum states $X$. Therefore,
\begin{align*}
||\hat{\rho}-\rho||_1\leq ||\hat{\rho}-A||_1+||A-\rho||_1\leq 2||A-\rho||_1\leq 2\epsilon.
\end{align*}
According to $\tr A=1$, we can let $\hat{\rho}=\lambda A^+$ where $\tr \hat{\rho}=1$ and $A=A^+-A^-$ such that $A^+,A^-\geq 0$ and $A^+A^-=0$.
\end{proof}
This is the first streaming quantum algorithm for quantum state tomography in the realm of quantum property testing, given the mentioned mistake-prone fact that Pauli measurements can not be preformed in the local measurement model \cite{FlammiaGrossLiuEtAl2012}.

The advantage of local measurement can be illustrated as the following example. For a general three qubit state $\rho_{1,2,3}$, our local measurement method consists of one-qubit measurements on each qubit.
In this sense, if we want to obtain correlation informaiton about $\rho_{1,2}$, our measurement does not corrupt the correlation information about $\rho_{1,3}$ very much in expectation. Formally, 
the local measurement structure enables us to derive the $k$-local tomography whose goal is to output good estimation of all $k$-qubit reduced density matrices with good precision with high probability.
\paragraph{Restatement of Theorem \ref{localtomography}}
The sample complexity of $k$-local tomography of $m$-qubit quantum state with a streaming algorithm is $O(\frac{108^k(\log{{{m}\choose{k}}}+k\log 6)}{\epsilon^2})$. For constant $k$, it is $\Theta(\frac{\log m}{\epsilon^2})$.
\begin{proof}
Our measurement scheme maps $k$-qubit state $\rho$ into a $6^k$ dimensional probability distribution $p=(p(1),p(2),\cdots, p({6^k}))$.
With $n$ copies of $\rho$, we can have $n$ i.i.d. samples of $p$ and suppose the empirical distribution is $\hat{p}=(\hat{p}(1),\hat{p}(2),\cdots, \hat{p}({6^k}))$.
According to Chernoff bound,
\begin{align*}
Pr(|p(i)-\hat{p}(i)|>\delta)<2\exp(-n\delta^2/2)
\end{align*}
Then with probability at least 
\begin{align*}
1-2\times 6^k\exp(-n\delta^2/2)
\end{align*}
we have that for any $i$, $|p(i)-\hat{p}(i)|<\delta$.
With probability at least $1-2\times 6^k\exp(-n\delta^2/2)$, we have
\begin{align*}
||p-\hat{p}||_2\leq 6^k\delta^2.
\end{align*}
In this case, according to the analysis of Theorem \ref{tomography}, one can output a $\hat{\rho}$ such that $||\rho-\hat{\rho}||_1\leq 6^k\sqrt{3}^k\delta$.

To have good estimations of all ${{m}\choose{k}}$, we only need to do the local measurement on each qubit for each copy. After obtaining the sample of the total $6^m$ dimensional distribution, we can compute the induced sample of each $6^k$ dimensional distribution which corresponds to each $k$-qubit state. The rest is a union bound which says that the probability of having a $\epsilon$ estimation,  by setting $\epsilon=6^k\sqrt{3}^k\delta$, of all $k$-reduced density matrices is at least
\begin{align*}
1-{{m}\choose{k}}2\times 6^k\exp(-n\delta^2/2).
\end{align*}
Let $n=O(\frac{108^k(\log{{{m}\choose{k}}}+k\log 6)}{\epsilon^2})$, we can have the above probability greater than $2/3$.

For constant $k$, it becomes $O(\frac{\log m}{\epsilon^2})$. To show this is tight, we only need to deal with
classical distribution $p=p_1\otimes p_2\otimes\cdots \otimes p_m$ where each $p_i=(\frac{1}{2}+\epsilon,\frac{1}{2}-\epsilon)$ or $p_i=(\frac{1}{2}-\epsilon,\frac{1}{2}+\epsilon)$ and the goal is to obtain $\epsilon$ estimation of each $p_i$. 

According to \cite{Mousavi2016}, $\Omega(\frac{1}{\epsilon^2}\log(\frac{1}{q}))$ samples are needed to achieve confidence at least $1-q$. Therefore, we require that 
\begin{align*}
(1-q)^m>\frac{2}{3}.
\end{align*}
That is $q=O(\frac{1}{m})$, that implies the bound $\Omega(\frac{\log m}{\epsilon^2})$.
\end{proof}
\section{Quantum State Certification} \label{QSC}
The connections developed in Section \ref{sec:independence}, together with the $\ell_2$-identity tester of probability distributions provided in \cite{Chan:2014:OAT:2634074.2634162}, also make efficient identity testing of quantum states possible.
\subsection{Independent measurement}
\paragraph{Restatement of Theorem \ref{iit} }
To identify $\rho,\sigma\in\mathcal{D}(\C^d)$ via local measurement, $O(\frac{d^2}{\epsilon^2})$ copies are sufficient to distinguish, with at least a $\frac{2}{3}$ probability, cases where $\rho=\sigma$ from cases where $||\rho-\sigma||_1\geq \epsilon$.

\begin{proof}
First map the state into probability distributions, say $p$ and $q$, through independent measurement with Theorem \ref{indepent-measurement-classical}, and follow by executing Algorithm \ref{algo:Identity-Test-Independent-Measurements}

\begin{algorithm}[H]
	\Input{$O(\frac{d^2}{\epsilon^2})$ copies of $\rho\in \mathcal{D}(\C^{d})$ and $O(\frac{d}{\epsilon^2})$ copies of $\sigma\in \mathcal{D}(\C^{d})$}	
    \Output{"Yes" with a probability of at least $\frac{2}{3}$ if $\rho=\sigma$; and "No" with a probability of at least $\frac{2}{3}$ if $||\rho-\sigma||_1>\epsilon$.}
    Run Algorithm \ref{algo:2-norm-Identity-Test} to distinguish between $p=q$ and $||p-q||_2\geq \frac{\epsilon}{\sqrt{d}(d+1)}$\;\tcc{$p$ and $q$ are the probability distributions obtained by measuring $\rho$ and $\sigma$ through the independent measurement with Theorem \ref{indepent-measurement-classical}, respectively.}
	\caption{\textsf{A Identity Test with Independent Measurement}}
	\label{algo:Identity-Test-Independent-Measurements}
\end{algorithm}

According to $||p-q||_2=\frac{||\rho-\sigma||_2}{d+1}$, we only need to distinguish cases where $p=q$ from cases where $||p-q||_2\geq \frac{||\rho-\sigma||_1}{\sqrt{d}(d+1)}\geq\frac{\epsilon}{\sqrt{d}(d+1)}$.
Choosing $b=\frac{\sqrt{2}}{d+1}\geq ||p||_2,||q||_2$ and invoking Theorem \ref{classical}, we have
$$O(\frac{b}{(\frac{\epsilon}{\sqrt{d}(d+1)})^2})=O(\frac{d^2}{\epsilon^2})$$
which is a sufficient number of copies.
\end{proof}
If we let $d=2^m$, the sample complexity is $O(\frac{d^{1.5+\log 3}}{\epsilon^2})$.

According to \cite{HHJ+16}, the sample complexity for tomography is $\rho\in\mathcal{D}(\C^d)$ is $\Theta(\frac{d^3}{\epsilon^2})$, which makes Algorithm \ref{algo:Identity-Test-Independent-Measurements} a better choice for identity testing after tomography.

As mentioned in the introduction, Algorithm \ref{algo:Identity-Test-Independent-Measurements} should be significantly easier to implement because it does not demand noiseless, universal quantum
computation with an \textit{exponential} number of qubits.

\subsection{Local measurement}

\paragraph{Restatement of Theorem \ref{lcoaliit}}
For $m$-qubit quantum states $\rho,\sigma$, $O(\frac{(6\sqrt{2})^m}{\epsilon^2})$ copies are sufficient to distinguish via local measurement, with at least a $\frac{2}{3}$ probability, cases where $\rho=\sigma$ from cases where $||\rho-\sigma||_1\geq \epsilon$.
\begin{proof}
First map the state into probability distributions, say $p$ and $q$, through local measurement with Theorem \ref{localqubit}, and follow by executing Algorithm \ref{algo:Identity-Test-Streaming-Algorithm},

\begin{algorithm}
	\Input{$O(\frac{(6\sqrt{2})^m}{\epsilon^2})$ copies of $\rho$ and $O(\frac{(6\sqrt{2})^m}{\epsilon^2})$  copies of $\sigma$}	
    \Output{"Yes" with a probability of at least $\frac{2}{3}$ if $\rho=\sigma$; and "No" with a probability of at least $\frac{2}{3}$ if $||\rho-\sigma||_1>\epsilon$.}
    Run Algorithm \ref{algo:2-norm-Identity-Test} to distinguish between $p=q$ and $||p-q||_2\geq \frac{\epsilon}{2^{m/2}3^m}$\;\tcc{$p$ and $q$ are the probability distributions obtained by measuring $\rho$ and $\sigma$ with the local measurement in Proposition \ref{localqubit}, respectively.}
	\caption{\textsf{A Identity Test Streaming Algorithm}}
	\label{algo:Identity-Test-Streaming-Algorithm}
\end{algorithm}

According to $||p-q||_2\geq \frac{||\rho-\sigma||_2}{3^m}$, we only need to distinguish cases that  $p=q$ from cases where $||p-q||_2\geq \frac{||\rho-\sigma||_1}{\sqrt{d}3^m}\geq\frac{\epsilon}{2^{m/2}3^m}$.
Choosing $b=\frac{2^{m/2}}{3^m}\geq ||p||_2,||q||_2$ and invoking Theorem \ref{classical}, we have
$$O(\frac{b}{(\frac{\epsilon}{2^{m/2}3^m})^2})=O(\frac{(6\sqrt{2})^m}{\epsilon^2})$$
which is a sufficient number of copies.
\end{proof}

This is significantly less complex than the streaming tomography Algorithm \ref{tomography} given in Section \ref{SAQ}.

\subsection{Comparison with a quantum Swap test}

The quantum swap test is widely used in identity testing for quantum states in the $\ell_2$ distance. The swap operator $S$ is a unitary that for any $1\leq i,j\leq d$,
$$
S\ket{i}\ket{j}=\ket{j}\ket{i}.
$$
From direct observation, we always have
$$
\tr[S(\rho\otimes\sigma)]=\tr(\rho\sigma),
$$
for all $\rho,\sigma\in\mathcal{D}(\C^d)$, and
by employing this operator, an observable $M$ can be constructed such that $|M|\leq 10 I$
\begin{align*}
\tr[M(\rho\otimes\rho\otimes\sigma\otimes\sigma)]=||\rho-\sigma||_2^2.
\end{align*}
This follows from letting
\begin{align*}
M=O_1-2O_2+O_3,
\end{align*}
where the bounded operators $O_1,O_2,O_3$ are
\begin{align*}
\tr[O_1(\rho\otimes\rho\otimes\sigma\otimes\sigma)]=&\tr(\rho^2),\\
\tr[O_2(\rho\otimes\rho\otimes\sigma\otimes\sigma)]=&\tr[\rho\sigma],\\
\tr[O_3(\rho\otimes\rho\otimes\sigma\otimes\sigma)]=&\tr(\sigma^2].
\end{align*}
The joint measurement is then
$$
M_1=\frac{10I+O}{20},~~M_2=\frac{10I-O}{20}.
$$
Applying this measurement to $\rho\otimes\rho\otimes\sigma\otimes\sigma$, the resulting probability is
$$p(1)=\frac{1}{2}+\frac{||\rho-\sigma||_2^2}{20},~p(2)=\frac{1}{2}-\frac{||\rho-\sigma||_2^2}{20}.$$

Using this measurement to distinguish cases where $\rho=\sigma$ from cases where $||\sigma-\rho||_2\geq\epsilon$, with high confidence, between the two probability distributions $p$ and $q$,  where
\begin{align*}
p(1)&=\frac{1}{2}, &p(2)=\frac{1}{2}.\\
q(1)&=\frac{1}{2}+\frac{\epsilon^2}{20},&q(2)=\frac{1}{2}-\frac{\epsilon^2}{20}.
\end{align*}
According to basic statistics, $O(\frac{1}{\epsilon^4})$ copies are needed.

In testing the $\ell_1$ norm with this approach, the complexity of distinguishing cases where $\rho=\sigma$ from cases where $||\sigma-\rho||_1\geq\epsilon$ through the $\ell_2$ bound, becomes $O(\frac{d^2}{\epsilon^4})$.

Compare with this Swap test based algorithm, Algorithm \ref{algo:Identity-Test-Independent-Measurements} uses fewer copies and is easier to implement.
For fixed constant $d$, Algorithm \ref{algo:Identity-Test-Streaming-Algorithm} use fewer copies and is much easier to implement.

\section{Testing Independence}\label{TI}
The goal of independence testing is to determine whether a fixed multipartite state $\rho$ is independent, i.e., in tensor product form, or far from being independent. Hence, in this section, we outline a series of testing algorithms and almost matching lower bounds in joint measurement setting, independent measurement, and in a streaming fashion. We begin with the bipartite independence testing and then generalize to multipartite independence testing.

\subsection{Testing bipartite independence}

This subsection presents the algorithms for testing bipartite independence with joint measurement, independent measurement, and local measurement, in that order, and concludes with a proof of the matching lower bounds in the joint measurement setting.
\begin{lemma}\label{thm1}
For a fixed $\rho \in \mathcal{D}(\C^{d_1}\otimes \C^{d_2})$, the sample complexity of independence testing,  i.e., of the form $\rho_1\otimes\rho_2$, or $||\rho-\sigma||_1>\epsilon$ for any independent $\sigma$, is
\begin{itemize}
\item $O(\frac{d_1d_2}{\epsilon^2})$ with joint measurement;
\item $O(\frac{d_1^2d_2^2}{\epsilon^2})$ with independent measurement; and
\item $O(\frac{d_1^{1.5+\log 3}d_2^{1.5+\log 3}}{\epsilon^2})$ with a streaming algorithm.
\end{itemize}
\end{lemma}
\begin{proof}
The algorithm for joint measurement follows

\begin{algorithm}[H]
	\Input{$n=O(\frac{d_1d_2}{\epsilon^2})$ copies of $\rho\in\mathcal{D}(\C^{d_1}\otimes \C^{d_2})$}	
    \Output{"Yes" with a probability of at least $\frac{2}{3}$ if $\rho$ is independent; and "No" with a probability of at least $\frac{2}{3}$ if $||\rho-\sigma||_1>\epsilon$ for any independent $\sigma$.}
	Use $\frac{n}{3}$ copies of $\rho$ to generate $\rho_1$\;\tcc{Trace out system 2}
	Use $\frac{n}{3}$ copies of $\rho$ to generate $\rho_2$\;\tcc{Trace out system 1}
    Run Algorithm \ref{algo:Identity-Test-Joint-Measurements} on $\frac{n}{3}$ copies of $\rho$ and $\frac{n}{3}$ copies of $\rho_1\otimes\rho_2$ with the parameter $\epsilon/3$\;
	\caption{\textsf{A Bipartite Independence Testing with Joint Measurement}}
	\label{algo:Bipartite-Independent-testing-joint-measurements}
\end{algorithm}

The correctness of these algorithm accords with Theorem \ref{BOW} by note that
\begin{itemize}
\item If $\rho$ is independent, then $\rho=\rho_1\otimes\rho_2$, and this algorithm will output "Yes" with high probability.
\item If $||\rho-\sigma||_1>\epsilon$ for any independent $\sigma$, then $||\rho-\rho_1\otimes\rho_2||_1>\epsilon/3$ by Proposition \ref{dis1}, and this algorithm will output "No" with high probability.
\end{itemize}
We can derive an independent measurement tester by replacing the identity tester in Algorithm \ref{algo:Identity-Test-Joint-Measurements} with Algorithm  \ref{algo:Identity-Test-Independent-Measurements}. From a similar analysis to the above, we have
$$O(\frac{d_1^2 d_2^2}{\epsilon^2})$$
which is a sufficient number of copies.

We can derive a local measurement tester by replacing the identity tester in Algorithm \ref{algo:Identity-Test-Joint-Measurements} with Algorithm \ref{algo:Identity-Test-Streaming-Algorithm}. From a similar analysis, we have
$$O(\frac{d_1^{1.5+\log 3}d_2^{1.5+\log 3}}{\epsilon^2})$$
which is a sufficient number of copies.
\end{proof}
Next, we provide a matching bound with joint measurement, up to a $\mathrm{poly}\log$ factor.

\begin{lemma}\label{thmb}
Let $\rho \in \mathcal{D}(\C^{d_1}\otimes \C^{d_2})$ with $d_1\geq d_2$. $\Omega(\frac{d_1d_2}{\epsilon^2})$ copies are necessary to test whether $\rho$ is independent or $\epsilon$-far from being independent in term of $\ell_1$ distance when $d_2>2000$; otherwise, if $d_2\leq2000$, $\Omega(\frac{d_1d_2}{\epsilon^2\log^3 d_1\log\log d_1})$ copies are necessary.
\end{lemma}
In cases where $d_1$ and $d_2$ are both very large, the bound is derived from the mixness test of Theorem \ref{mixness} in \cite{OW15}, where the constant $2000$ comes from the upper and lower bound of the constant in that theorem. To deal with "unbalanced" cases where
only $d_1$ or $d_2$ is small--here, let us says $d_2$--we split the $d_1$ system into many systems of dimension $d_2$, which transforms the original unbalance of bipartite problem into a problem of "balanced" multipartite independence testing. Then, we use Proposition \ref{dis2m}.

\begin{proof}
First, note that it suffices to consider cases where $d_1d_2$ are sufficiently large since $\Omega(\frac{1}{\epsilon^2})$ samples are
required to distinguish the two classical distributions, i.e., a $[2] \times [2]$ uniform distribution
from the distribution
$(\frac{1+2\epsilon}{4},\frac{1-2\epsilon}{4},\frac{1-2\epsilon}{4},\frac{1+2\epsilon}{4})$.

To show the lower bound for a general $d_1$ and $d_2$, assume there is an algorithm, Algorithm A, that uses $f(d_1,d_2,\epsilon)$ copies to decide whether a given $\rho\in\mathcal{D}(\C^{d_1}\otimes \C^{d_2})$ is independent or $\epsilon$-far from being independent with at least a $2/3$ probability of successful. By using Algorithm A as an oracle, the following algorithm can distinguish cases where $\rho=\frac{I_{d_1}}{d_1}\otimes \frac{I_{d_2}}{d_2}$ from cases where $||\rho-\frac{I_{d_1}}{d_1}\otimes \frac{I_{d_2}}{d_2}||_1>\epsilon$ for any $t>1$.

\begin{algorithm}[H]
	\Input{$n=100f(d_1,d_2,\frac{(t-1)\epsilon}{4t})+300t^2\frac{d_1}{\epsilon^2}+\Theta(\frac{d_2}{t^2(t-1)^2\epsilon^2})$ copies of $\rho\in\mathcal{D}(\C^{d_1}\otimes \C^{d_2})$}	
    \Output{"Yes" with a probability of at least $\frac{2}{3}$ if $\rho=\frac{I_{d_1}}{d_1}\otimes\frac{I_{d_2}}{d_2}$; and "No" with a probability of at least $\frac{2}{3}$ if $||\rho-\frac{I_{d_1}}{d_1}\otimes\frac{I_{d_2}}{d_2}||_1>\epsilon$.}
		Repeat Algorithm \ref{mix-OW}, with $100t^2\frac{d_1}{\epsilon^2}$ copies of $\rho$, three times to test whether $\rho_1=\frac{I_{d_1}}{d_1}$ or $||\rho_1-\frac{I_{d_1}}{d_1}||_1>\epsilon/t$ with at least a $\frac{20}{27}$ probability of success\;
\If{"No"}
{
Return ``No"\;
}
\Else{
Employ Algorithm \ref{mix-OW} with $\Theta(\frac{t^2d_2}{(t-1)^2\epsilon^2})$ copies of $\rho$ to test whether $\rho_2=\frac{I_{d_2}}{d_2}$ or $||\rho_1-\frac{I_{d_1}}{d_1}||_1>\frac{(t-1)\epsilon}{4t}$ with
at least a $\frac{27}{28}$ probability of success\;
\If{"No"}
{
Return ``No"\;
}
\Else{
    Run Algorithm A 100 times to test whether $\rho$ is independent or is $\frac{(t-1)\epsilon}{4t}$-far from being independent with
   at least a $\frac{28}{30}$ probability of success\;
    \If{"Yes"}
    {
    Return ``Yes"\;
    }
    \Else{
    Return ``No"\;
    }
    }
}	
	\caption{\textsf{A Bipartite Identity test A for a maximally mixed state}}
	\label{algo:Bipartite Identity test A for maximally mixed state}
\end{algorithm}
To see this algorithm to succeed at detecting whether $\rho$ is maximally mixed with high probability, note that

If $\rho=\frac{I_{d_1}}{d_1}\otimes \frac{I_{d_2}}{d_2}$, in Line 1, the algorithm will output $\rho_1=\frac{I_{d_1}}{d_1}$ with a probability of at least $\frac{20}{27}$; in Line 5, the algorithm will output $\rho_2=\frac{I_{d_2}}{d_2}$ with a probability of at least ${\frac{27}{28}}$; in Line 9, $\rho$ will be independent with a probability of at least ${\frac{28}{30}}$. Overall, Algorithm \ref{algo:Bipartite Identity test A for maximally mixed state} will output "Yes" with a probability of at least $\frac{2}{3}$.

If $||\rho-\frac{I_{d_1}}{d_1}\otimes \frac{I_{d_2}}{d_2}||_1>\epsilon$, then one of the following three statements will be true:
$\rho_1$ is $\epsilon/t$-far from $\frac{I_{d_1}}{d_1}$; or $\rho_2$ is $\frac{(t-1)\epsilon}{4t}$-far from $\frac{I_{d_2}}{d_2}$; or $\rho$ is $\frac{(t-1)\epsilon}{4t}$-far from being independent. Otherwise, assume that there exists an $\sigma_1$ and an $\sigma_2$, such that
    $||\rho-\sigma_1\otimes\sigma_1||_1<\frac{(t-1)\epsilon}{4t}$, $||\rho_1-\frac{I_{d_1}}{d_1}||_1<\frac{\epsilon}{t}$ and $||\rho_2-\frac{I_{d_2}}{d_2}||_1<\frac{(t-1)\epsilon}{4t}$. According to Proposition \ref{dis1}, we have $||\rho-\rho_1\otimes\rho_2||_1<\frac{3(t-1)\epsilon}{4t}$. Then by the triangle inequality and Lemma \ref{dis2}, we have
    $$||\rho-\frac{I_{d_1}}{d_1}\otimes \frac{I_{d_2}}{d_2}||_1\leq ||\rho-\rho_1\otimes\rho_1||_1+||\frac{I_{d_1}}{d_1}\otimes \frac{I_{d_2}}{d_2}-\rho_1\otimes\rho_1||_1< \frac{3(t-1)\epsilon}{4t}+\frac{\epsilon}{t}+\frac{(t-1)\epsilon}{4t}=\epsilon.$$
    Contradiction!
Therefore, in this case, the algorithm outputs "No" with a probability of at least $\min\{\frac{20}{27},\frac{27}{28},\frac{28}{30}\}>\frac{2}{3}$.

This algorithm uses $n=100f(d_1,d_2,\frac{(t-1)\epsilon}{4t})+300t^2\frac{d_1}{\epsilon^2}+\Theta(\frac{t^2d_2}{(t-1)^2\epsilon^2})$ copies of $\rho$. Invoking Theorem \ref{mixness}, we know that $0.15\frac{d_1d_2}{\epsilon^2}$ copies are necessary to test, with at least a 2/3 probability of success, whether $\rho$ is the
maximally mixed state or whether it is $\epsilon$-far.

We must have
$$100f(d_1,d_2,\frac{(t-1)\epsilon}{4t})+300t^2\frac{d_1}{\epsilon^2}+\Theta(t^2\frac{d_2}{(t-1)^2\epsilon^2})\geq 0.15\frac{d_1d_2}{\epsilon^2}.$$
If $d_1$ and $d_2$ are both sufficiently large, we can choose a constant $t$ such that $300t^2\frac{d_1}{\epsilon^2}+\Theta(t^2\frac{d_2}{(t-1)^2\epsilon^2})=o(\frac{d_1d_2}{\epsilon^2})$, which implies
$$f(d_1,d_2,\epsilon)\geq \Omega(\frac{16t^2 d_1d_2}{(t-1)^2\epsilon^2})=\Omega(\frac{d_1d_2}{\epsilon^2}).$$
If $d_1$ is sufficiently large and $d_2$ is not sufficiently large but $d_2>2000$, we can choose $t=\sqrt{\frac{2000.5}{2000}}$, then
$$f(d_1,d_2,c\epsilon)\geq 0.15\frac{d_1d_2}{\epsilon^2}-300t^2\frac{d_1}{\epsilon^2}+\Omega(t^2\frac{d_2}{(t-1)^2\epsilon^2})=\Omega(\frac{d_1}{\epsilon^2})=\Omega(\frac{d_1d_2}{\epsilon^2}),$$
with the constant $c=\frac{t-1}{4t}$.
Thus, for $d_2>2000$,
$$f(d_1,d_2,\epsilon)\geq \Omega(\frac{d_1d_2}{\epsilon^2}).$$
The above technique does not work with a small $d_2$, because the number of copies required to test a $d_1$ system $300t^2\frac{d_1}{\epsilon^2}$ and the number of copies required to test a total system of $0.15\frac{d_1d_2}{\epsilon^2}$ are of the same order.

To deal with this unbalanced case, we develped a dimension splitting technique that transforms bipartite independence into $k$-partite independence. First observe that the sample complexity for independence testing in $\mathcal{D}(\C^{d_1}\otimes \C^{d_2})$ is no less than the sample complexity for an independence test of $\mathcal{D}(\C^{d}\otimes \C^{2})$ for $d=2^{[\log d_1]}\leq d_1$. Therefore, without loss of generality, assume that $d_2=2$ and $d_1=2^k$ instead of $d_2\leq 2000$, and that $d_1$ is sufficiently large.

We still assume that there is an Algorithm A that uses $f(2^k,2,\epsilon)$ copies to decide, with at least a $2/3$ probability of success, whether a given $\rho\in \mathcal{D}(\C^{2^k\times 2^k}\otimes \C^{2\times 2})$ is independent or $\epsilon$-far from independent in the $2^k$ and $2$ bipartitions.
Any such $\rho$ can be regarded as a $k+1$ qubit state, and the qubit systems will be labeled as $S=\{1,2,\dots,k,k+1\}$. $\rho_i$ denotes the reduced density matrix of the $i$-th qubit of $\rho$. Algorithm A is a bipartite independence tester for a $k+1$ qubit system in the bipartition of $k$ qubits and $1$ qubit. In the following, Algorithm A is applied as a black box to the bipartition $i$ and $S \setminus \{i\}$  for any $i$ to test the identity of $\rho$ and $\frac{I_{d_1}}{d_1}\otimes \frac{I_{d_2}}{d_2}$.

\begin{algorithm}[H]
	\Input{$n=\Theta[(k+1)\log k f(2^k,2,\frac{\epsilon}{6(k+1)})]+\Theta[(k+1)\log k\frac{(k+1)^2}{\epsilon^2}]$ copies of $\rho$.}	
\Output{"Yes" with a probability of at least $\frac{2}{3}$ if $\rho=\otimes_{i=1}^{k+1}\otimes \frac{I_{2}}{2}$; and "No" with a probability of at least $\frac{2}{3}$ if $||\rho-\otimes_{i=1}^{k+1}\otimes \frac{I_{2}}{2}||_1>\epsilon$.}
    \For{$i\gets1$ \KwTo $k+1$}{
    Repeat Algorithm \ref{mix-OW}, with $\Theta(\frac{(k+1)^2}{\epsilon^2})$ copies of $\rho$ each time, $\Theta(\log k)$ times to test whether $\rho_i=\frac{I_{2}}{2}$ or $||\rho_i-\frac{I_{2}}{2}||_1>\frac{\epsilon}{6(k+1)}$ 	at least a $1-\frac{1}{k^2}$ probability of success\;
    \If{No}{
Return "No"\;
}
\Else{
	Run Algorithm A $\Theta(\log k)$ times, with $f(2^k,2,\frac{\epsilon}{6(k+1)})$ copies each time, to test whether $\rho$ is independent or $\frac{\epsilon}{6(k+1)}$-far from being independent in the bipartition $\{i\}$ and $S\setminus\{i\}$ with at least a $1-\frac{1}{k^2}$ probability of success\;
\If{No}
{
Return "No"\;
}
}
    }
    Return "Yes"\;
	\caption{\textsf{A Bipartite Identity Test B for a maximally mixed state}}
	\label{algo:Bipartite Identity test B for maximally mixed state}
\end{algorithm}

To see this algorithm succeed in detecting whether $\rho$ is maximally mixed with high probability, we note that

 If $\rho=\frac{I_{d_1}}{d_1}\otimes \frac{I_{d_2}}{d_2}$, then $\rho_i=\frac{I_2}{2}$ when $\rho$ is regarded as a $k+1$ qubit state. It is independent in any bipartition $\{i\}$ and $S\setminus\{i\}$.
    For each $i$, the passing probability of the test $\rho_i=\frac{I_2}{2}$ is at least $1-\frac{1}{k^2}$. For each $i$, the passing probability of the independence test in the bipartition $\{i\}$ and $S\setminus\{i\}$ is at least $1-\frac{1}{k^2}$.
    In total, Algorithm \ref{algo:Bipartite Identity test B for maximally mixed state} will accept with a probability of at least $(1-\frac{1}{k^2})^{O(k)}=1-o(1)>\frac{2}{3}$.

If $||\rho-\frac{I_{d_1}}{d_1}\otimes \frac{I_{d_2}}{d_2}||_1>\epsilon$, at least one of the following two statements is true:
\begin{itemize}
\item $||\rho_i-\frac{I_2}{2}||_1>\frac{\epsilon}{6(k+1)}$ for some $1\leq i\leq k+1$; and/or
\item $\rho$ is $\frac{\epsilon}{6(k+1)}$-far from independent in the bipartition $\{i\}$ and $S\setminus\{i\}$ for some $1\leq i\leq k+1$.
\end{itemize}
Otherwise, $||\rho_i-\frac{I_2}{2}||_1\leq\frac{\epsilon}{6(k+1)}$ and $\rho$ is $\frac{\epsilon}{6(k+1)}$ close to being independent in the bipartition $\{i\}$ and $S\setminus\{i\}$ for all $1\leq i\leq k+1$.

According to Proposition \ref{dis2m}, we have $$||\rho-\rho_1\otimes\rho_2\otimes\cdots\otimes\rho_{k+1}||_1\leq 5(k+1)\frac{\epsilon}{6(k+1)}=\frac{5\epsilon}{6}.$$
By Lemma \ref{dism1}, we have
\begin{eqnarray*}
&&||\rho-\frac{I_{d_1}}{d_1}\otimes \frac{I_{d_2}}{d_2}||_1\\
&=&||\rho-\frac{I_{2}}{2}\otimes \frac{I_{2}}{2}\otimes\cdots\otimes\frac{I_{2}}{2}||_1\\
&\leq& ||\rho-\rho_1\otimes\rho_2\otimes\cdots\otimes\rho_{k+1}||_1+||\frac{I_{2}}{2}\otimes \frac{I_{2}}{2}\otimes\cdots\otimes\frac{I_{2}}{2}-\rho_1\otimes\rho_2\otimes\cdots\otimes\rho_{k+1}||_1\\
&\leq & \frac{5\epsilon}{6}+\sum_{i=1}^{k+1}||\rho_i-\frac{I_2}{2}||_1\\
&\leq& \epsilon.
\end{eqnarray*}
    Contradiction!
Therefore, the algorithm outputs "No" with a probability of at least $1-\frac{1}{k^2}>\frac{2}{3}$ in this case.

Invoking Theorem \ref{mixness}, we know that $\Theta(\frac{d_1d_2}{\epsilon^2})=\Theta(\frac{2^{k+1}}{\epsilon^2})$ copies are necessary to test whether $\rho$ is a
maximally mixed state or $\epsilon$-far with at least a 2/3 probability of success.
Algorithm \ref{algo:Bipartite Identity test B for maximally mixed state} uses $\Theta[(k+1)\log k f(2^k,2,\frac{\epsilon}{6(k+1)})]+\Theta[(k+1)\log k\frac{(k+1)^2}{\epsilon^2}]$ copies of $\rho$.
We must have
\begin{eqnarray*}
&&\Theta[(k+1)\log k f(2^k,2,\frac{\epsilon}{6(k+1)})]+\Theta[(k+1)\log k\frac{(k+1)^2}{\epsilon^2}]\geq \Theta(\frac{2^{k+1}}{\epsilon^2})\\
&\Rightarrow& f(2^k,2,\frac{\epsilon}{6(k+1)})\geq \Theta(\frac{2^{k}}{k\log k\epsilon^2})\\
&\Rightarrow& f(2^k,2,\epsilon)\geq \Theta(\frac{2^{k}}{k^3\log k\epsilon^2})\\
&\Rightarrow& f(d_1,d_2,\epsilon)= \Omega(\frac{d_1}{\log^3 d_1\log\log d_1\epsilon^2})=\Omega(\frac{d_1d_2}{\log^3 d_2\log\log d_1\epsilon^2})
\end{eqnarray*}
That is, if $d_2$ is a small constant, $\Omega(\frac{d_1d_2}{\log^3 d_1\log\log d_1\epsilon^2})$ copies are necessary to test the independence of $\rho\in\mathcal{D}(\C^{d_1}\otimes \C^{d_2})$.

\end{proof}

\subsection{Multipartite independence testing}

In this subsection, we generalize the results of the bipartite independence testing in the previous subsection to multipartite independence testing.

The obvious generalization of the bipartite independence testing to $m$-partite would work using bipartite independence in any $m-1$ parties versus $1$ party. Our goal is to test independence in this scenario with an accuracy of $O(\frac{\epsilon}{m})$ and at least a $1-\frac{1}{m^2}$ probability of success. The correctness of the algorithm follows from Proposition \ref{dis2m}, and the generalization incurs an $O(m^3\log m)$ factor. For constant $m$, $O(m^3\log m)$ is still constant. Thus, the complexity of the different algorithm variants would be $O(\frac{\Pi_{i=1}^m d_i}{\epsilon^2})$ with joint measurement, $O(\frac{\Pi_{i=1}^m d_i^2}{\epsilon^2})$ with independent measurement, and $O(\frac{\Pi_{i=1}^m d_i^{1.5+\log 3}}{\epsilon^2})$ with a streaming algorithm where all $d_i$ are to the power of $2$.

With a super-constant $m$, algorithms could be built that achieve the same complexity using Diakonikolas and Kane’s \cite{DK16} recursion idea coupled with our previous bipartite independence tester.

\paragraph{Restatement of the upper bound part in \ref{thm2}}
The sample complexity of independent testing for $\mathcal{D}(\C^{d_1}\otimes \C^{d_2}\otimes\cdots\otimes\C^{d_m})$, i.e., distinguishing, with at least a $\frac{2}{3}$ probability of success, the cases where $\rho$ is in the tensor product form $\rho_1\otimes\rho_2\otimes\cdots\otimes\rho_m$, or $||\rho-\sigma||_1>\epsilon$ for any tensor product $\sigma$ is:
\begin{itemize}
\item $O(\frac{\Pi_{i=1}^m d_i}{\epsilon^2})$ with joint measurement;
\item $O(\frac{\Pi_{i=1}^m d_i^2}{\epsilon^2})$ with independent measurement; and
\item $O(\frac{\Pi_{i=1}^m d_i^{1.5+\log 3}}{\epsilon^2})$ with a streaming algorithm where all $d_i$ are to the power of $2$.
\end{itemize}

\begin{proof}
Beginning with the joint measurement variant, we can assume that all $d_i \geq 2$, otherwise removing that term does not affect the problem. According to Theorem \ref{thm1} and the discussion above, if $m<100$, we know that, in a joint measurement setting, there exists a sufficiently large $C>0$ and an algorithm for testing $m$-partite quantum independence, with at least a $\frac{2}{3}$ probability of success, using $C\frac{\Pi_{i=1}^m d_i}{\epsilon^2}$ copies given an $\ell_1$ distance parameter of $\epsilon>0$.
Therefore, we prove that, for any $\epsilon>0$, there exists an algorithm for testing $m$-partite quantum independence using $1002C\frac{\Pi_{i=1}^m d_i}{\epsilon^2}$ copies given an $\ell_1$ distance parameter of $\epsilon>0$ with at least a $\frac{2}{3}$ probability of success by induction. This statement is true for $m<100$. Now, suppose this statement is true for $m\leq k$, and we
can derive Algorithm \ref{algo:$m$-partite Independence testing} for $k<m<2k+1$.

\begin{algorithm}[H]
	\Input{$1002C\frac{\Pi_{i=1}^m d_i}{\epsilon^2}$ copies for joint measurement of $\rho\in \mathcal{D}(\C^{d_1}\otimes \C^{d_2}\otimes\cdots\otimes\C^{d_m})$.}	
\Output{"Yes" with a probability of at least $\frac{2}{3}$ if $\rho=\otimes_{i=1}^{m}\rho_i$; and "No" with a probability of at least $\frac{2}{3}$ if $||\rho-\sigma||_1>\epsilon$ for any independent $\sigma$.}
	We first partition $[m]$ into two sets $S_1=\{1,2,\dots,[\frac{m}{2}]\}$ and $S_2=\{[\frac{m}{2}]+1,\dots,m\}$.
	Call Algorithm \ref{algo:Bipartite-Independent-testing-joint-measurements} $40$ times with $25C\frac{\Pi_{i=1}^m d_i}{\epsilon^2}$ copies of $\rho$ each time for joint measurement to test the independence of $\rho$ in the bipartition $S_1$ and $S_2$ given an $\ell_1$ distance parameter of $\epsilon/5$ with at least a $\sqrt[3]{\frac{2}{3}}$ probability of successful\;

    \If{No}{
    Return ''No";
    }
    \Else{
    Call Algorithm $[\frac{m}{2}]$-partite independence testing $2^{50}/25$ times with $25C\frac{\Pi_{i=1}^{[\frac{m}{2}]} d_i}{\epsilon^2}$ copies of $\rho$ each time with joint measurement to test the $[\frac{m}{2}]$-partite independence of $\rho_{S_1}$ given an $\ell_1$ distance parameter $\epsilon/5$ with at least a $\sqrt[3]{\frac{2}{3}}$ probability of successful\;
    \If{No}{
    Return ''No";
    }

    \Else{
    Call Algorithm $m-[\frac{m}{2}]$-partite independence testing $2^{50}/25$ times with $25C\frac{\Pi_{i=m-[\frac{m}{2}]}^{m} d_i}{\epsilon^2}$ copies of $\rho$ each time with joint measurement to test the $m-[\frac{m}{2}]$-partite independence of $\rho_{S_2}$ given an $\ell_1$ distance parameter of $\epsilon/5$ with at least a $\sqrt[3]{\frac{2}{3}}$ probability of successful\;
    \If{No}{
    Return ''No";
    }
    \Else{
     Return ''Yes";
    }
}
}
	\caption{\textsf{An $m$-partite Independence Test}}
	\label{algo:$m$-partite Independence testing}
\end{algorithm}

To see this algorithm succeed with high probability, we note that

If $\rho$ is independent, then $\rho$ is also independent in $S_1$ and $S_2$ since $25C\frac{\Pi_{i=1}^m d_i}{\epsilon^2}$ of $\rho$ are sufficient for testing independence in these bipartitions with at least a ${\frac{2}{3}}$ probability of success given a trace distance parameter of $\epsilon/5$.  According to the Chernoff bound, the probability of passing the test in Line 2 is at least $\sqrt[3]{\frac{2}{3}}$. This procedure is repeated 40 times. To see it also passing the test in Line 6 and the test in Line 10 with a probability of at least $\sqrt[3]{\frac{2}{3}}$, note that
    $$
    C\frac{\Pi_{i=1}^m d_i}{\epsilon^2}>2^{50}/25\cdot25C\frac{\Pi_{i=1}^{[\frac{m}{2}]} d_i}{\epsilon^2}, 2^{50}/25\cdot25C\frac{\Pi_{i=[\frac{m}{2}]+1}^{m}d_i}{\epsilon^2}
    $$
which is a sufficient number of copies for this algorithm, and the probability of outputting "Yes" is at least $\sqrt[3]{\frac{2}{3}}$.

If $||\rho-\sigma||_1>\epsilon$ for any independent $\sigma$, then one of the following three statements will be true:
$\rho$ is $\epsilon/5$-far from independent in the bipartition $S_1$ and $S_2$; or $\rho_{S_1}$ is $\epsilon/5$-far from independent; or $\rho_{S_2}$ is $\epsilon/5$-far from independent. Otherwise, assume that there exists an $\sigma_{S_1}$ and $\psi_{S_1}$ that is independent in $S_1$ and $\sigma_{S_2}$ and $\psi_{S_2}$ and independent in $S_2$ such that
$||\rho-\sigma_{S_1}\otimes\sigma_{S_2}||_1<\epsilon/5$, $||\rho_{S_1}-\psi_{S_1}||_1<\epsilon/5$ and $||\rho_{S_2}-\psi_{S_2}||_1<\epsilon/5$.

According to Proposition \ref{dis1}, we have $||\rho-\rho_{S_1}\otimes\rho_{S_2}||_1<3\epsilon/5$. Then, by the triangle inequality and Lemma \ref{dis2}, we have
 $$||\rho-\psi_{S_1}\otimes\psi_{S_2}||_1\leq ||\rho-\rho_{S_1}\otimes\rho_{S_2}||_1+||\psi_{S_1}\otimes\psi_{S_2}-\rho_{S_1}\otimes\rho_{S_2}||_1\leq \epsilon.$$
Contradiction! This algorithm outputs "No" with a probability of at least $\sqrt[3]{\frac{2}{3}}>\frac{2}{3}$.

We can derive an independent measurement tester by replacing the identity tester in Algorithm  \ref{algo:Identity-Test-Joint-Measurements} with Algorithm \ref{algo:Identity-Test-Independent-Measurements}, along with a bipartite independence tester for independent measurement. Through a similar analysis to the above, we have
$$O(\frac{\Pi_{i=1}^m d_i^2}{\epsilon^2})$$
which is a sufficient number of copies.

We can derive a streaming algorithm tester by replacing the identity tester in Algorithm \ref{algo:Identity-Test-Joint-Measurements} with Algorithm \ref{algo:Identity-Test-Streaming-Algorithm}, plus a bipartite independence tester for local measurement. By a similar analysis, we have
$$O(\frac{\Pi_{i=1}^m d_i^{1.5+\log 3}}{\epsilon^2})$$
which is a sufficient number of copies.	
\end{proof}
Using the dimension splitting technique within the proof of Theorem \ref{thmb}, we can also prove that the bound $O(\frac{\Pi_{i=1}^m d_i}{\epsilon^2})$ for joint measurement is tight up to a polylog factor.

\paragraph{Restatement of the lower bound in \ref{thm2}}
Let $\rho \in \mathcal{D}(\C^{d_1}\otimes \C^{d_2}\otimes\cdots\otimes\C^{d_m})$ with $d_1\geq d_2\geq\cdots\geq d_m$. To test whether $\rho$ is $m$-partite independent or $\epsilon$-far from $m$-partite independent in terms of trace distance,
$\Omega(\frac{\Pi_{i=1}^m d_i}{\epsilon^2\log^3 d_1\log\log d_1})$ copies are necessary for $\Pi_{i=2}^m d_i\leq2000$, and $\Omega(\frac{\Pi_{i=1}^m d_i}{\epsilon^2})$ copies are necessary for $\Pi_{i=2}^m d_i>2000$.

\begin{proof}
First, note that it suffices to only consider cases where $\Pi_{i=1}^m d_i$ are sufficiently large since $\Theta(\frac{1}{\epsilon^2})$ samples are
required to distinguish an independent distribution from the following classical distributions on $[2] \times [2]\times\cdots\times[2]$
$$p=(\frac{1+2\epsilon}{4},\frac{1-2\epsilon}{4},\frac{1-2\epsilon}{4},\frac{1+2\epsilon}{4})\times \mu_2\times\cdots\times\mu_m$$
where $\mu_i$ is a uniform distribution on $[2]$.
From the set of independent distributions, let
$$q=q_1\times q_2\times\cdots\times q_m$$
observe that the distance between $p$ and the set of independent distribution is still at least $\Theta(\epsilon)$. Let $p'=(\frac{1+2\epsilon}{4},\frac{1-2\epsilon}{4},\frac{1-2\epsilon}{4},\frac{1+2\epsilon}{4})$. Thus
$$
||q-p||_1\geq ||p'-q_1\times q_2||_1\geq \Theta(\epsilon).
$$
However, any test, which can distinguish $p$ from independent distributions, can distinguish $p'$ from independent distributions. As noted in the discussion of the lower bounds with bipartite independence testing, $\Theta(\frac{1}{\epsilon^2})$ samples are
required to distinguish $p'$ from independent distributions.

To show the lower bound for general dimensions, assume there exists an Algorithm A that uses $f(d_1,d_2,d_3,\dots,d_n,\epsilon)$ copies to decide whether a given multipartite $\rho\in \mathcal{D}(\C^{d_1}\otimes \C^{d_2}\otimes\cdots\otimes\C^{d_m})$ is independent or $\epsilon$-far from independent with a probability of at least $\frac{2}{3}$.
Just like the bipartite case, we can formalize the case where all $\Pi_{i=1}^m d_i$ are sufficiently large into two subcases.

Case 1: $\Pi_{i=2}^m d_i\leq 2000$. $d_1$ is sufficiently large. For any $\rho_{1,2}\in \mathcal{D}(\C^{d_1}\otimes \C^{d_2})$, the following state
\begin{equation*}
\rho=\rho_{1,2}\otimes\frac{I_{d_3}}{d_3}\otimes\cdots\otimes\frac{I_{d_m}}{d_m}
\end{equation*}
and $\sigma=\sigma_1\otimes\sigma_2\otimes\cdots\otimes\sigma_m$ which satisfies
\begin{equation*}
||\rho-\sigma||_1\geq ||\rho_{1,2}-\sigma_{1,2}||_1=||\rho_{1,2}-\sigma_1\otimes\sigma_2||_1.
\end{equation*}
Alternatively, let $\sigma=\sigma_1\otimes\sigma_2\otimes\frac{I_{d_3}}{d_3}\otimes\cdots\otimes\frac{I_{d_m}}{d_m}$ satisfies
\begin{equation*}
||\rho-\sigma||_1=||\rho_{1,2}-\sigma_1\otimes\sigma_2||_1.
\end{equation*}
Therefore, $\rho$ is $m$-partite independent if $\rho_{1,2}$ is bipartite independent. Moreover, for any $\epsilon>0$, $\rho$ is $\epsilon$-far from $m$-partite independent if, and only if, $\rho_{1,2}$ is $\epsilon$-far from bipartite independent.
The bipartite independence of $\rho_{1,2}$ can be tested by testing the $m$-partite independence of $\rho$ using Algorithm A. From the lower bound of the bipartite independence test, we know that
\begin{equation*}
f(d_1,d_2,d_3,\dots,d_n,\epsilon)\geq \Omega(\frac{d_1}{\epsilon^2\log^3 d_1\log\log d_1})=\Omega(\frac{\Pi_{i=1}^m d_i}{\epsilon^2\log^3 d_1\log\log d_1}).
\end{equation*}

Case 2: $\Pi_{i=2}^m d_i>2000$. Since $\Pi_{i=1}^md_i$ is sufficiently large, there must be a bipartition of $S=\{1,2,\dots,m\}$ into $S_1\subset S$ and $S_2=S\setminus S_1$ such that
$\Pi_{i\in S_1}d_i,\Pi_{i\in S_2} d_i>2000$ where $\Pi_{i\in S_1}d_i$ is sufficiently large. To see this, observe that, if $d_1$ is sufficiently large, we can let $S_1=\{1\}$, and $S_2=\{2,\dots,m\}$. Otherwise $d_1$ is a constant, which means all $d_i$ are constant, and $m$ must be sufficiently large because $\Pi_{i=1}^md_i$ is sufficiently large. Then, let $S_1=\{1,2,\dots,[\frac{m}{2}]\}$ and $S_2=\{[\frac{m}{2}]+1,\dots,m\}$, and without loss of generality, assume $\Pi_{i\in S_1}d_i\geq \Pi_{i\in S_2} d_i>2000$.

Algorithm \ref{algo:Bipartite Identity test A for maximally mixed state} can be used to test whether, for all $t>1$, $\rho=\frac{I_{d_1}}{d_1}\otimes \frac{I_{d_2}}{d_2}\otimes\cdots\otimes\frac{I_{d_m}}{d_m}$ or $||\rho-\frac{I_{d_1}}{d_1}\otimes \frac{I_{d_2}}{d_2}\otimes\cdots\otimes\frac{I_{d_m}}{d_m}||_1 \epsilon$.
Testing whether $\rho$ is $m$-partite independent or $\frac{(t-1)\epsilon}{4t}$-far from $m$-partite independent with at least a $\frac{28}{30}$ probability of success requires $100f(d_1,d_2,d_3,\dots,d_n,\frac{(t-1)\epsilon}{4t})$ copies of $\rho$. Similarly, testing whether $\rho_{S_1}=\otimes_{i\in S_1}\frac{I_{d_i}}{d_i}$ or $||\rho_{S_1}-\otimes_{i\in S_1}\frac{I_{d_i}}{d_i}||_1>\epsilon/t$ with at least a $\frac{20}{27}$  probability of success uses $300t^2\frac{\Pi_{i\in S_1}d_i}{\epsilon^2}$ copies of $\rho$; and testing whether $\rho_{S_2}=\otimes_{i\in S_2}\frac{I_{d_i}}{d_i}$ or $||\rho_{S_2}-\otimes_{i\in S_2}\frac{I_{d_i}}{d_i}||_1>\frac{(t-1)\epsilon}{4t}$ with at least a $\frac{27}{28}$  probability of success uses $\Theta(\frac{t^2\Pi_{i\in S_2} d_i}{(t-1)^2\epsilon^2})$ copies of $\rho$.

This algorithm can also succeed at detecting whether $\rho$ is maximally mixed with high probability. In that:
\begin{itemize}
\item If $\rho=\frac{I_{d_1}}{d_1}\otimes \frac{I_{d_2}}{d_2}\otimes\cdots\otimes\frac{I_{d_m}}{d_m}$, then it is $m$-partite independent. The probability success is at least $\frac{2}{3}$.

\item If $||\rho-\frac{I_{d_1}}{d_1}\otimes \frac{I_{d_2}}{d_2}\otimes\cdots\otimes\frac{I_{d_m}}{d_m}||_1>\epsilon$, then one of the following three statements will be true following the same arguments as in the bipartite case:
$\rho_{S_1}$ is $\epsilon/t$-far from $\otimes_{i\in S_1}\frac{I_{d_i}}{d_i}$; or $\rho_{S_2}$ is $\frac{(t-1)\epsilon}{4t}$-far from $\otimes_{i\in S_2}\frac{I_{d_i}}{d_i}$; or $\rho$ is $\frac{(t-1)\epsilon}{4t}$-far from independent in the $S_1$ and $S_2$ bipartition. If one of the previous two statements is true, the algorithm outputs "No" with a probability of at least $\frac{2}{3}$. Otherwise, $\rho$ is $\frac{(t-1)\epsilon}{4t}$-far from being bipartite independent. Note that the set of $m$-partite independent states is a subset of the bipartite independent states. Thus, $\rho$ is $\frac{(t-1)\epsilon}{4t}$-far from $m$-partite independence, and the algorithm outputs "No" with a probability of at least $\frac{2}{3}$.
\end{itemize}
Invoking Theorem \ref{mixness},
$$100f(d_1,d_2,d_3,\dots,d_n,\frac{(t-1)\epsilon}{4t})+300t^2\frac{\Pi_{i\in S_1}d_i}{\epsilon^2}+\Theta(\frac{t^2\Pi_{i\in S_2} d_i}{(t-1)^2\epsilon^2})\geq 0.15\frac{\Pi_{i=1}^m d_i}{\epsilon^2}.$$
We can choose $t=\sqrt{\frac{2000.5}{2000}}$, then
$$100f(d_1,d_2,d_3,\dots,d_n,\frac{(t-1)\epsilon}{4t})\geq 0.15\frac{\Pi_{i=1}^m d_i}{\epsilon^2}-300t^2\frac{\Pi_{i\in S_1} d_i}{\epsilon^2}-\Theta(t^2\frac{\Pi_{i\in S_2} d_i}{(t-1)^2\epsilon^2})=\Theta(\frac{\Pi_{i\in S_1} d_i}{\epsilon^2})=\Theta(\frac{\Pi_{i=1}^m d_i}{\epsilon^2}).$$
That is
$$
f(d_1,d_2,d_3,\dots,d_n,\epsilon)=\Omega(\frac{\Pi_{i=1}^m d_i}{\epsilon^2}).
$$
\end{proof}
As a direct consequence, observe the following by noticing
$2^{11}>2000$,
\begin{Cor}
If $m\geq 12$, the sample complexity of $m$-partite independence testing is $\Theta(\frac{\Pi_{i=1}^m d_i}{\epsilon^2})$.
\end{Cor}

\section{Testing Properties of Collections of Quantum States}\label{TPC}
The problem of property testing collections of discrete distributions was studied in \cite{LRR13,DK16}. In this section, we explore the quantum counterpart to this problem, which is to test the properties of a collection of quantum states in a query model with $n$ quantum states $\rho_1$, . . . , $\rho_n$  and a given index that we can choose to access.

\subsection{Identity of collections}
This first demonstration is to distinguish cases where all $\rho_i$ are identical from cases where there is no quantum state $\sigma$ such that
$$
\frac{1}{n}\sum_{i=1}^n ||\sigma-\rho_i||_1\leq \epsilon
$$
Algorithm \ref{algo:Identical-Test-Collection-Query-Model}, for the above test with joint measurement, has a similar structure to Diakonikolas and Kane’s \cite{DK16} algorithm for testing the identity of collections of probability distributions.
\begin{algorithm}[H]
	\Input{Access to quantum states $\rho_1$, . . . , $\rho_n$ on $\mathcal{D}(\C^{d})$ and $ \epsilon> 0$.}	
    \Output{"Yes" with a probability of at least $\frac{2}{3}$ if $\rho_i$s are identical; and "No" with a probability of at least $\frac{2}{3}$ if there is no quantum state $\sigma$ such that $
\frac{1}{n}\sum_{i=1}^n ||\sigma-\rho_i||_1> \epsilon
$.}
    Let $L$ be a sufficiently large constant\;
    \For{$k\gets 0$ \KwTo $\lceil\log_2(n(n-1))\rceil$}{
    Select $2^{3k/2}C$ uniformly random pair of elements $(i,j)\in [n]\times [n]$ with restriction $i\neq j$\;
    For each selected $(i,j)$, use tester of Algorithm \ref{BOW} to distinguish between $\rho_i=\rho_j$ and $||\rho_i-\rho_j||_1>2^{k-1}\epsilon$ with a failure probability of at most $L^{-2}6^{-k}$\;
    If any of these testers returned "No", return "No"\;
}
Return "Yes"\;
	\caption{\textsf{An Identical Identity Test for Collections with a Query Model}}
	\label{algo:Identical-Test-Collection-Query-Model}
\end{algorithm}

Note that, if all $\rho_i$s are identical, the probability of success is $\Pi_{k=0}^{\lceil\log_2(n(n-1))\rceil} p_k$,
where
\begin{align*}
p_k\geq (1-L^{-2}6^{-k})^{2^{3k/2}L}\geq (1-2^{3k/2}LL^{-2}6^{-k})=1-\frac{2^{k/2}}{3^kL}.
\end{align*}
Then
\begin{align*}
\Pi_{k=0}^{\lceil\log_2(n(n-1))\rceil} p_k\geq 1-\sum_{k=0}^{\lceil\log_2(n(n-1))\rceil}\frac{2^{k/2}}{3^kL}\geq 1-O(\frac{1}{L}).
\end{align*}
However, if for any $\sigma$, $\frac{1}{n}\sum_{i=1}^n ||\sigma-\rho_i||_1>\epsilon$, the in particular, $\frac{1}{n}\sum_{i=1}^n ||\rho_j-\rho_i||_1>\epsilon$ for any $j$. That is
\begin{align*}
\frac{1}{n(n-1)}\sum_{i\neq j} ||\rho_i-\rho_j||_1>\epsilon=\frac{\epsilon}{2}+\sum_{k=0} \frac{\sqrt{2}-1}{2} (2^{-3k/2})2^k\epsilon> \frac{\epsilon}{2}+\sum_{k=0} \frac{1}{5} (2^{-3k/2})2^k\epsilon.
\end{align*}
Observe that
\begin{align*}
\frac{1}{n(n-1)}\sum_{i\neq j} ||\rho_i-\rho_j||_1<& \frac{|\{(i,j): ||\rho_i-\rho_j||_1<\frac{\epsilon}{2}\}|\frac{\epsilon}{2}}{n(n-1)}+\sum_{k=0} \frac{|\{(i,j):2^{k-1}\epsilon\leq ||\rho_i-\rho_j||_1< 2^k\epsilon\}|{2^k\epsilon}}{n(n-1)}\\
<&\frac{\epsilon}{2}+\sum_{k=0} \frac{|\{(i,j):2^{k-1}\epsilon\leq ||\rho_i-\rho_j||_1\}|{2^k\epsilon}}{n(n-1)}
\end{align*}
Therefore, for some $k\geq 0$, it must hold that
$$|\{(i,j):2^{k-1}\epsilon\leq ||\rho_i-\rho_j||_1\}|>\frac{1}{5} (2^{-3k/2})n(n-1).$$
Actually, there is always such a $k\leq \lceil\log_2(n(n-1))\rceil$. And, if we find some $k_0>\lceil\log_2(n(n-1))\rceil$ with above property, then the above property is also true for $k=\lceil\log_2(n(n-1))\rceil$.
\begin{align*}
&|\{(i,j):2^{k_0-1}\epsilon\leq ||\rho_i-\rho_j||_1\}|>\frac{1}{5} (2^{-3k_0/2})n(n-1)\\
\Rightarrow &|\{(i,j):2^{k_0-1}\epsilon\leq ||\rho_i-\rho_j||_1\}|\geq 1\\
\Rightarrow &|\{(i,j):2^{\lceil\log_2(n(n-1))\rceil-1}\epsilon\leq ||\rho_i-\rho_j||_1\}|\\
\geq &|\{(i,j):2^{k_0-1}\epsilon\leq ||\rho_i-\rho_j||_1\}|\\
\geq &1 \\
\geq&\frac{1}{5} (2^{-3\lceil\log_2(n(n-1))\rceil/2})n(n-1).
\end{align*}
Here, the probability of selecting some $(i,j)$ with this property is at least
\begin{align*}
1-(1-\frac{1}{5\times 2^{3k/2}})^{2^{3k/2}L}\geq 1-O(e^{-L/5}).
\end{align*}
After this, the corresponding tester will return "No" with high probability.

The sample complexity of this algorithm is
$$
\sum_{k=0} 2^{3k/2}L \times O(\frac{d}{(2^{k-1}\epsilon)^2})=O(\frac{d}{\epsilon^2}).
$$

In the independent measurement setting, Algorithm \ref{algo:Identity-Test-Independent-Measurements} can replace Algorithm \ref{algo:Identity-Test-Joint-Measurements} and, from a similar analysis, the total sample complexity becomes $O(\frac{d^2}{\epsilon^2})$.

In the local measurement setting, Algorithm \ref{algo:Identity-Test-Streaming-Algorithm} can replace Algorithm \ref{algo:Identity-Test-Joint-Measurements}, with a sample complexity of $O(\frac{d^{1.5+\log 3}}{\epsilon^2})$ if all $d_i$ are to the power of $2$.

In fact, this idea also holds with the more general version of the problem, i.e., where there are $n$ states $\rho_1$, . . . , $\rho_n$ on $\mathcal{D}(\C^{d})$. Here, the goal is to distinguish cases where all $\rho_i$ are identical from cases where there is no quantum state $\sigma$ such that
$$
\sum_{i=1}^n c_i||\sigma-\rho_i||_1\leq \epsilon
$$
where $c_i> 0$ and $C_0\leq\sum_{i=1}^nc_i\leq C_1$ for absolute constant $0<C_0\leq C_1$.
To see this, choose rational $t_i$ such that $c_i/2\leq t_i\leq c_i$; then we can distinguish cases where all $\rho_i$ are identical from cases where, for any quantum state $\sigma$,
$$
\sum_{i=1}^n t_i||\sigma-\rho_i||_1\geq \frac{\epsilon}{2}.
$$
Note that the last condition is equivalent to a rational $\mu_i=\frac{t_i}{\sum_i t_i}$ such that for any quantum state $\sigma$
$$
\sum_{i=1}^n \mu_i||\sigma-\rho_i||_1\geq \frac{\epsilon}{2\sum_i t_i}\geq \frac{\epsilon}{2C_1}=\Theta(\epsilon).
$$
Let $\mu_i=\frac{n_i}{m}$ for the integers $n_i$ and $m$.
This is equivalent to a collection of $\tilde{\rho_j}$
$$
\frac{1}{m}\sum_{j=1}^{m} ||\sigma-\tilde{\rho_j}||_1\geq \Theta(\epsilon),
$$
where the number of $j$s that satisfy $\tilde{\rho_j}=\rho_i$ is $n_i$, which verifies the following restatement of \ref{cidentical}.
\paragraph{restatement of Theorem \ref{cidentical}.}
Given access to the quantum states $\rho_1$, . . . , $\rho_n$ on $\mathcal{D}(\C^{d})$ and an explicit $c_i>0$ with $C_1\geq\sum_{i} c_i\geq C_0>0$ where $C_0,C_1$ are absolute constants, the sample complexity of distinguishing, with at least a $\frac{2}{3}$ probability of success, the cases where all $\rho_i$ are identical from the cases where $\sum_i c_i ||\rho_i-\sigma||_1> \epsilon$ for any $\sigma$ is
\begin{itemize}
\item  $\Theta(\frac{d}{\epsilon^2})$ with joint measurement;
\item $O(\frac{d^2}{\epsilon^2})$ with independent measurement; and
\item $O(\frac{d^{1.5+\log 3}}{\epsilon^2})$ with a streaming algorithm where all $d_i$ are to the power of $2$.
\end{itemize}

The optimality of joint measurement comes from applying Theorem \ref{mixness} and the following simple case: For even $n$, $\rho_1=\rho_2=\cdots=\rho_{[\frac{n}{2}]}$
and $\rho_{[\frac{n}{2}]+1}=\cdots=\rho_n=\frac{I_d}{d}$, with $||\rho_1-\frac{I_d}{d}||_1>2\epsilon$, and $c_i=\frac{1}{n}$.

\subsection{Independence of collections}
In this test, there are $n$ quantum states $\rho_1$, . . . , $\rho_n$ on $\mathcal{D}(\C^{d_1}\otimes \C^{d_2}\otimes\cdots\otimes\C^{d_m})$, and the goal is to distinguish cases where all $\rho_i$ are independent, i.e.,  $\rho_i=\otimes_{k=1}^m\sigma_{k,i}$ for some $\sigma_{k,i}\in\mathcal{D}(\C^{d_k})$ for $1\leq k\leq m$, from cases where,
$$
\frac{1}{n}\sum_{i=1}^n ||\otimes_{k=1}^m\sigma_{k,i}-\rho_i||_1> \epsilon,
$$
for any $\sigma_{k,i}\in\mathcal{D}(\C^{d_k})$.

The algorithm for the bipartite case is:

\begin{algorithm}[H]
	
	\Input{Access to quantum states $\rho_1$, . . . , $\rho_n$ on $\mathcal{D}(\C^{d_1}\otimes \C^{d_2})$ with $\epsilon> 0$.}	
    \Output{"Yes" with a probability of at least $\frac{2}{3}$ if $\rho_i$s are independent; and "No" with a probability of at least $\frac{2}{3}$ if there are no quantum states $\sigma_{1,i},\sigma_{2,i}$ such that $
\frac{1}{n}\sum_{i=1}^n ||\sigma_{1,i}\otimes\sigma_{2,i}-\rho_i||_1> \epsilon
$}
    Let $L$ be a sufficiently large constant\;
    \For{$k\gets 0$ \KwTo $\lceil\log_2 n\rceil$}{
    Select $2^{3k/2}L$ uniformly random pair of elements $i\in [n]$\;
    For each selected $(i,j)$, use Algorithm \ref{algo:Identity-Test-Joint-Measurements} to distinguish between $\rho_i$ being independent and $\rho_i$ being $2^{k-1}\epsilon$-far from being independent with a failure probability of at most $L^{-2}6^{-k}$\;
    If any of these testers returned “No”, return “No”\;
}
Return "Yes"\;
	\caption{\textsf{An Independence Test for Collections with a Query Model}}
	\label{algo:Independence-Test-Collection-Query-Model}
\end{algorithm}

The sample complexity of this algorithm is
$$
O(\frac{d_1d_2}{\epsilon^2}),
$$
and the analysis is the same as for Algorithm \ref{algo:Identical-Test-Collection-Query-Model}.

In the independent measurement setting, Algorithm  \ref{algo:Identity-Test-Independent-Measurements} can replace Algorithm \ref{algo:Identity-Test-Joint-Measurements} and, from a similar analysis, the total sample complexity becomes
$$
O(\frac{d_1^2d_2^2}{\epsilon^2}).
$$
In the independent measurement setting, Algorithm \ref{algo:Identity-Test-Streaming-Algorithm} can replace Algorithm  \ref{algo:Identity-Test-Joint-Measurements}, with a sample complexity of
$$O(\frac{d_1^{1.5+\log 3}d_2^{1.5+\log 3}}{\epsilon^2}).$$
Algorithm \ref{algo:Independence-Test-Collection-Query-Model}  can be directly generalized into an $m$-partite version following the framework in Section \ref{TI}, which leads to the following restatement of Theorem \ref{cindependent}.

\paragraph{Restatement of Theorem \ref{cindependent}}
Given sample access to quantum states $\rho_1$, . . . , $\rho_n$ in $\mathcal{D}(\C^{d_1}\otimes \C^{d_2}\otimes\cdots\otimes\C^{d_m})$ with $d_1\geq d_2\geq\cdots\geq d_m$ and explicit $c_i>0$ with  $C_1\geq\sum_{i} c_i\geq C_0>0$ where $C_0,C_1$ are absolute constants, the sample complexity of distinguishing, with at least a $\frac{2}{3}$ probability of success, the cases where all $\rho_i$ are $m$-partite independent from the cases where $\sum_i c_i ||\rho_i-\otimes_{k=1}^m\sigma_{k,i}||_1> \epsilon$ for any $\sigma_{k,i}\in\mathcal{D}(\C^{d_k})$ is
\begin{itemize}
\item $\Tilde{\Theta}(\frac{\Pi_{i=1}^m d_i}{\epsilon^2})$ with joint measurement;
\item $O(\frac{\Pi_{i=1}^md_i^2}{\epsilon^2})$ with independent measurement; and
\item $O(\frac{\Pi_{i=1}^md_i^{1.5+\log 3}}{\epsilon^2})$ with a streaming algorithm where all $d_i$ are to the power of $2$.
 \end{itemize}

The lower bound with the joint measurement approach derives from applying Theorem \ref{thmb} to the simple case, where all $\rho_i$ are identical. The goal then becomes one of independence testing for a single state.

\subsection{Independent and identical collections}

In this test, there are $n$ quantum states $\rho_1$, . . . , $\rho_n$ on $\mathcal{D}(\C^{d_1}\otimes \C^{d_2}\otimes\cdots\otimes\C^{d_m})$, and the goal is to distinguish
cases where the $\rho_i=\otimes_{k=1}^m\sigma_{k}$ for some $\sigma_k\in\mathcal{D}(\C^{d_k})$ from cases where
$$
\frac{1}{n}\sum_{i=1}^n ||\otimes_{k=1}^m\sigma_{k}-\rho_i||_1> \epsilon,
$$
 for any $\sigma_{k}\in\mathcal{D}(\C^{d_k})$.

\begin{algorithm}[H]
	\Input{Access to quantum states $\rho_1$, . . . , $\rho_n$ in $\mathcal{D}(\C^{d_1}\otimes \C^{d_2})$ with $\epsilon> 0$.}	
    \Output{"Yes" with a probability of at least $\frac{2}{3}$ if $\rho_i$s are identical and independent; and "No" with a probability of at least $\frac{2}{3}$ if there are no quantum states $\sigma_{1},\sigma_{2}$ such that $
\frac{1}{n}\sum_{i=1}^n ||\sigma_{1}\otimes\sigma_{2}-\rho_i||_1\leq \epsilon
$}

    Run the Identical Test \ref{algo:Identical-Test-Collection-Query-Model} for $\frac{\epsilon}{3}$ with a failure probability of at most $0.01$\;
    Run the Independence Test \ref{algo:Independence-Test-Collection-Query-Model} for $\frac{\epsilon}{3}$ with a failure probability of at most $0.01$\;
    If any of these testers returned “NO”, return “NO”\;
    Otherwise, return "Yes"\;
	\caption{\textsf{An Identical Independence Test for Collections with a Query Model}}
	\label{algo:Identical-Independence-Test-Collection-Query-Model}
\end{algorithm}

The correctness of this algorithm follows from observing that:
if $\frac{1}{n}\sum_{i}||\rho_i-\sigma_{1,i}\otimes\sigma_{2,i}||_1\leq \frac{\epsilon}{3}$ and $\frac{1}{n}\sum_{i}||\rho_i-\sigma||_1\leq \frac{\epsilon}{3}$, then $\frac{1}{n}\sum_{i}||\sigma-\sigma_{1,i}\otimes\sigma_{2,i}||_1\leq \frac{2\epsilon}{3}$, and there exists a $j$ such that $||\sigma-\sigma_{1,j}\otimes\sigma_{2,j}||\leq \frac{2\epsilon}{3}$. Therefore,
$$\frac{1}{n}\sum_{i}||\sigma_{1,j}\otimes\sigma_{2,j}-\rho_i||_2\leq \epsilon.$$
The sample complexity of this algorithm is
$$
O(\frac{d_1d_2}{\epsilon^2}).
$$
In the independent measurement setting, Algorithm \ref{algo:Identity-Test-Independent-Measurements} can replace Algorithm  \ref{algo:Identity-Test-Joint-Measurements} and, from a similar analysis, the total sample complexity becomes
$$
O(\frac{d_1^2d_2^2}{\epsilon^2}).
$$
In the local measurement setting, Algorithm \ref{algo:Identity-Test-Streaming-Algorithm} can replace Algorithm \ref{algo:Identity-Test-Joint-Measurements}, with a sample complexity of
$$O(\frac{d_1^{1.5+\log 3}d_2^{1.5+\log 3}}{\epsilon^2}).$$

Algorithm \ref{algo:Identical-Independence-Test-Collection-Query-Model} can be directly generalized into an $m$-partite version following the framework in Section \ref{TI}.

\begin{Prop}\label{iai}
Given sample access to the quantum states $\rho_1$, . . . , $\rho_n$ $\mathcal{D}(\C^{d_1}\otimes \C^{d_2}\otimes\cdots\otimes\C^{d_m})$ with $\epsilon> 0$ and $d_1\geq d_2\geq\cdots\geq d_m$ and explicit numbers $c_i>0$ as absolute constants $C_1\geq\sum_{i} c_i\geq C_0>0$ for absolute constant $C_0,C_1$, the sample complexity of distinguishing, with at least a $\frac{2}{3}$ probability, the cases where $\rho_i=\otimes_{k=1}^m\sigma_{k}$ for some $\sigma_k\in\mathcal{D}(\C^{d_k})$ and the cases where $\frac{1}{n}\sum_{i=1}^n ||\otimes_{k=1}^m\sigma_{k}-\rho_i||_1> \epsilon$  for any $\sigma_{k}\in\mathcal{D}(\C^{d_k})$ is
\begin{itemize}
\item $\Tilde{\Theta}(\frac{\Pi_{i=1}^m d_i}{\epsilon^2})$ with joint measurement;
\item $O(\frac{\Pi_{i=1}^md_i^2}{\epsilon^2})$ with independent measurement; and
\item $O(\frac{\Pi_{i=1}^md_i^{1.5+\log 3}}{\epsilon^2})$ with a streaming algorithm where all $d_i$ are to the power of $2$.
 \end{itemize}
 The lower bound for joint measurement derives from applying Theorem \ref{thmb} to the simple case where all $\rho_i$ are identical.  The goal then becomes one of independence testing a single state.
\end{Prop}

\section{A Conditional Independence Test of Classical-Quantum-Quantum States}\label{CIT}

The demonstrations in this section cover both joint and independent measurements. A basic component of conditional independence testing generally is an efficient estimator of the $\ell_2$ distance between a bipartite quantum state and the tensor product of its marginal.

With joint measurement, directly employing the $\ell_2$ distance estimator in Theorem \ref{BOW} would require at least $12$ copies of the bipartite state. However, this requirement can be weakened to no less than $4$ copies, which is optimal. This number of $4$ is crucial in our analysis, as a suboptimal number affects the complexity significantly.
With independent measurement, the $\ell_2$ distance estimator given in Section \ref{sec:independence} does not fit because it would destroy the tensor product structure, i.e., independence. Using the two steps introduced in Section \ref{sec:local}, an $\ell_2$ distance estimator works as long as there are at least $4$ copies. This preserves the tensor product structure in the sense that the image of the independent bipartite state is still in tensor product form, i.e., it is the tensor product of the probability distributions.

Within these estimators, the framework given in \cite{Canonne:2018:TCI:3188745.3188756} can be used for conditional independence testing of classical distributions.
\subsection{Joint Measurement}

\begin{lemma}\label{joint-testing}
Given a bipartite quantum state $\rho_{1,2}\in\mathcal{D}(\C^{d_1}\otimes\C^{d_2})$, there is an estimator,  denoted by
$\mathrm{Estimator-independent}: \mathcal{D}(\C^{d_1}\otimes\C^{d_2})\times \mathbb{N}\mapsto \mathbb{R}$ which measures $n$ copies of $\rho_{1,2}$
such that for $n\geq 4$
\begin{align*}
&\mathbb{E}[\mathrm{Estimator-joint}(\rho_{1,2},n)]=||\rho_{1,2}-\rho_1\otimes \rho_2||_2^2,\\
&\mathrm{Var}[\mathrm{Estimator-joint}(\rho_{1,2},n)]=O(\frac{||\rho_{1,2}-\rho_1\otimes \rho_2||_2^2}{n}+\frac{1}{n^2}).
\end{align*}
\end{lemma}
\begin{proof}If $n\geq 12$, $6[\frac{n}{6}]$ copies of $\rho_{1,2}$ are used. $2[\frac{n}{6}]$ copies are used to generate $2[\frac{n}{6}]$ copies of $\rho_A$, and $2[\frac{n}{6}]$ copies are used to generate $2[\frac{n}{6}]$ copies of $\rho_B$. Now there are $2[\frac{n}{6}]\geq 4$ copies of $\rho_1\otimes \rho_2$ and $2[\frac{n}{6}]\geq 4$ copies of $\rho_{1,2}$.

Theorem \ref{BOW} provides an $\mathrm{Estimator}:\C^{d_1\times d_1}\otimes\C^{d_2\times d_2}\times \mathbb{N}\mapsto \mathbb{R}^{+}$,
\begin{align*}
&\mathbb{E}[\mathrm{Estimator-joint}(\rho_{1,2},n)]=||\rho_{1,2}-\rho_1\otimes \rho_2||_2^2,\\
&\mathrm{Var}[\mathrm{Estimator-joint}(\rho_{1,2},n)]=O(\frac{||\rho_{1,2}-\rho_1\otimes \rho_2||_2^2}{2[\frac{n}{6}]}+\frac{1}{(2[\frac{n}{6}])^2})=O(\frac{||\rho_{1,2}-\rho_1\otimes \rho_2||_2^2}{n}+\frac{1}{n^2}).
\end{align*}
If $4\leq n<12$, only $4$ copies needed to be used. Construct an observable $M$ such that $|M|\leq 10 I$ such that
\begin{align*}
\mathbb{E}[\mathrm{Estimator-joint}(\rho_{1,2},4)]=\tr[M(\rho_{1,2}^{\otimes 4})]=||\rho_{1,2}-\rho_1\otimes \rho_2||_2^2.
\end{align*}
This follows from letting
\begin{align*}
M=O_1-2O_2+O_3,
\end{align*}
where the bounded operators $O_1,O_2,O_3$
\begin{align*}
\tr[O_1(\rho_{1,2}^{\otimes 4})]=&\tr(\rho_{1,2}^2),\\
\tr[O_2(\rho_{1,2}^{\otimes 4})]=&\tr[\rho_{1,2}(\rho_{1}\otimes\rho_2)],\\
\tr[O_3(\rho_{1,2}^{\otimes 4})]=&\tr(\rho_{1}^2)\tr(\rho_2^2)=\tr[(\rho_1\otimes\rho_2)^2].
\end{align*}
For $4\leq n<12$, we have
\begin{align*}
&\mathrm{Var}[\mathrm{Estimator-joint}(\rho_{1,2},n)]\\
=&\mathbb{E}(M^2)-\mathbb{E}^2(M)\\
=&\tr[M^2(\rho_{1,2}^{\otimes 4})]-\tr^2[M(\rho_{1,2}^{\otimes 4})]\\
\leq&\tr[M^2(\rho_{1,2}^{\otimes 4})]\\
\leq& 100\\
\leq& O(\frac{||\rho_{1,2}-\rho_1\otimes \rho_2||_2^2}{12}+\frac{1}{12^2}).
\end{align*}
Therefore, the statement is true for all $n\geq 4$.
\end{proof}

Consider the set of states
$$
\tau_{ABC}=\rho_{ABC}\in\Delta(C)\otimes \mathcal{D}(\C^{d_1}\otimes\C^{d_2}),
$$
where $|C|=n$.

\begin{algorithm}[H]
	\Input{Access to classical-quantum-quantum states $\rho=\sum_{c}p_c\op{c}{c}\otimes\rho^c_{AB}\in\tau_{ABC}$ and $\epsilon> 0$.}	
    \Output{"Yes" with a probability of at least $\frac{2}{3}$ if $\rho$ is conditionally independent; and "No" with a probability of at least $\frac{2}{3}$ if there is no conditionally independent $\sigma$ such that $||\rho-\sigma||_1\leq \epsilon$.}
Choose $L>0$ be a sufficiently large constant\;
\If{$n\geq \frac{d_1^4d_2^4}{\epsilon^8}$}{
$m\leftarrow L\frac{d_1^{\frac{4}{7}}d_2^{\frac{4}{7}}n^{\frac{6}{7}}}{\epsilon^{\frac{8}{7}}}$\;
}
\If{$\frac{d_1^{\frac{4}{3}}d_2^{\frac{4}{3}}}{\epsilon^{\frac{8}{3}}}\leq n<\frac{d_1^4d_2^4}{\epsilon^8}$}{
$m\leftarrow L\frac{\sqrt{d_1d_2}n^{\frac{7}{8}}}{\epsilon}$\;
}
\If{$n\leq \frac{d_1^{\frac{4}{3}}d_2^{\frac{4}{3}}}{\epsilon^{\frac{8}{3}}}$}{
$m\leftarrow L\frac{\sqrt{n}d_1d_2}{\epsilon^2}$\;
}
\tcc{Equivalently, $m\leftarrow L\max\{\frac{\sqrt{n}d_1d_2}{\epsilon^2},\min\{\frac{d_1^{\frac{4}{7}}d_2^{\frac{4}{7}}n^{\frac{6}{7}}}{\epsilon^{\frac{8}{7}}},\frac{\sqrt{d_1d_2}n^{\frac{7}{8}}}{\epsilon}\}\}$.}
$\xi\leftarrow\frac{1-\frac{5}{2e}}{2}\min\{\frac{m\epsilon^2}{4d_1d_2},\frac{m^4\epsilon^4}{32d_1^2d_2^2n^3}\}$\;
Set $M$ according to a $\mathrm{Poisson}(m)$ distribution\;
Draw $M$ copies of $\rho$ and measure the system $C$ on a computational basis for each copy. Let $S$ denote the multi-set of samples\;
\For{all $c\in C$}{
Let $S_c$ be the $|S_c|$ copies of $\rho^c_{AB}$\;
\If{$|S_c|\geq 4$}{
$A_c\leftarrow |S_c|\times \mathrm{Estimator-joint}(\rho^c_{AB},|S_c|)$\;
}
\Else
{
$A_c\leftarrow 0$\;
}
}
    \If{$A=\sum_{c\in C} A_c>\xi$}{
     Returned “NO”\;}
     \Else{
    Return "Yes"\;
    }
	\caption{\textsf{A Conditional Independence Test with Joint Measurement}}
	\label{algo:Conditional-Independence-Test-Joint}
\end{algorithm}
For the sake of completeness and clarifying the parameters, the details of the computation of \cite{Canonne:2018:TCI:3188745.3188756} for this case are provided in the Appendix.

\begin{Prop}\label{ciev}
If $\rho$ is conditionally independent, $$\mathbb{E}(A)=0.$$
If $\rho$ is $\epsilon$-far from being conditionally independent,
$$\mathbb{E}(A)\geq (1-\frac{5}{2e})\min\{\frac{m\epsilon^2}{4d_1d_2},\frac{m^4\epsilon^4}{32d_1^2d_2^2n^3}\}.$$
In both cases,
$$
\mathrm{Var}(A)=O(\min\{m,n\}+\mathbb{E}(A)).
$$
\end{Prop}
$m_0=\max\{\frac{\sqrt{n}d_1d_2}{\epsilon^2},\min\{\frac{d_1^{\frac{4}{7}}d_2^{\frac{4}{7}}n^{\frac{6}{7}}}{\epsilon^{\frac{8}{7}}},\frac{\sqrt{d_1d_2}n^{\frac{7}{8}}}{\epsilon}\}$ actually guarantees that
$$
\min\{\frac{m_0\epsilon^2}{d_1d_2},\frac{m_0^4\epsilon^4}{d_1^2d_2^2n^3}\}\geq \sqrt{\min\{m_0,n\}}.
$$
By setting $m=Lm_0$, we always have
\begin{align*}
&\min\{\frac{m\epsilon^2}{d_1d_2},\frac{m^4\epsilon^4}{d_1^2d_2^2n^3}\}\geq \sqrt{L}\times\sqrt{\min\{m,n\}}\\
\Rightarrow &\xi\geq \frac{(1-\frac{5}{2e})\sqrt{L}}{64}\times\sqrt{\min\{m,n\}},\\
&\mathbb{E}(A)\geq \frac{(1-\frac{5}{2e})\sqrt{L}}{32}\times\sqrt{\min\{m,n\}}.
\end{align*}
If $\rho$ is conditionally independent, then
\begin{align*}
\mathrm{Pr}[A>\xi]\leq\frac{\mathrm{Var}(A)}{\xi^2}= O(\frac{\min\{m,n\}}{\xi^2})\leq \frac{1}{3}.
\end{align*}
If $\rho$ is $\epsilon$-far from being conditionally independent, then
\begin{align*}
\mathrm{Pr}[A<\xi]\leq\mathrm{Pr}[A<\frac{\mathbb{E}(A)}{2}]\geq\mathrm{Pr}[|A-\mathbb{E}(A)|>\frac{\mathbb{E}(A)}{2}]
\leq\frac{\mathrm{Var}(A)}{4\mathbb{E}^2(A)}= O(\frac{\min\{m,n\}}{\mathbb{E}^2(A)}+\frac{1}{\mathbb{E}(A)})\leq \frac{1}{3}.
\end{align*}
Therefore, the following is validity
\paragraph{Restatement of the joint measurement part in Theorem \ref{cid}}
$O(\max\{\frac{\sqrt{n}d_1d_2}{\epsilon^2},\min\{\frac{d_1^{\frac{4}{7}}d_2^{\frac{4}{7}}n^{\frac{6}{7}}}{\epsilon^{\frac{8}{7}}},\frac{\sqrt{d_1d_2}n^{\frac{7}{8}}}{\epsilon}\}\})$ copies
are sufficient for distinguishing between $\rho_{ABC}$ as being conditionally independent from $\epsilon$-far from it being conditionally independent in the classical-quantum-quantum state $\rho_{ABC}\in \tau_{ABC}$ with $|C|=n$ using joint measurement.
\subsection{Independent measurement}
For $\rho_{1,2}$, let $\sigma_{1,2}=\rho_{1}\otimes\rho_2$ with $\rho_1$ and $\rho_2$ being the marginal of $\rho_{1,2}$. Apply the independent measurement $\mathcal{M}=\mathcal{M}_1\otimes\mathcal{M}_2$ given in Section \ref{sec:local} on $\rho_{1,2}$ and obtain $p_{1,2}$. $q_{1,2}$ is a product distribution $q_{1,2}=q_1\otimes q_2$ where $q_1\in\Delta([d_1(d_1+1)])$ is obtained from a POVM $\mathcal{M}_1$ of $\rho_1$, and $q_2\in\Delta([d_2(d_2+1])$ is obtained from a POVM $\mathcal{M}_2$ of $\rho_2$. In other words, $q_1=p_1$ and $q_2=p_2$.
Therefore
\begin{align*}
||p_{1,2}-p_1\otimes p_2||_2=\frac{||\rho_{1,2}-\rho_1\otimes\rho_2||_2}{(d_1+1)(d_2+1)}.
\end{align*}
The $\ell_2$ distance between $\rho_{1,2}$ and $\rho_1\otimes\rho_2$ can be tracked by tracking the $\ell_2$ distance between $p_{1,2}$ and $p_1\times p_2$--in other words, by tracking $p_{1,2}$.
According to Theorem \ref{local-measurement-classical},
\begin{align*}
||p_{1,2}||_2\leq& \frac{2}{(d_1+1)(d_2+1)},\\
||p_1||_2||p_2||_2\leq&\frac{\sqrt{2}}{(d_1+1)}\frac{\sqrt{2}}{(d_2+1)}=\frac{2}{(d_1+1)(d_2+1)}.
\end{align*}
Combining this with Theorem \ref{classical-estimate}, gives rise to Lemma \ref{variance-independent}.

\begin{lemma}\label{variance-independent}
For the bipartite quantum state $\rho_{1,2}\in \mathcal{D}(\C^{d_1}\otimes\C^{d_2})$, there is an estimator, denoted by
$\mathrm{Estimator-independent}: \mathcal{D}(\C^{d_1}\otimes\C^{d_2})\times \mathbb{N}\mapsto \mathbb{R}$ which measures $n$ copies of $\rho_{1,2}$ using independent measurement such that
for $n\geq 4$,
\begin{align*}
\mathbb{E}[\mathrm{Estimator-independent}(\rho_{1,2},n)]=&||\rho_{1,2}-\rho_1\otimes \rho_2||_2^2,\\
\mathrm{Var}[\mathrm{Estimator-independent}(\rho_{1,2},n)]=&O[\frac{(d_1+1)(d_2+1)||\rho_{1,2}-\rho_1\otimes \rho_2||_2^2}{n}+\frac{(d_1+1)^2(d_2+1)^2}{n^2}]\\
=&O[\frac{d_1d_2||\rho_{1,2}-\rho_1\otimes \rho_2||_2^2}{n}+\frac{d_1^2d_2^2}{n^2}].
\end{align*}
\end{lemma}
Algorithm \ref{algo:Conditional-Independence-Test-Independent} below is based on this estimator.

\begin{algorithm}[H]
	\Input{Access to classical-quantum-quantum states $\rho=\sum_{c}p_c\op{c}{c}\otimes\rho^c_{AB}\in \tau_{ABC}$ and $\epsilon> 0$.}	
    \Output{"Yes" with a probability of at least $\frac{2}{3}$ if $\rho$ is conditional independent; and "No" with a probability of at least $\frac{2}{3}$ if there is no conditionally independent $\sigma$ such that $||\rho-\sigma||_1\leq \epsilon$.}
Choose $L>0$ be a sufficient large constant\;
\If{$n\geq \frac{d_1^6d_2^6}{\epsilon^8}$}{
$m\leftarrow L\frac{d_1^{\frac{6}{7}}d_2^{\frac{6}{7}}n^{\frac{6}{7}}}{\epsilon^{\frac{8}{7}}}$\;
}
\If{$\frac{d_1^{\frac{10}{3}}d_2^{\frac{10}{3}}}{\epsilon^{\frac{8}{3}}}\leq n<\frac{d_1^4d_2^4}{\epsilon^8}$}{
$m\leftarrow L\frac{d_1^{\frac{3}{4}}d_2^{\frac{3}{4}}n^{\frac{7}{8}}}{\epsilon}$\;
}
\If{$0<n< \frac{d_1^{\frac{10}{3}}d_2^{\frac{10}{3}}}{\epsilon^{\frac{8}{3}}}$}{
$m\leftarrow L\frac{\sqrt{n}d_1^2d_2^2}{\epsilon^2}$\;
}
\tcc{Equivalently, $m\leftarrow L\max\{\frac{\sqrt{n}d_1^2d_2^2}{\epsilon^2},\min\{\frac{d_1^{\frac{3}{4}}d_2^{\frac{3}{4}}n^{\frac{7}{8}}}{\epsilon},\frac{d_1^{\frac{6}{7}}d_2^{\frac{6}{7}}n^{\frac{6}{7}}}{\epsilon^{\frac{8}{7}}}\}\}$.}
$\xi\leftarrow\frac{1-\frac{5}{2e}}{2}\min\{\frac{m\epsilon^2}{4d_1d_2},\frac{m^4\epsilon^4}{32d_1^2d_2^2n^3}\}$\;
Set $M$ according to $\mathrm{Poisson}(m)$ distribution\;
Draw $M$ copies of $\rho$, measure the system  $C$ on a computational basis for each copy and measure the $\C^{d_1\times d_1}$, $\C^{d_2\times d_2}$ systems with the measurements $\mathcal{M}_1$ and $\mathcal{M}_2$, respectively. Let $S$ denote the multi-set of measurement outcome\;
\For{all $c\gets 0$ \KwTo $\lceil\log_2 n\rceil$}{
Let $S_c$ be the $|S_c|$ copies of $\rho^c_{AB}$\;
\If{$|S_c|\geq 4$}{
$A_c\leftarrow |S_c|\times \mathrm{Estimator-independent}(\rho^c_{AB},|S_c|)$\;
}
\Else
{
$A_c\leftarrow 0$\;
}
}
    \If{$A=\sum_{c} A_c>\xi$}{
     Returned “NO”\;}
     \Else{
    Return "Yes"\;
    }
	\caption{\textsf{A Conditional Independence Test with Independent Measurement}}
	\label{algo:Conditional-Independence-Test-Independent}
\end{algorithm}
Similar to Proposition \ref{ciev}, we have:
\begin{Prop}
If $\rho$ is conditionally independent, $$\mathbb{E}(A)=0.$$
If $\rho$ is $\epsilon$-far from being independent,
$$\mathbb{E}(A)\geq (1-\frac{5}{2e})\min\{\frac{m\epsilon^2}{4d_1d_2},\frac{m^4\epsilon^4}{32d_1^2d_2^2n^3}\}.$$
In both cases,
$$
\mathrm{Var}(A)=O(d_1^2d_2^2\min\{m,n\}+d_1d_2\mathbb{E}(A)).
$$
\end{Prop}
\begin{proof}
Reuse the proof of \cite{Canonne:2018:TCI:3188745.3188756}, the bound of $\mathbb{E}(A)$ here is exactly the same as the bound of $\mathbb{E}(A)$ in joint measurement setting given in the Appendix.

By the law of total variance,
\begin{align*}
\mathrm{Var} A=\mathbb{E}[\mathrm{Var}(A|a)]+\mathrm{Var}(\mathbb{E}[A|a])
\end{align*}
where $a=(a_c)_{c\in C}$.

According to Lemma \ref{variance-independent}, we know that there exists a constant $t>0$ such that
\begin{align*}
&\mathbb{E}[\mathrm{Var}(A|a)]\\
\leq& \mathbb{E}a_c^2[t(\frac{d_1d_2b_c}{a_c}+\frac{d_1^2d_2^2}{a_c^2})\mathbb{1}_{a_c\geq 4}]\\
=&t(d_1d_2\mathbb{E}(A)+d_1^2d_2^2\mathbb{E}\mathbb{1}_{a_c\geq 4})\\
\leq &t(d_1d_2\mathbb{E}(A)+d_1^2d_2^2\mathbb{E}\mathbb{1}_{a_c\geq 1})\\
\leq &t(d_1d_2\mathbb{E}(A)+d_1^2d_2^2\min\{n,m\})
\end{align*}
where we use the fact that
\begin{align*}
\mathbb{E}\mathbb{1}_{a_c\geq 1}=\sum_{c\in C}[1-e^{-mp_c}]\leq |C|=n,\\
\sum_{c\in C}[1-e^{-mp_c}]\leq \sum_{c\in C}mp_c =m.
\end{align*}
The second term satisfies
\begin{align*}
\mathrm{Var}(\mathbb{E}[A|a])=\sum_{c\in C} b_c^2\mathrm{Var[a_c\mathbb{1}_{a_c\geq 4}]}\leq \sum_{c\in C} 2^2\mathbb{E}[a_c\mathbb{1}_{a_c\geq 4}]\leq 4R\mathbb{E}(A).
\end{align*}
\end{proof}

The chosen of $m_0=\max\{\frac{\sqrt{n}d_1^2d_2^2}{\epsilon^2},\min\{\frac{d_1^{\frac{3}{4}}d_2^{\frac{3}{4}}n^{\frac{7}{8}}}{\epsilon},\frac{d_1^{\frac{6}{7}}d_2^{\frac{6}{7}}n^{\frac{6}{7}}}{\epsilon^{\frac{8}{7}}}\}\}$
actually guarantees that
$$
\min\{\frac{m_0\epsilon^2}{d_1d_2},\frac{m_0^4\epsilon^4}{d_1^2d_2^2n^3}\}\geq d_1d_2\sqrt{\min\{m_0,n\}}.
$$
By setting $m=Lm_0$, we always have
\begin{align*}
&\min\{\frac{m\epsilon^2}{d_1d_2},\frac{m^4\epsilon^4}{d_1^2d_2^2n^3}\}\geq \sqrt{L}\times d_1d_2\sqrt{\min\{m,n\}}\\
\Rightarrow &\xi\geq \frac{(1-\frac{5}{2e})\sqrt{L}}{64}\times d_1d_2\sqrt{\min\{m,n\}},\\
&\mathbb{E}(A)\geq \frac{(1-\frac{5}{2e})\sqrt{L}}{32}\times d_1d_2\sqrt{\min\{m,n\}}.
\end{align*}
If $\rho$ is conditionally independent, then
\begin{align*}
\mathrm{Pr}[A>\xi]\leq\frac{\mathrm{Var}(A)}{\xi^2}= O(d_1d_2\min\{m,n\}{\xi^2})\leq \frac{1}{3}.
\end{align*}
If $\rho$ is $\epsilon$-far from conditionally independent, then
\begin{align*}
\mathrm{Pr}[A<\xi]\leq\mathrm{Pr}[A<\frac{\mathbb{E}(A)}{2}]\leq\mathrm{Pr}[|A-\mathbb{E}(A)|\geq
\frac{\mathbb{E}(A)}{2}] \leq\frac{\mathrm{Var}(A)}{4\mathbb{E}^2(A)}= O(\frac{d_1^2d_2^2\min\{m,n\}}{\mathbb{E}^2(A)}+\frac{d_1d_2}{\mathbb{E}(A)})\leq \frac{1}{3}.
\end{align*}
Therefore, the following is valid:
\paragraph{Independent measurement: partial restatement of Theorem \ref{cid}}
$O(\max\{\frac{\sqrt{n}d_1^2d_2^2}{\epsilon^2},\min\{\frac{d_1^{\frac{3}{4}}d_2^{\frac{3}{4}}n^{\frac{7}{8}}}{\epsilon},\frac{d_1^{\frac{6}{7}}d_2^{\frac{6}{7}}n^{\frac{6}{7}}}{\epsilon^{\frac{8}{7}}}\}\})$ copies are sufficient to test whether $\rho_{ABC}$ is conditionally independent or $\epsilon$-far from $\mathcal{P}_{A,B|C}$ with a classical-quantum-quantum state $\rho_{ABC}\in\tau_{ABC}$ where $|C|=n$ using independent measurement.
\section{Discussion and Acknowledgments}
There are many interesting open problems. Developing lower bound techniques for independent measurement and local measurement is of great interest. As we mentioned, there is no much lower bound technique for quantum state property testing, even information-theoretical lower bounds for joint measurement. Also, techniques for classical property testing are not applicable in the independent measurement and local measurement setting because of the rich structure of independent (local) measurement.
 It would also be interesting to consider the problems of this paper for the low-rank quantum states, where some techniques were developed in \cite{kueng2016distinguishing}. 

We thank Steve Flammia for his insightful discussion about the $k$-local tomography.
We thank Youming Qiao for his helpful comments on the previous version of this manuscript. We thank Tongyang Li for pointing out relevant reference \cite{GL19}. We thank Ryan O'Donnell and John Wright for telling us the Sanov's theorem and its relation to \cite{HHJ+16}.
This work was supported by DE180100156.
\bibliographystyle{alpha}
\bibliography{opt-tomo}
\section{Appendix}
The following is an analysis of the conditional independence testing algorithm for probability distributions given in \cite{Canonne:2018:TCI:3188745.3188756}, which is relevant to the correctness our Algorithm \ref{algo:Conditional-Independence-Test-Joint}.

To show the correctness of this algorithm, the following lemma is used.
\begin{lemma}\cite{Canonne:2018:TCI:3188745.3188756}\label{Poisson}
There exists an absolute constant $R> 0$ such that, for any $\lambda>0$ and $N$ distributed according to $\mathrm{Poisson}(\lambda)$,
\begin{align*}
\mathrm{Var}[N\mathbb{1}_{N\geq 4}]\leq R\mathbb{E}[N\mathbb{1}_{N\geq 4}].
\end{align*}
\end{lemma}

An analysis of $\mathbb{E}(A)$ and $\mathrm{Var}(A)$ as follows. Let $a_c=|S_c|$, $b_c=||\rho^c_{AB}-\rho^c_A\otimes\rho^c_B||_2^2$.
From Lemma \ref{joint-testing} and direct observation,
\begin{align*}
\mathbb{E}[A|a_c]=a_cb_c\mathbb{1}_{\{a_c\geq 4\}}.
\end{align*}
According to the Poissonization technique, we know that $a_c$s are independent, and the distribution of $a_c$ is governed by $\mathrm{Poisson}(mp_c)$.

Therefore,
\begin{align*}
\mathbb{E}[A]=\sum_{c\in C}b_c\mathbb{E}[a_c\mathbb{1}_{\{a_c\geq 4\}}]=\sum_{c\in C} b_c\sum_{k\geq 4} k\frac{(mp_c)^k}{k!}=\sum_{c\in C} b_c f(mp_c),
\end{align*}
where
\begin{align*}
f(x)=x-e^{-x}\sum_{k=0}^2 \frac{x^{k+1}}{k!}=e^{-x}\sum_{k=3}^{\infty} \frac{x^{k+1}}{k!}
\end{align*}
Note that $x\in\mathbb{R}^{+}$, and we always have
\begin{align*}
f(x)\geq \gamma\min\{x,x^4\},
\end{align*}
where $\gamma=f(1)=1-\frac{5}{2e}$.

To see this, if $x\geq 1$,
\begin{align*}
\{\frac{f(x)}{x}\}'=e^{-x}\frac{x^2}{2}\geq 0.
\end{align*}
If $x\leq 1$,
\begin{align*}
\{\frac{f(x)}{x^4}\}'=\{e^{-x}\sum_{k=0}^{\infty}\frac{x^k}{(k+3)!}\}'=e^{-x}\sum_{k=0}^{\infty}[\frac{k+1}{(k+4)!}-\frac{1}{(k+3)!}]x^k=-3e^{-x}\sum_{k=0}^{\infty}\frac{x^k}{(k+4)!}\leq 0.
\end{align*}

That is,
\begin{align*}
\mathbb{E}[A]\geq \sum_{c\in C}\gamma\min\{mp_c,m^4p_c^4\}b_c.
\end{align*}

In cases where $\rho_{ABC}$ is conditionally independent, then $b_c=\tr||\rho^c_{AB}-\rho^c_A\otimes\rho^c_B||_2^2=0$,
\begin{align*}
\mathbb{E}[A]=0.
\end{align*}

In cases where $\rho_{ABC}=\sum_{c\in C}p_c\op{c}{c}\otimes \rho^c_{AB}$ is $\epsilon$-far from $\mathcal{P}_{A,B|C}$, let $\tilde{\rho}_{ABC}=\sum_{c\in C}p_c\op{c}{c}\otimes \rho^c_{A}\otimes \rho^c_{B}$, and we have
\begin{align*}
&||\rho_{ABC}-\tilde{\rho}_{ABC}||_1> \epsilon\\
\Leftrightarrow &\sum_{c\in C}p_c ||\rho^c_{AB}-\rho^c_{A}\otimes \rho^c_{B}||_1> \epsilon\\
\Rightarrow &\sum_{c\in C}p_c ||\rho^c_{AB}-\rho^c_{A}\otimes \rho^c_{B}||_2> \frac{\epsilon}{\sqrt{d_1d_2}}\\
\mathrm{that~is,}~&\sum_{c\in C}mp_c \sqrt{b_c}> \frac{m\epsilon}{\sqrt{d_1d_2}}\\
\Rightarrow &\sum_{c\in C,mp_c\geq 1}mp_c \sqrt{b_c}> \frac{m\epsilon}{2\sqrt{d_1d_2}} \\
\mathrm{or}~&\sum_{c\in C,mp_c< 1}mp_c \sqrt{b_c}> \frac{m\epsilon}{2\sqrt{d_1d_2}}\\
\Rightarrow &\sum_{c\in C,mp_c\geq 1}mp_c b_c \geq\frac{(\sum_{c\in C,mp_c\geq 1}mp_c \sqrt{b_c})^2}{\sum_{c\in C,mp_c\geq 1} mp_c}\geq\frac{(\frac{m\epsilon}{2\sqrt{d_1d_2}})^2}{m}>\frac{m\epsilon^2}{4d_1d_2}\\
\mathrm{or}~&\sum_{c\in C,mp_c< 1}(mp_c)^4 b_c\geq \frac{(\sum_{c\in C,mp_c< 1}mp_c b_c^{\frac{1}{4}}b_c^{\frac{3}{4}\times\frac{1}{3}})^4}{(\sum_{c\in C,mp_c< 1}b_c^{\frac{1}{3}})^3}\geq \frac{(\sum_{c\in C,mp_c< 1}mp_c \sqrt{b_c})^4}{2n^3}>\frac{m^4\epsilon^4}{32d_1^2d_2^2n^3}.\\
\Rightarrow&\mathbb{E}[A]> \gamma\min\{\frac{m\epsilon^2}{4d_1d_2},\frac{m^4\epsilon^4}{32d_1^2d_2^2n^3}\}.
\end{align*}
The second last inequalities accord with the Cauchy–Schwarz inequality, the H\"older inequality and $b_c\leq 2$.

By the law of total variance, the bound of the variance of $A$ is
\begin{align*}
\mathrm{Var} A=\mathbb{E}[\mathrm{Var}(A|a)]+\mathrm{Var}(\mathbb{E}[A|a])
\end{align*}
where $a=(a_c)_{c\in C}$.

According to Lemma \ref{joint-testing}, we know that a constant $t>0$ exists such that
\begin{align*}
\mathbb{E}[\mathrm{Var}(A|a)]\leq \mathbb{E}a_c^2[t(\frac{b_c}{a_c}+\frac{1}{a_c^2})\mathbb{1}_{a_c\geq 4}]=t(\mathbb{E}(A)+\mathbb{E}\mathbb{1}_{a_c\geq 4})\leq t(\mathbb{E}(A)+\mathbb{E}\mathbb{1}_{a_c\geq 1})\leq t(\mathbb{E}(A)+\min\{n,m\})
\end{align*}
where we use the fact that
\begin{align*}
\mathbb{E}\mathbb{1}_{a_c\geq 1}=\sum_{c\in C}[1-e^{-mp_c}]\leq |C|=n,\\
\sum_{c\in C}[1-e^{-mp_c}]\leq \sum_{c\in C}mp_c =m.
\end{align*}
The second term can be bounded as
\begin{align*}
\mathrm{Var}(\mathbb{E}[A|a])=\sum_{c\in C} b_c^2\mathrm{Var[a_c\mathbb{1}_{a_c\geq 4}]}\leq \sum_{c\in C} 2^2\mathbb{E}[a_c\mathbb{1}_{a_c\geq 4}]\leq 4R\mathbb{E}(A).
\end{align*}
\end{document}